\documentclass[reqno,12pt,a4paper]{article}

\usepackage{amsfonts,amssymb,amsthm, 
amsmath,amscd, bbm, bm, mathabx,
mathrsfs,delarray,subfigure,framed,graphicx}

\usepackage[dvipsnames]{xcolor}[svgnames,x11names]

\usepackage{hyperref,enumitem}
\usepackage{cite}
\usepackage{xypic}
\usepackage{pstricks}

\definecolor{light-gray}{gray}{0.96}
\definecolor{shadecolor}{named}{light-gray}

\usepackage[scr=rsfs,cal=dutchcal]{mathalfa}

\usepackage{qcircuit}

\usepackage{scalerel}

\usepackage{fancyhdr}
\fancyheadoffset[RE,LO]{0.\textwidth} 

\hypersetup{
  colorlinks,
  linkcolor=blue,
  linktoc=all,
  citecolor=blue,
   urlcolor=blue
}

\usepackage{sectsty}

\colorlet{sectitlecolor}{blue}
\colorlet{sectboxcolor}{white}
\colorlet{secnumcolor}{blue}

\sectionfont{\color{sectitlecolor}}

\makeatletter
\renewcommand\@seccntformat[1]{%
  \colorbox{sectboxcolor}{\textcolor{secnumcolor}{\csname the#1\endcsname}}%
  \quad
}
\makeatother

\usepackage{etoolbox}
\patchcmd{\thebibliography}{\section*{\refname}}{}{}{}

\DeclareSymbolFont{sfletters}{OML}{cmbrm}{m}{it}  

\DeclareMathSymbol{\sfGamma}{\mathord}{sfletters}{"00}
\DeclareMathSymbol{\sfDelta}{\mathord}{sfletters}{"01}
\DeclareMathSymbol{\sfTheta}{\mathord}{sfletters}{"02}
\DeclareMathSymbol{\sfLambda}{\mathord}{sfletters}{"03}
\DeclareMathSymbol{\sfXi}{\mathord}{sfletters}{"04}
\DeclareMathSymbol{\sfPi}{\mathord}{sfletters}{"05}
\DeclareMathSymbol{\sfSigma}{\mathord}{sfletters}{"06}
\DeclareMathSymbol{\sfUpsilon}{\mathord}{sfletters}{"07}
\DeclareMathSymbol{\sfPhi}{\mathord}{sfletters}{"08}
\DeclareMathSymbol{\sfPsi}{\mathord}{sfletters}{"09}
\DeclareMathSymbol{\sfOmega}{\mathord}{sfletters}{"0A}
\DeclareMathSymbol{\sfalpha}{\mathord}{sfletters}{"0B}
\DeclareMathSymbol{\sfbeta}{\mathord}{sfletters}{"0C}
\DeclareMathSymbol{\sfgamma}{\mathord}{sfletters}{"0D}
\DeclareMathSymbol{\sfdelta}{\mathord}{sfletters}{"0E}
\DeclareMathSymbol{\sfepsilon}{\mathord}{sfletters}{"0F}
\DeclareMathSymbol{\sfzeta}{\mathord}{sfletters}{"10}
\DeclareMathSymbol{\sfeta}{\mathord}{sfletters}{"11}
\DeclareMathSymbol{\sftheta}{\mathord}{sfletters}{"12}
\DeclareMathSymbol{\sfiota}{\mathord}{sfletters}{"13}
\DeclareMathSymbol{\sfkappa}{\mathord}{sfletters}{"14}
\DeclareMathSymbol{\sflambda}{\mathord}{sfletters}{"15}
\DeclareMathSymbol{\sfmu}{\mathord}{sfletters}{"16}
\DeclareMathSymbol{\sfnu}{\mathord}{sfletters}{"17}
\DeclareMathSymbol{\sfxi}{\mathord}{sfletters}{"18}
\DeclareMathSymbol{\sfpi}{\mathord}{sfletters}{"19}
\DeclareMathSymbol{\sfrho}{\mathord}{sfletters}{"1A}
\DeclareMathSymbol{\sfsigma}{\mathord}{sfletters}{"1B}
\DeclareMathSymbol{\sftau}{\mathord}{sfletters}{"1C}
\DeclareMathSymbol{\sfupsilon}{\mathord}{sfletters}{"1D}
\DeclareMathSymbol{\sfphi}{\mathord}{sfletters}{"1E}
\DeclareMathSymbol{\sfchi}{\mathord}{sfletters}{"1F}
\DeclareMathSymbol{\sfpsi}{\mathord}{sfletters}{"20}
\DeclareMathSymbol{\sfomega}{\mathord}{sfletters}{"21}
\DeclareMathSymbol{\sfvarepsilon}{\mathord}{sfletters}{"22}
\DeclareMathSymbol{\sfvartheta}{\mathord}{sfletters}{"23}
\DeclareMathSymbol{\sfvarpi}{\mathord}{sfletters}{"24}
\DeclareMathSymbol{\sfvarrho}{\mathord}{sfletters}{"25}
\DeclareMathSymbol{\sfvarsigma}{\mathord}{sfletters}{"26}
\DeclareMathSymbol{\sfvarphi}{\mathord}{sfletters}{"27}

\DeclareMathSymbol{\spartial}{\mathord}{sfletters}{"40}

\DeclareMathSymbol{\sfA}{\mathord}{sfletters}{"41}
\DeclareMathSymbol{\sfB}{\mathord}{sfletters}{"42}
\DeclareMathSymbol{\sfC}{\mathord}{sfletters}{"43}
\DeclareMathSymbol{\sfD}{\mathord}{sfletters}{"44}
\DeclareMathSymbol{\sfE}{\mathord}{sfletters}{"45}
\DeclareMathSymbol{\sfF}{\mathord}{sfletters}{"46}
\DeclareMathSymbol{\sfG}{\mathord}{sfletters}{"47}
\DeclareMathSymbol{\sfH}{\mathord}{sfletters}{"48}
\DeclareMathSymbol{\sfI}{\mathord}{sfletters}{"49}
\DeclareMathSymbol{\sfJ}{\mathord}{sfletters}{"4A}
\DeclareMathSymbol{\sfK}{\mathord}{sfletters}{"4B}
\DeclareMathSymbol{\sfL}{\mathord}{sfletters}{"4C}
\DeclareMathSymbol{\sfM}{\mathord}{sfletters}{"4D}
\DeclareMathSymbol{\sfN}{\mathord}{sfletters}{"4E}
\DeclareMathSymbol{\sfO}{\mathord}{sfletters}{"4F}
\DeclareMathSymbol{\sfP}{\mathord}{sfletters}{"50}
\DeclareMathSymbol{\sfQ}{\mathord}{sfletters}{"51}
\DeclareMathSymbol{\sfR}{\mathord}{sfletters}{"52}
\DeclareMathSymbol{\sfS}{\mathord}{sfletters}{"53}
\DeclareMathSymbol{\sfT}{\mathord}{sfletters}{"54}
\DeclareMathSymbol{\sfU}{\mathord}{sfletters}{"55}
\DeclareMathSymbol{\sfV}{\mathord}{sfletters}{"56}
\DeclareMathSymbol{\sfW}{\mathord}{sfletters}{"57}
\DeclareMathSymbol{\sfX}{\mathord}{sfletters}{"58}
\DeclareMathSymbol{\sfY}{\mathord}{sfletters}{"59}
\DeclareMathSymbol{\sfZ}{\mathord}{sfletters}{"5A}
\DeclareMathSymbol{\sfa}{\mathord}{sfletters}{"61}
\DeclareMathSymbol{\sfb}{\mathord}{sfletters}{"62}
\DeclareMathSymbol{\sfc}{\mathord}{sfletters}{"63}
\DeclareMathSymbol{\sfd}{\mathord}{sfletters}{"64}
\DeclareMathSymbol{\sfe}{\mathord}{sfletters}{"65}
\DeclareMathSymbol{\sff}{\mathord}{sfletters}{"66}
\DeclareMathSymbol{\sfg}{\mathord}{sfletters}{"67}
\DeclareMathSymbol{\sfh}{\mathord}{sfletters}{"68}
\DeclareMathSymbol{\sfi}{\mathord}{sfletters}{"69}
\DeclareMathSymbol{\sfj}{\mathord}{sfletters}{"6A}
\DeclareMathSymbol{\sfk}{\mathord}{sfletters}{"6B}
\DeclareMathSymbol{\sfl}{\mathord}{sfletters}{"6C}
\DeclareMathSymbol{\sfm}{\mathord}{sfletters}{"6D}
\DeclareMathSymbol{\sfn}{\mathord}{sfletters}{"6E}
\DeclareMathSymbol{\sfo}{\mathord}{sfletters}{"6F}
\DeclareMathSymbol{\sfp}{\mathord}{sfletters}{"70}
\DeclareMathSymbol{\sfq}{\mathord}{sfletters}{"71}
\DeclareMathSymbol{\sfr}{\mathord}{sfletters}{"72}
\DeclareMathSymbol{\sfs}{\mathord}{sfletters}{"73}
\DeclareMathSymbol{\sft}{\mathord}{sfletters}{"74}
\DeclareMathSymbol{\sfu}{\mathord}{sfletters}{"75}
\DeclareMathSymbol{\sfv}{\mathord}{sfletters}{"76}
\DeclareMathSymbol{\sfw}{\mathord}{sfletters}{"77}
\DeclareMathSymbol{\sfx}{\mathord}{sfletters}{"78}
\DeclareMathSymbol{\sfy}{\mathord}{sfletters}{"79}
\DeclareMathSymbol{\sfz}{\mathord}{sfletters}{"7A}

\usepackage[LGR,T1]{fontenc}
\usepackage{amsmath,etoolbox}

\newcommand{\declarebsfgreek}[2]{%
  \protected\csdef{bsf#1}{\mathord{\text{\bsfgreekfont#2}}}%
}
\newcommand{\bsfgreekfont}{\usefont{LGR}{cmss}{bx}{it}}

\declarebsfgreek{alpha}{a}
\declarebsfgreek{beta}{b}
\declarebsfgreek{gamma}{g}
\declarebsfgreek{delta}{d}
\declarebsfgreek{epsilon}{e}
\declarebsfgreek{zeta}{z}
\declarebsfgreek{eta}{h}
\declarebsfgreek{theta}{j}
\declarebsfgreek{iota}{i}
\declarebsfgreek{kappa}{k}
\declarebsfgreek{lambda}{l}
\declarebsfgreek{mu}{m}
\declarebsfgreek{nu}{n}
\declarebsfgreek{xi}{x}
\declarebsfgreek{omicron}{o}
\declarebsfgreek{pi}{p}
\declarebsfgreek{rho}{r}
\declarebsfgreek{sigma}{s}
\declarebsfgreek{tau}{t}
\declarebsfgreek{upsilon}{u}
\declarebsfgreek{phi}{f}
\declarebsfgreek{chi}{q}
\declarebsfgreek{psi}{y}
\declarebsfgreek{omega}{w}
\declarebsfgreek{varsigma}{c}

\declarebsfgreek{Gamma}{G}
\declarebsfgreek{Delta}{D}
\declarebsfgreek{Upsilon}{U}
\declarebsfgreek{Omega}{W}
\declarebsfgreek{Theta}{J}
\declarebsfgreek{Lambda}{L}
\declarebsfgreek{Xi}{X}
\declarebsfgreek{Pi}{P}
\declarebsfgreek{Sigma}{S}
\declarebsfgreek{Phi}{F}
\declarebsfgreek{Psi}{Y}

\newcommand{\declarebsfitalic}[2]{%
  \protected\csdef{bsf#1}{\mathord{\text{\bsfitalicfont#2}}}%
}
\newcommand{\bsfitalicfont}{\usefont{T1}{cmss}{bx}{it}}

\declarebsfitalic{a}{a}
\declarebsfitalic{b}{b}
\declarebsfitalic{c}{c}
\declarebsfitalic{d}{d}
\declarebsfitalic{e}{e}
\declarebsfitalic{f}{f}
\declarebsfitalic{g}{g}
\declarebsfitalic{h}{h}
\declarebsfitalic{i}{i}
\declarebsfitalic{j}{j}
\declarebsfitalic{k}{k}
\declarebsfitalic{l}{l}
\declarebsfitalic{m}{m}
\declarebsfitalic{n}{n}
\declarebsfitalic{o}{o}
\declarebsfitalic{p}{p}
\declarebsfitalic{q}{q}
\declarebsfitalic{r}{r}
\declarebsfitalic{s}{s}
\declarebsfitalic{t}{t}
\declarebsfitalic{u}{u}
\declarebsfitalic{v}{v}
\declarebsfitalic{x}{x}
\declarebsfitalic{y}{y}
\declarebsfitalic{z}{z}

\declarebsfitalic{A}{A}
\declarebsfitalic{B}{B}
\declarebsfitalic{C}{C}
\declarebsfitalic{D}{D}
\declarebsfitalic{E}{E}
\declarebsfitalic{F}{F}
\declarebsfitalic{G}{G}
\declarebsfitalic{H}{H}
\declarebsfitalic{I}{I}
\declarebsfitalic{J}{J}
\declarebsfitalic{K}{K}
\declarebsfitalic{L}{L}
\declarebsfitalic{M}{M}
\declarebsfitalic{N}{N}
\declarebsfitalic{O}{O}
\declarebsfitalic{P}{P}
\declarebsfitalic{Q}{Q}
\declarebsfitalic{R}{R}
\declarebsfitalic{S}{S}
\declarebsfitalic{T}{T}
\declarebsfitalic{U}{U}
\declarebsfitalic{V}{V}
\declarebsfitalic{X}{X}
\declarebsfitalic{Y}{Y}
\declarebsfitalic{Z}{Z}

\newcommand{\bfsfH}{\sfH\!\!\!\!\!\sfH}

\newcommand{\msA}{{\sf{A}}}
\newcommand{\msB}{{\sf{B}}}
\newcommand{\msC}{{\sf{C}}}
\newcommand{\msD}{{\sf{D}}}

\newcommand{\msF}{{\sf{F}}}

\newcommand{\msI}{{\sf{I}}}

\newcommand{\msK}{{\sf{K}}}

\newcommand{\msM}{{\sf{M}}}

\newcommand{\msP}{{\sf{P}}}

\newcommand{\msR}{{\sf{R}}}
\newcommand{\msS}{{\sf{S}}}
\newcommand{\msT}{{\sf{T}}}
\newcommand{\msU}{{\sf{U}}}
\newcommand{\msV}{{\sf{V}}}

\newcommand{\msY}{{\sf{Y}}}
\newcommand{\msZ}{{\sf{Z}}}

\newcommand{\bbF}{{\mathbb{F}}}

\newcommand{\bbH}{{\mathbb{H}}}

\newcommand{\bbN}{{\mathbb{N}}}

\newcommand{\bbT}{{\mathbb{T}}}

\newcommand{\bbZ}{{\mathbb{Z}}}

\newcommand{\scC}{{\matheul{C}}}

\newcommand{\scH}{{\matheul{H}}}

\newcommand{\scW}{{\matheul{W}}}

\newcommand{\clE}{{\mathcal{E}}}

\newcommand{\clG}{{\mathcal{G}}}

\newcommand{\clU}{{\mathcal{U}}}

\newcommand{\bra}[1]{{\langle{#1}|}}

\newcommand{\ket}[1]{{|{#1}\rangle}}

\newcommand{\braket}[2]{{\langle{#1}|{#2}\rangle}}

\newcommand{\exval}[3]{{\langle{#1}|{#2}|{#3}\rangle}}

\usepackage{accents}

\sectionfont{\fontsize{12}{15}\selectfont}
\subsectionfont{\fontsize{12}{15}\selectfont}

\usepackage{enumitem}


\DeclareMathOperator{\id}{id}

\DeclareMathOperator{\Hom}{Hom}

\DeclareMathOperator{\Obj}{Obj}

\DeclareMathOperator{\pro}{pr}

\DeclareMathOperator{\Vect}{Vect}
\DeclareMathOperator{\Mod}{Mod}

\DeclareMathOperator{\Gal}{Gal}

\DeclareMathOperator{\End}{End}

\DeclareMathOperator{\Aut}{Aut}

\DeclareMathOperator{\tr}{tr}
\DeclareMathOperator{\Tr}{Tr}

\DeclareMathOperator{\ee}{e}

\DeclareMathOperator{\GR}{GR}
\DeclareMathOperator{\GF}{GF}

\DeclareMathOperator{\rst}{rem}

\DeclareMathOperator{\supp}{supp}

\numberwithin{equation}{subsection} 
\numberwithin{subsection}{section} 

\newcommand{\ceqref}[1]{{\textcolor{blue}{\eqref{#1}}}}
\newcommand{\cref}[1]{{\textcolor{blue}{\ref{#1}}}}
\newcommand{\ccite}[1]{{\textcolor{blue}{\hspace{-.40pt}\cite{#1}}}}

\newcommand{\sss}{{\hbox{\large $\sum$}}}
\newcommand{\ppp}{{\hbox{\large $\prod$}}}
\newcommand{\ooo}{{\hbox{\large $\bigotimes$}}}
\newcommand{\uuu}{{\hbox{\large $\bigcup$}}}

\newcommand{\ddd}{{\hbox{\large $\bigoplus$}}}

\newcommand{\ul}[1]{{\underline{#1}}}
\newcommand{\ol}[1]{{\overline{#1}}}

\newcommand{\hfpt}{\hspace{.75pt}}
\newcommand{\mhfpt}{\hspace{-.75pt}}

\newcommand\mycom[2]{\genfrac{}{}{-3pt}{}{#1}{#2}}

\font\euler=eusm10 at 12.8 truept
\font\scripteuler=eusm7
\font\scriptscripteuler=eusm5 
\def\eul{\fam=12}
\newcommand{\matheul}[1]{{{\eul #1}}}

\newtheorem{defi}{{\sf Definition}}[subsection]
\newtheorem{prop}{{\sf Proposition}}[subsection]
\newtheorem{theor}{{\sf Theorem}}[subsection]
\newtheorem{lemma}{{\sf Lemma}}[subsection]
\newtheorem{exa}{{\sf Example}}[subsection]

\newtheorem{cor}{{\sf Corollary}}[subsection]

\DeclareMathSymbol{*}{\mathbin}{symbols}{"03} 

\makeatletter
\renewcommand{\eqref}[1]{\tagform@{\ref{#1}}}
\def\maketag@@@#1{\hbox{#1}}
\makeatother

\DeclareMathVersion{normal2}

\begin{document}

\thispagestyle{empty} 

\vskip1.5cm
\begin{large} 
{\flushleft\textcolor{blue}{\sffamily\bfseries Calibrated hypergraph states:}}
{\flushleft\textcolor{blue}{\sffamily\bfseries II calibrated hypergraph state construction and applications}} 
\end{large}
\vskip1.2cm
\hrule height 1.5pt
\vskip1.2cm
{\flushleft{\sffamily \bfseries Roberto Zucchini}\\
\it Department of Physics and Astronomy,\\
University of Bologna,\\
I.N.F.N., Bologna division,\\
viale Berti Pichat, 6/2\\
Bologna, Italy\\
Email: \textcolor{blue}{\tt \href{mailto:roberto.zucchini@unibo.it}{roberto.zucchini@unibo.it}}, 
\textcolor{blue}{\tt \href{mailto:zucchinir@bo.infn.it}{zucchinir@bo.infn.it}}}


\vskip.9cm 
{\flushleft\sc
Abstract:} 
Hypergraph states are a special kind of multipartite states encoded by
hypergraphs relevant in quantum error correction, measurement--based quantum computation,
quantum non locality and entanglement.
In a series of two papers, we introduce and investigate calibrated hypergraph states,
an extension of weighted hypergraph states codified by 
hypergraphs equipped with calibrations, a broad generalization of weightings.
The guiding principle informing our approach is that a constructive theory of hypergraph
states must be based on a categorical framework for both 
hypergraphs and multi qudit states  constraining hypergraph
states enough to render the determination of their general structure possible. 
In this second paper, we build upon the graded $\varOmega$ monadic framework worked out in the
companion paper, focusing on qudits over a generic Galois ring. 
We explicitly construct a calibrated hypergraph state map as a special morphism of the
calibrated hypergraph and multi qudit state $\varOmega$ monads.
We further prove that the calibrated hypergraph states so yielded are locally maximally entangleable
stabilizer states, elucidate their relationship to weighted hypergraph states, show that
they reduce to the weighted ones in the familiar qubit
case and prove through examples that this is no longer the case for higher qudits. 
\vspace{2mm}
\par\noindent
MSC: 05C65 81P99 81Q99 

\vfill\eject

{\color{blue}\tableofcontents}

\vfill\eject

\renewcommand{\sectionmark}[1]{\markright{\thesection\ ~~#1}}

\section{\textcolor{blue}{\sffamily Introduction}}\label{sec:intro}

Graphs and hypergraphs are discrete topological structures which find many applications
in the modelling of the binary and higher-order relations of the 
objects and in the analysis of the data sets occurring in a variety scientific disciplines
\ccite{Berge:1973gth,Ouvrard:2020hir}.
Graphs and hypergraphs both feature vertices, but while graphs allow only 
edges connecting pairs of distinct vertices hypergraphs admit hyperedges
joining any number of vertices.
Graphs and hypergraphs can further be enriched with edge and hyperedge
weight data yielding multigraphs and weighted hypergraphs, respectively. 

In quantum information and computation theory, graphs and hypergraphs are employed to tackle a number
of fundamental problems as reviewed in the next subsection.


\subsection{\textcolor{blue}{\sffamily Graph and hypergraph states}}\label{subsec:ghgstrev}

Graph states are a distinguished kind of multipartite states encoded by graphs. They
were originally introduced by Schlingemann and Hein, Eisert and Briegel 
in refs. \!\!\ccite{Schlingemann:2003cag,Hein:2003meg} over twenty years ago.
Since then they have emerged as key elements 
of the analysis of quantum error correction \ccite{Schlingemann:2000ecg,Bausch:2019tpn},
measurement-based quantum computation \ccite{Raussendorf:2001oqc},
quantum secret sharing \ccite{Markham:2008gss}
and Bell non locality \ccite{Scarani:2005ncs,Guehne:2004big,Baccari:2020bst}
and entanglement \ccite{Kruszynska:2006epp,Toth:2005dsf,Jungnitsch:2011ewg}. 
Ref. \!\!\ccite{Hein:2006ega} offers a comprehensive introduction to the subject. 

Hypergraph states are a wider class of multipartite states encoded by hypergraphs. 
They were introduced and studied more than a decade ago by Qu {\it et al.} and Rossi {\it et al.}
in refs. \!\!\ccite{Qu:2012eqs,Rossi:2012qhs} as a broad extension of graph states.
As these latter, they have found significant applications to the study of quantum error correction \ccite{Wagner:2017shs},
measurement-based quantum computation \ccite{Gachechiladze:2019mbh,Takeuchi:2018qcu},
state verification and self-testing \ccite{Morimae:2017vmq,Zhu:2018evh} and 
Bell non locality \ccite{Gachechiladze:2015evr}   
and entanglement \ccite{Guehne:2014nch,Ghio:2017med}.  
Background material on hypergraph states is available in ref. \!\!\ccite{Gachechiladze:2019phd}.

Qudits are higher level computational units generalizing familiar 2--level qubits. 
Their larger state space can be employed to process and store more information and allows in principle
for a diminution of circuit complexity together with an improvement of algorithmic efficiency.
They were first introduced and investigated in refs. \!\!\ccite{Ashikhmin:2000nbc,Gheorghiu:2011qsg}. Ref.
\!\!\ccite{Wang:2020qhd} reviews qudits and their use in higher dimensional quantum computing. 

Qudit graph states have been examined in a number of works including \ccite{Helwig:2013amq,Keet:2010qss}. 
Qudit hypergraph states were studied first in refs. \!\!\ccite{Steinhoff:2016:qhs,Xiong:2017qhp}. 
They were also considered in relation to quantum error correction in ref. \!\!\ccite{Looi:2007ecq}.

Graph and hypergraph states, in spite of their broad variety
reviewed above, have in common certain basic features. They are stabilizer states, 
i.e. one-dimensional error correcting codes \ccite{Gottesman:1997scq,Garcia:2014gss}. 
Further, they are locally maximally entangleable states \ccite{Kruszynska:2008lem}. 

A key problem in quantum information theory is the classification of multipartite states 
in entanglement local equivalence classes. Local unitary/local Clifford, separable and stochastic local
operations with classical communications equivalence constitute the main classing schemes.
Graph and hypergraph states are particularly suitable for this kind of investigation
because of the richness of their mathematical properties. This matter has been investigated in 
refs. \ccite{Nest:2003lct,Nest:2004lce,Nest:2004uce,Bravyi:2005ghz} for graph states
and \ccite{Lyons:2014luh} for hypergraph states and has been considered also for  qudit 
states \ccite{Hostens:2004qma,Bahramgiri:2006cnb}.


\subsection{\textcolor{blue}{\sffamily Plan of the endeavour}}\label{subsec:wsproject}

With the present endeavour, we introduce and study calibrated hypergraph states,
an ample extension of weighted hypergraph states codified by 
hypergraphs equipped with calibrations construed as generalization of weightings.
The guiding principle which we have followed is that a viable theory
of hypergraph states should refer to a suitable categorical framework for
hypergraphs on one hand and multi qudit states on the other
and within this impose requirements on hypergraph
states sufficient to make the determination of their general structure possible.
It results in a rather general framework covering most of the graph state constructions
mentioned in the previous subsection. Our approach parallels that of
Ionicioiu and Spiller in ref. \!\!\ccite{Ionicioiu:2012egq} and we are indebted to these authors for inspiration.
However, unlike theirs, our theory cannot be properly described as axiomatic but rather as constructive. 

The endeavour is naturally divided in two parts, hereafter mentioned as I and II,
which we outline next. 

In I, we introduce graded $\varOmega$ monads, concrete Pro
categories isomorphic to the Pro category $\varOmega$ of finite von Neumann ordinals and equipped with an associative and unital
graded multiplication, and their morphisms, maps of $\varOmega$ monads compatible with their monadic structure. 
We then show that both calibrated hypergraphs and multi qudit states 
naturally organize in graded $\varOmega$ monads. In this way, we pave the way to the construction of the calibrated
hypergraph state map as a special morphism of these $\varOmega$ monads carried out in II. 

In II, relying on the graded $\varOmega$ monadic set--up worked out in I and concentra\-ting on 
qudits over a generic Galois ring, 
we explicitly construct a calibrated hypergraph state map as a special morphism of the
calibrated hypergraph and multi qudit state $\varOmega$ monads.
We further prove that calibrated hypergraph states are locally maximally entangleable
stabilizer states, clarify their relationship to weighted hypergraph states, demonstrate that
they reduce to the weighted ones in the familiar qubit
case and prove through examples that this is no longer the case for higher qudits. 
We finally illustrate a number of special technical results and applications. 

The present paper covers Part II. Its contents are outlined in more detail in subsect. \cref{subsec:design}
below.


\subsection{\textcolor{blue}{\sffamily Review of the graded monadic framework of hypergraph states}}\label{subsec:plan}

In this subsection, we review the $\varOmega$ monadic approach to calibrated hypergraphs and multi qudit states
worked out in I, which constitutes the categorical framework of
the calibrated hypergraph state construction. 
Our discussion will be merely informal. The reader interested in 
the precise definitions of the concepts presented and rigorous statements
of the results obtained is referred to the paper I. 

In I, we stated and justified the principle inspiring our work:
hypergraphs and their data on one hand and multipartite quantum states on the other are naturally 
described by Pro categories isomorphic to the Pro category of finite von Neumann ordinals and
equipped with a graded monadic multiplicative structure. 
Refs. \!\!\ccite{MacLane:1978cwm,Fong:2019act}
and app. A of I provide background on category theory.

A Pro category \,\ccite{Maclane:1965cta,Boardman:1968hes,May:1972ils,Markl:2002otp}
is a strict monoidal category all of whose objects are monoidal powers of a single generating object. 
The category of finite von Neumann ordinals $\varOmega$ is the prototypical example of Pro category.
It is a concrete category featuring the ordinal sets $[l]$ and functions $f:[l]\rightarrow[m]$ as objects and morphisms,
respectively, and equipped with a monoidal product formalizing the set theoretic operation of disjoint union of ordinal
sets and functions through the usual ordinal sum.

The main categories studied in I, such as the various hypergraph categories
and the multi dit mode categories, are instances of $\varOmega$ categories. 
An $\varOmega$ category is a Pro category $D\varOmega$ together with a strict monoidal isofunctor
$D:\varOmega\rightarrow D\varOmega$. 
Though from a purely category theoretic perspective an $\varOmega$ category $D\varOmega$ is indistinguishable from
$\varOmega$ itself, the explicit realizations the objects $D[l]$ and morphisms $Df:D[l]\rightarrow D[m]$ of $D\varOmega$ have 
as sets and functions when $D\varOmega$ is a concrete category do matter in the applications of the theory.

A graded $\varOmega$ monad is a concrete $\varOmega$ category $D\varOmega$ 
together with an associative and unital graded monadic multiplication. The monoidal and monadic structures
of $D\varOmega$ are compatible but distinct. While the former acts on the objects and morphisms
of $D\varOmega$, the latter does on the elements forming the objects of $D\varOmega$, 
the relationship of the two being analogous to that occurring between tensor multiplication of
Hilbert spaces and vectors belonging to those spaces. Graded $\varOmega$ monad morphisms can be defined
as maps of graded $\varOmega$ monads compatible with their monadic structures. Graded $\varOmega$ monads
and their morphisms form a category $\ul{\rm GM}_\varOmega$, which is closely related 
to the category $\ul{\rm GM}$ of graded monads \ccite{Mellies:2012mea,Katsumata:2014pms,Mellies:2017pcm,Fujii:2019gim}. 

The finite ordinal category $\varOmega$, $\varOmega$ categories and graded $\varOmega$ monads
are studied in detail in sect. 2 of I. 

The ordinal Pro category $\varOmega$ provides a
combinatorial model which describes units as diverse as hypergraph vertices or dits. Each set of $l$ units is identified with
the ordinal set $[l]$ viewed as a standard unit set. The functions mapping a set of $l$ units into one
of $m$ units are accordingly identified with the ordinal functions $f:[l]\rightarrow[m]$ 
construed as standard functions of standard unit sets. The operation of disjoint union of sets and functions of units
are then described by the monoidal products of the ordinal sets and functions which represent them.
This model underlies the main graded $\varOmega$ monads entering our theory.

Hypergraphs are collections of hyperedges and hyperedges are non empty sets of vertices. 
Hypergraphs can be bare or else endowed with hyperedge data such as calibrations and weightings.
Bare, calibrated and weighted hypergraphs are organized in appropriate graded $\varOmega$ monads,
the bare, calibrated and weighted hypergraph $\varOmega$ monads 
$G\varOmega$, $G_C\varOmega$ and $G_W\varOmega$, respectively. 
For each ordinal $[l]$, $G[l]$ is the set of all hypergraphs over the vertex set $[l]$;  
in the same manner, for each ordinal function $f:[l]\rightarrow[m]$, $Gf:G[l]\rightarrow G[m]$ is the function
of hypergraph sets induced by $f$ as a vertex set function. The monadic multiplicative structure of
$G\varOmega$ is patterned after the disjoint union of hypergraphs as collections of hyperedges.
$G_C\varOmega$ and $G_W\varOmega$ have a richer content than $G\varOmega$ but otherwise a completely similar 
design. 

The hypergraph $\varOmega$ monad $G\varOmega$ and its calibrated and weighted augmentations $G_C\varOmega$ and 
$G_W\varOmega$ are studied in depth in sect. 3 of I. 

In classical theory, the distinct configurations of a single cdit are indexed by some finite set $\msR$ of labels.
Computation theory requires that $\msR$ is endowed with some algebraic structure depending on context.
In any case, $\msR$ is at least a commutative monoid.
In quantum theory, through base encoding, the independent states of a single qudit are indexed by the same monoid $\msR$
as its cdit counterpart. 
Such states span a state Hilbert space $\scH_1$ of dimension $|\msR|$. 


Multi cdit configurations and multi qudit states are described by specific graded $\varOmega$ monads,
the multi cdit configuration and  multi qudit state $\varOmega$ monads $E\varOmega$ and $\scH_E\varOmega$.
For each ordinal $[l]$, $E[l]$ is the set of all configurations of the cdit set $[l]$; 
correspondingly, for each ordinal function $f:[l]\rightarrow[m]$, $Ef:E[l]\rightarrow E[m]$ is the
function of configuration sets induced by $f$ as a function of cdit sets. The monadic multiplicative
structure of $E\varOmega$ is defined as concatenation of configurations viewed as strings of $\msR$ labels.
Since $\scH_E\varOmega$ is the basis encoding of of $E\varOmega$, 
the layout of $\scH_E\varOmega$ mirrors closely that of $E\varOmega$ with multi qudit state Hilbert spaces 
replacing multi cdit configuration sets, linear operators between those spaces
taking the place of the functions between those sets and 
tensor multiplication of state vectors substituting concatenation of configuration strings. 

The multi cdit configuration and multi qudit state $\varOmega$ monads $E\varOmega$ are 
$\scH_E\varOmega$ are defined and studied in detail in sect. 4 of I. 


\subsection{\textcolor{blue}{\sffamily The calibrated hypergraph state construction}}\label{subsec:design}

The calibrated hypergraph state construction is at the same time the final goal and the
ultimate justification of the elaboration of the categorical set--up worked out in I.
The proper definition and the study of the calibrated hypergraph state map
lies at the heart of the present paper. 

Before proceeding to showing how the calibrated hypergraph state construction can be formulated in an
$\varOmega$ monadic framework, it is necessary to define the notation used.
A calibrated hypergraph specified by a hypergraph $H$ and a calibration $\varrho$ over $H$
is expressed by a pair $(H,\varrho)$. The monadic multiplication of calibrated hypergraphs is indicated by the symbol
$\smallsmile$ to recall its similarity to disjoint union.
The calibrated hypergraph monadic unit is denoted by $(O,\varepsilon)$ to remind its being empty. A multi qudit state vector
is denoted by a Dirac ket $\ket{\varPsi}$. The monadic multiplication of state vectors is indicated by
the tensor product symbol $\otimes$, in accordance with its mathematical meaning. The state vector monadic unit
is written as $\ket{0}$ signifying the complex unit.

The calibrated hypergraph state construction associates with any calibrated hypergraph
$(H,\varrho)\in G_C[l]$ a hypergraph state $\ket{(H,\varrho)}\in\scH_E[l]$. It enjoys three basic properties.
First, it is covariant, meaning that 
\begin{equation}
\label{i/whgsts16}
\scH_Ef\ket{(H,\varrho)}=\ket{G_Cf(H,\varrho)}
\end{equation}
for any ordinal function $f:[l]\rightarrow[m]$. Second, it is compatible with monadic multiplication
in the sense that 
\begin{equation}
\label{i/whgsts18}
\ket{(H,\varrho)\smallsmile(K,\varsigma)}=\ket{(H,\varrho)}\otimes\ket{(K,\varsigma)}
\end{equation}
for $(H,\varrho)\in G_C[l]$, $(K,\varsigma)\in G_C[m]$. Third, it satisfies 
\begin{equation}
\label{i/whgsts19}
\ket{(O,\varepsilon)}=\ket{0}.
\end{equation}
By these properties, the collection of the maps 
$\ket{-}:G_C[l]\rightarrow \scH_E[l]$ assigning the 
hypergraph state $\ket{(H,\varrho)}\in\scH_E[l]$ to each calibrated hypergraph $(H,\varrho)\in G_C[l]$
for the different values of $l$ specify a distinguished morphism of the graded $\varOmega$ monads 
$G_C\varOmega$, $\scH_E\varOmega$. 
This $\varOmega$ monad theoretic interpretation of the calibrated hypergraph state map $\ket{-}$ is not only 
interesting in its own, but underlies a powerful approach to the analysis of calibrated hypergraph states:
relations \ceqref{i/whgsts16}--\ceqref{i/whgsts19} constrain to a considerable extent the 
map $\ket{-}$ and so provide valuable information about its form 
and essentially determine the nature of the hypergraph calibration data, which we have not spelled out yet.
This, ultimately, is the reason why our categorical set--up is not an end in itself 
but provides a basic framework for the study of hypergraph states.
We outline the results of this investigation next.


As we noticed earlier, computation theory requires the set $\msR$ labelling a cdit's configurations
to have some kind of algebraic structure.  It is often posited that $\msR$ is a Galois field. However, this assumption
may sometimes be relaxed. We have found that for our construction of hypergraph states to work it is enough to demand
that $\msR$ is a Galois ring.
Refs. \ccite{Wan:2011gfr,Bini:2002gfr,Kibler:2017gre} and app. \cref{app:galois} 
provide background on Galois ring theory. 

A Galois ring $\msR$ is characterized uniquely up to isomorphism by its characteristic $p^r$,
a prime power, and degree $d$. It contains $q=p^{rd}$ elements altogether.
It further features a prime subring $\msP$, a minimal Galois subring isomorphic to the ring $\bbZ_{p^r}$ of integers
modulo $p^r$. A trace function $\tr:\msR\rightarrow\msP$ is defined with relevant surjectivity and non singularity
properties.

A Galois ring $\msR$ is finite. Thus, for any fixed $x\in\msR$ the exponents $u\in\bbN$ giving distinct powers
$x^u$ form a finite set $\msZ_x$ of exponents depending on $x$. $\msZ_x$ has an algebraic structure
of cyclic monoid. The monoids $\msZ_x$ for varying $x\in\msR$
can be assembled in the direct sum monoid $\msZ=\bigoplus_x\msZ_x$, 
the cyclicity monoid of $\msR$. The power of a ring element $x\in\msR$ to the exponent $u\in\msZ$ is defined to be given
by $x^u=x^{u_x}$ if $u=(u_x)_{x\in\msR}$. 

The properties of Galois rings most relevant in the theory of calibrated hypergraph states,
in particular the proper definition and structure of their cyclicity monoids,
are reviewed in sect. \cref{sec:galqud}. 
Other topics related to these, such as the qudit Pauli group
and quantum Fourier transformation, though well-known, are reviewed in the same section for later use. 

We illustrate next the calibrated hypergraph state construction in some detail. 
An exponent function of a given vertex set $X$ is any function $w:X\rightarrow\msZ$.
Since $\msZ$ is a monoid, the exponent functions of $X$ form a monoid, the exponent
monoid $\msZ^X$. A calibration of $X$ is any function $\varrho_X:\msZ^X\rightarrow\msP$.
A calibration $\varrho$  over a hypergraph $H\in G[l]$ is an assignment to each hyperedge $X\in H$ of a calibration
$\varrho_X$ of $X$.

In structural terms, the hypergraph state $\ket{(H,\varrho)}\in\scH_E[l]$ associated with a calibrated hypergraph
$(H,\varrho)\in G_C[l]$ reads as 
\begin{equation}
\label{i/whgsts10}
\ket{(H,\varrho)}=D_{(H,\varrho)}\ket{0_l}, \vphantom{\sum}
\end{equation}
where $D_{(H,\varrho)}$ is a linear operator of $\scH_E[l]$ of the form 
\begin{equation}
\label{i/whgsts11}
D_{(H,\varrho)}=\mycom{{}_\sss}{{}_{x\in E[l]}}F_l{}^+\ket{x}\,\omega^{\sigma_{(H,\varrho)}(x)}\hfpt\bra{x}F_l.
\end{equation}
Here, $F_l$ is the quantum Fourier transform operator and $\omega=\exp(2\pi i/p^r)$. 
The vectors $\ket{x}$, $x\in E[l]$, constitute the qudit Hadamard basis of
$\scH_E[l]$ to which $\ket{0_l}$ belongs. 
$\sigma_{(H,\varrho)}:E[l]\rightarrow\msP$ is a phase function depending on $(H,\varrho)$ reading as 
\begin{equation}
\label{i/whgsts9}
\sigma_{(H,\varrho)}(x)=\mycom{{}_\sss}{{}_{X\in H}}\mycom{{}_\sss}{{}_{w\in\msZ^X}}\varrho_X(w)
\tr\left(\mycom{{}_\ppp}{{}_{r\in X}}x_r{}^{w(r)}\right) 
\end{equation}
with $x\in E[l]$. 

The definition of hypergraphs states we have submitted is justified ultimately 
by the following theorem which is the main result of the paper. 

\begin{theor} \label{theor:main}
The calibrated hypergraph state map defined by expressions \ceqref{i/whgsts10}--\ceqref{i/whgsts9}
satisfies properties \ceqref{i/whgsts16}--\ceqref{i/whgsts19}.
\end{theor}

\noindent
The hypergraph states outputted by the map have all the properties expected
from such states.

\begin{theor} \label{theor:basic}
The calibrated hypergraph states given by expressions \ceqref{i/whgsts10}--\ceqref{i/whgsts9}
are locally maximally entangleable stabilizer states \ccite{Kruszynska:2008lem}. 
\end{theor}

\noindent
The  stabilizer group operators $K_{(H,\varrho)}(a)$ of the hypergraph states $\ket{(H,\varrho)}$ 
can be explicitly found. Further, the associated hypergraph orthonormal bases $\ket{(H,\varrho),a}$
can also be obtained. 

The phase function $\sigma_{(H,\varrho)}$ shown in \ceqref{i/whgsts9}
determines the hypergraph state $\ket{(H,\varrho)}$ and its properties. 
The definitions of $\sigma_{(H,\varrho)}$ we have put forward is however not straightforward and for such a reason 
needs to be properly commented.

The phase function $\sigma_{(H,\varrho)}$ receives a contribution from each hyperedge $X$ of the hypergraph $H$.
In turn, the component of $X$ gets a contribution from each exponent function $w$ of $X$.
The component of $w$ carries a $\msP$ valued weight $\varrho_X(w)$ determined by the calibration $\varrho_X$ of $X$.
It further features the Galois qudit variables $x_r$ attached to the vertices $r\in X$ raised to the 
powers $w(r)$, $x_r{}^{w(r)}$. Notice however that the $w(r)$ are not ordinary exponents as defined in elementary algebra
and consequently the $x_r{}^{w(r)}$ are not ordinary powers. Rather, the $w(r)$ are tuples of ordinary exponents
$w(r)_u$, one for each element $u$ of the ring $\msR$, and the expressions $x_r{}^{w(r)}$ really stand
for the ordinary powers $x_r{}^{w(r)_{x_r}}$. $\sigma_{(H,\varrho)}(x)$ so can be described as a polynomial function 
in the Galois qudit variables $x_r$ only in the more general sense just explained.

There is a least integer $\delta\in\bbN_+$ depending only on the ring $\msR$ such that the powers 
$x^k$ with $k\in[\delta]$ comprises all non negative powers of any ring element $x\in\msR$.
Hypergraph states with phase functions given by an expression similar to \ceqref{i/whgsts9}
with the exponent monoid $\msZ^X$ replaced by the set $[\delta]^X$ of exponent
functions $w:X\rightarrow[\delta]$  can be shown to be instances of calibrated hypergraph states.
These states are truly polynomial but however not the most general possible ones for a general Galois ring $\msR$.
This conclusion must be reconsidered when $\msR$ is a Galois field, as will be addressed momentarily. 

Our calibrated hypergraph state construction scheme is in its main lines
analogous to that of previous literature such as refs. \!\!\ccite{Qu:2012eqs,Rossi:2012qhs}
or, in the qudit case, \ccite{Steinhoff:2016:qhs,Xiong:2017qhp}. In those works, however,
weighted hypergraphs rather than calibrated hypergraphs are used. Stated in our $\varOmega$ monadic language
for comparison, with any weighted hypergraph $(H,\alpha)\in G_W\varOmega$
these authors  associate a weighted hypergraph state $\ket{(H,\alpha)}\in\scH_E[l]$ given by
an expression analogous to \ceqref{i/whgsts10}, \ceqref{i/whgsts11} but with $\sigma_{(H,\alpha)}$ of the form 
\begin{equation}
\label{i0/whgsts9}
\sigma_{(H,\alpha)}(x)=\mycom{{}_\sss}{{}_{X\in H}}\alpha_X
\tr\left(\mycom{{}_\ppp}{{}_{r\in X}}x_r\right),
\end{equation}
where $\alpha_X\in\msP$ is the weight factor of hyperedge $X$.
This construction certainly has a sound grounding but the hypergraph states it produces 
are not sufficiently general for property \ceqref{i/whgsts16} to be enjoyed for a generic
ordinal function $f$. By design, conversely, the calibrated hypergraph states 
possess this property. 

Formula \ceqref{i0/whgsts9} is a direct extrapolation to the qudit case of the expression holding in the basic
qubit case, where the operator generating the hypergraph state can be seen to reduce to a sequence 
of familiar controlled phase gates, one for each hyperedge of the underlying hypergraph counting weights.
For this reason, the reader may be unwilling to modify it in the way
we have indicated. This point of view is certainly legitimate. We reply to this criticism by pointing out that
extrapolations are by their nature non unique, calling for the exploration of alternatives.
We quote next a further results which may support our proposal.

\begin{theor} \label{theor:calweicomp}
For each weighted hypergraph $(H,\alpha)\in G_W[l]$, there exists a calibrated hypergraph
$(H,\varrho)\in G_C[l]$ with the same underlying hypergraph $H\in G[l]$ with the property that
$\ket{(H,\alpha)}=\ket{(H,\varrho)}$.
\end{theor}

\noindent
In other words, calibrated hypergraph states contain weighted hypergraph states as a subset.
In this sense, the calibrated hypergraph approach essentially subsumes the weighted hypergraph one.
We could summarize this relationship in simple words by saying that while 
the phase functions of calibrated hyper graph states
exhibit all possible powers of the Galois qudit variables those of the weighted hypergraph states do
only the $0$-th and $1$-st ones.
The critical question about whether genuinely new hypergraph state entanglement classes are generated
by broadening the scope of hypergraph theory in this way will not be addressed in this work. 

Fig. \cref{fig:hgraphdiag} compares pictorially through a simple example the different ways
a hypergraph's weight and calibration data encode the associated weighted and calibrated hypergraph states. 

\begin{figure}[!t]
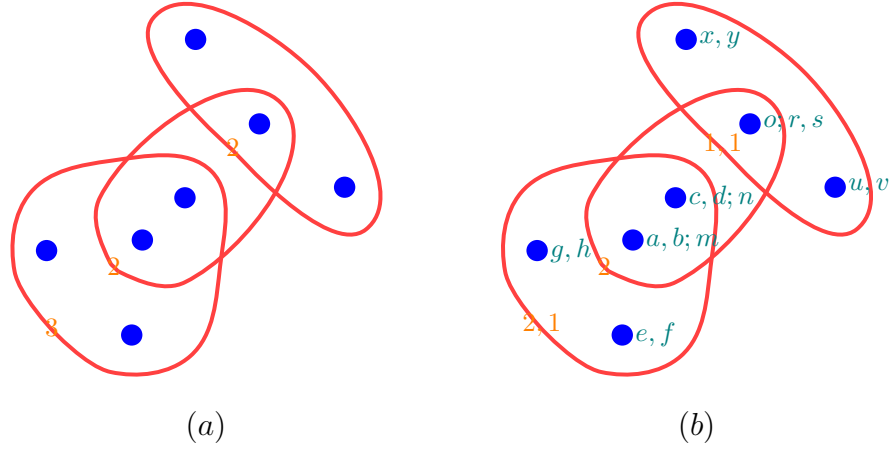

\vspace{.75cm}
\hbox{
\hspace{1.5cm}
{\psset{xunit=.28,yunit=.28,runit=.4}{
\pspicture(-.5,-1.5)(17.5,17.5)

\pscurve[linecolor=red!75,linewidth=1.5pt](2.25,2.25)(5,.25)(8.75,1)(10.15,5.25)(9.75,9.75)(5,10.25)(1.25,8.75)(.5,5)(2.25,2.25)

\pscurve[linecolor=red!75,linewidth=1.5pt](5.195,5.195)(5,5.4)(4.5,8.5)(13,13)(8.5,4.5)(5.4,5)(5.195,5.195)

\pscurve[linecolor=red!75,linewidth=1.5pt](10.75,10.75)(17,7)(14.5,14.5)(7,17)(10.75,10.75)


{\psset{linecolor=blue}
\qdisk(6,2){4.pt}
\qdisk(2,6){4.pt}
\qdisk(6.5,6.5){4.pt}
\qdisk(8.5,8.5){4.pt}
\qdisk(12,12){4.pt}
\qdisk(16,9){4.pt}
\qdisk(9,16){4.pt}
}

\uput[45](1.5,1.5){\footnotesize\textcolor{orange}{$3$}}
\uput[45](4.4,4.4){\footnotesize\textcolor{orange}{$2$}}
\uput[45](10.,10.){\footnotesize\textcolor{orange}{$2$}}
\uput[-90](9.5,-1){$(a)$}

\endpspicture
}}

\hspace{1cm}

{\psset{xunit=.28,yunit=.28,runit=.4}{   
\pspicture(-.5,-1.5)(17.5,17.5)

\pscurve[linecolor=red!75,linewidth=1.5pt](2.25,2.25)(5,.25)(8.75,1)(10.15,5.25)(9.75,9.75)(5,10.25)(1.25,8.75)(.5,5)(2.25,2.25)

\pscurve[linecolor=red!75,linewidth=1.5pt](5.195,5.195)(5,5.4)(4.5,8.5)(13,13)(8.5,4.5)(5.4,5)(5.195,5.195)

\pscurve[linecolor=red!75,linewidth=1.5pt](10.75,10.75)(17,7)(14.5,14.5)(7,17)(10.75,10.75)


{\psset{linecolor=blue}
\qdisk(6,2){4.pt}
\qdisk(2,6){4.pt}
\qdisk(6.5,6.5){4.pt}
\qdisk(8.5,8.5){4.pt}
\qdisk(12,12){4.pt}
\qdisk(16,9){4.pt}
\qdisk(9,16){4.pt}
}

\uput[45](1.5,1.5){\hspace{-5pt}\footnotesize\textcolor{orange}{$2,1$}}
\uput[45](4.4,4.4){\footnotesize\textcolor{orange}{$2$}}
\uput[45](10.,10.){\hspace{-5pt}\footnotesize\textcolor{orange}{$1,1$}}
\uput[0](6.5,6.5){\footnotesize\textcolor{teal}{$a,b;m$}}
\uput[0](8.5,8.5){\footnotesize\textcolor{teal}{$c,d;n$}}
\uput[0](12,12){\footnotesize\textcolor{teal}{$o;r,s$}}
\uput[0](6,2){\footnotesize\textcolor{teal}{$e,f$}}
\uput[0](2,6){\footnotesize\textcolor{teal}{$g,h$}}
\uput[0](16,9){\footnotesize\textcolor{teal}{$u,v$}}
\uput[0](9,16){\footnotesize\textcolor{teal}{$x,y$}}
\uput[-90](9.5,-1){$(b)$}

\endpspicture
}}
}
\vspace{.8cm}
\caption{\small Hypergraphs can be rendered pictorially by representing
vertices as dots (blue) and hyperedges as closed lines encircling the dots corresponding
to their vertices (red). Panel $(a)$ shows a weighted hypergraph encoding a weighted hypergraph
state. The integer juxtaposed to each hyperedge (orange) indicates its weight.
Panel $(b)$ shows a calibrated hypergraph encoding a calibrated hypergraph state.
The symbols attached to each vertex (teal) indicate the exponents to which the underlying qudit's
Galois variable is raised. The integers collocated upon each hyperedge (orange) indicate the weights
with which the corresponding exponent assignments of the encircled vertices appear. When a vertex belongs to several
hyperedges, several exponent strings separated by colons are placed on it. 
\label{fig:hgraphdiag}}
\end{figure}

For qubits, for which the relevant Galois ring $\msR$ is the simple
binary field $\bbF_2$, the calibrated hypergraph state construction does not yield
new states  beyond the weighted hypergraph ones already known.

\begin{theor} \label{theor:qudqubcomp}
Assume that $\msR=\bbF_2$ is the binary field. Then, for each calibrated hypergraph $(H,\varrho)\in G_C[l]$
there exists a weighted hypergraph $(K,\alpha)\in G_W[l]$ such that 
$\ket{(K,\alpha)}=\ket{(H,\varrho)}$ up to a sign.
\end{theor}

The matters outlined above constitute the main results of our endeavour. The are found in sect. \cref{sec:grstt},
which contains also a optimization scheme for the classification problem of hypergraph states
based on the set--up proposed. 

When the Galois ring $\msR$ is a field, the general expression of the calibrated hyper\-graph state map
we have obtained can be recast into an alternative form 
thanks to the existence of a distinguished set of polynomials drawn from the
polynomial ring $\msR[\sfx]$ computing the  powers of the field's elements to
the cyclicity monoid's exponents.

\begin{prop} \label{prop:i/huxu}
Let $\msR$ be a Galois field. Then, for each exponent $u\in\msZ$, there is a unique polynomial $m_u(\sfx)\in\msR[\sfx]$
of degree at most $q-1$ such that $x^u=m_u(x)$ for every $x\in\msR$. 
\end{prop}

\noindent
By this result, the phase function $\sigma_{(H,\varrho)}$ given in \ceqref{i/whgsts9} can so be expressed through 
the polynomials $m_u(\sfx)$ as
\begin{equation}
\label{i/whgsts9/pol}
\sigma_{(H,\varrho)}(x)=\mycom{{}_\sss}{{}_{X\in H}}\mycom{{}_\sss}{{}_{w\in\msA^X}}\varrho_X(w)
\tr\left(\mycom{{}_\ppp}{{}_{r\in X}}m_{w(r)}(x_r)\right)\!.
\end{equation}
An algorithm exists which furnishes the polynomials $m_u(\sfx)$ in terms only of the data underlying the field $\msR$. 
These results do not extend to a general Galois ring.


Suppose that $\msR$ is a prime field $\bbF_p$. Fix a reference element 
$x^*\in\msR$, typically $x^*=p-1\in\bbF_p$. 
Let $X\subseteq[l]$ be a hyperedge with at least two vertices. 
We think of a distinguished vertex $r_X\in X$ as the target vertex and the remaining vertices $r\in X\,\setminus\,\{r_X\}$
as the control vertices of $X$. 
The controlled phase gate $CZ_{(X,r_X)}$ can be defined as the operator of $\scH_E[l]$ acting as follows. 
$CZ_{(X,r_X)}F_l{}^+\ket{x}=F_l{}^+\ket{x}$ unless all control qudits are
in the state $\ket{x^*}\in\scH_1$, in which case we have $CZ_{(X,r_X)}F_l{}^+\ket{x}=F_l{}^+\ket{x}\,\omega^{x_{r_X}}$, 
the phase factor $\omega^{x_{r_X}}$ being so  determined by the state $\ket{x_{r_X}}\in\scH_1$ of the target
qudit  \ccite{Giri:2024qtt}.

A marked hypergraph $(H,r_\cdot)$ consists of a hypergraph
$H\in G[l]$ together with a choice of a target vertex $r_X\in X$ for each hyperedge $X\in H$.
The hypergraph state associated with a marked hypergraph $(H,r_\cdot)$ is 
\begin{equation}
\label{i/qcz3}
\ket{(H,r_\cdot)}=\mycom{{}_\ppp}{{}_{X\in H}}CZ_{(X,r_X)}\ket{0_l}.
\end{equation}
The above definition of hypergraph state closely parallels the standard one of the qubit case and extends that provided 
in \ccite{Giri:2024qtt} for the qutrit case. $\ket{(H,r_\cdot)}$ is in fact built by the controlled
phase operators $CZ_{(X,r_X)}$ and can arguably be regarded as the proper form of hypergraph state for prime dimension
qudits. However, it can be shown that $\ket{(H,r_\cdot)}$ is an instance of calibrated hypergraph state. $\ket{(H,r_\cdot)}$ 
has special properties. In particular, its phase function can be expressed as a computable polynomial with $\bbF_p$
coefficients in the qudit variables $x_r$ containing powers of these latter up to the $p-1$--th one.
For $p>2$, hypergraph states of this kind are
interesting, among other reasons, because they are examples of calibrated hypergraph states which are not weighted.

The above selected application of the theory of calibrated hypergraph states we have
developed are detailed in sect. \cref{sec:techno}.








\vfill\eject

\noindent
\markright{\textcolor{blue}{\sffamily Conventions}} 


\noindent
\textcolor{blue}{\sffamily Conventions.}
In this paper, we adopt the following notational conventions.

\begin{enumerate}[leftmargin=*]

\item \label{it:conv1} We indicate by $|A|$ the cardinality of a finite set $A$.

\item \label{it:conv2} For any set $A$, we let $e_A:\emptyset\rightarrow A$ be 
the empty function with range $A$. 


\item \label{it:conv3} Complying with the most widely used convention computer science,
we denote by $\bbN$ the set of all non negative integers. Hence, $0\in\bbN$.
We denote  by $\bbN_+$ the set of all strictly positive integers. 

\item \label{it:conv4} For $l\in\bbN$, we let $[l]=\emptyset$ if $l=0$ and $[l]=\{0,\dots,l-1\}$ if $l>0$
be the standard finite von Neumann ordinals.

\item \label{it:conv5} If $l\in\bbN$ and $A\subset\bbN$, we let $A+l=\{i+l|i\in A\}$. Accordingly, if $M\subset P\bbN$
(the power set of $\bbN$) is a collection of subsets of $\bbN$, we let $M+l=\{A+l|A\in M\}$.

\item \label{it:conv6} An indexed finite set is a finite set $A$ together with a bijection $a:[|A|]\rightarrow A$.
An indexed finite set $(A,a)$ is so described by the $|A|$--tuple $(a(0),\ldots,a(|A|-1))$.
As a rule, 
we shall write the indexed set as $A$ and the associated $|A|$--tuple
as $(a_0,\ldots,a_{|A|-1})$ for simplicity.
When $A$ is a totally ordered set, e.g. a set of integers,
then it is tacitly assumed that the indexing employed
is the one such that $a_r<a_s$ for $r,s\in[|A|]$ with $r<s$, unless otherwise stated. 

\item \label{it:conv7} If $A$, $B$ are finite sets and $f:A\rightarrow B$ is a function and furthermore $A$ is indexed
as $(a_0,\ldots,a_{|A|-1})$, then $f$ is specified by the $|A|$--tuple $(f(a_0),\ldots,f(a_{|A|-1})_{A}$
of its values. 

\end{enumerate}

\vfill\eject

\renewcommand{\sectionmark}[1]{\markright{\thesection\ ~~#1}}

\vspace{.75mm}

\section{\textcolor{blue}{\sffamily Galois qudits and all that}}\label{sec:galqud}

The $\varOmega$ monadic framework of hypergraphs and multi dit modes 
elaborated in sects. 3 and 4 of I respectively is an essential element of the construction
of qudit calibrated hypergraph states carried out subsequently in sect. \cref{sec:grstt}. The theory however requires
that the finite monoid that indexes single dit modes is endowed with a sufficiently
rich algebraic structure, namely that of Galois ring.
In this section, we introduce the algebraic framework and tools necessary for an appropriate and efficient 
treatment of Galois dits. 


Indexing dits by a Galois ring adds new elements to our framework. 
Galois rings are briefly reviewed in subsect. \cref{subsec:galoisrev}.
More information can be found in app. \cref{app:galois}.
Subsect. \cref{subsec:cycmongal} presents a detailed analysis of cyclic submonoids of Galois rings, a technical subject
matter which however deserves a careful treatment for its relevance in the hypergraph state construction.
In subsect. \cref{subsec:qdgalois}, a Galois ring module theory of multi Galois cdit configurations is introduced.
Finally, in subsect. \cref{subsec:qdpauli}, the Galois qudit Pauli group and quantum Fourier transformation
are surveyed. All this material is accompanied by illustrative examples.




\vspace{.75mm}

\subsection{\textcolor{blue}{\sffamily Galois rings}}\label{subsec:galoisrev}

In a classical computational perspective, the basic datum of the multi cdit configuration $\varOmega$ monad $E\varOmega$
is a Galois ring. Namely, the finite commutative monoid $\msR$ used for the construction of $E\varOmega$
in subsect. 4.1 of I is assumed to be the additive monoid of a Galois ring.  
The reason for requiring $\msR$ to have a richer Galois ring structure
rather than a mere finite commutative monoid one will become clear in due time. 

A finite commutative ring is a finite set in which operations of addition, subtraction and multiplication 
are defined and additive and multiplicative unities $0$, $1$ are given with the same formal properties
of the corresponding operations and unities of familiar integer number arithmetic. In particular, multiplication is commutative
and distributive over addition. Multiplicative inverses need not exist. Idempotent and nilpotent
elements may be present. A Galois ring is a finite ring with special properties.
Not all finite rings are Galois, in contrast to finite fields which always are. 
The theory of Galois rings is extremely well-developed \ccite{Wan:2011gfr,Bini:2002gfr}.
See also app. \cref{app:galois} for a brief summary. Here,
we shall quote only those properties which are relevant in our analysis. 

By definition, a Galois ring is a finite commutative ring $\msR$ with unity whose zero divisors
form a principal ideal of the form $p\msR$ for some prime number $p\geq 2$. 
Such ring  is local, since the ideal $p\msR$ is the only maximal ideal.  

A Galois ring $\msR$ is characterized by its characteristic and degree. The former is a prime power, i.e.
an integer of the form $p^r$ with $p,r\in\bbN_+$, $p\geq2$ prime (the same as that mentioned in
the previous paragraph). The latter is an integer $d\in\bbN_+$.
All Galois rings of the same characteristic and degree are isomorphic. One therefore writes
$\msR=\GR(p^r,d)$ to emphasize the isomorphy class of $\msR$ when appropriate. 

A Galois ring $\msR=\GR(p^r,d)$ has cardinality $p^{rd}$. For $d=1$, $\msR$ is just the ring $\bbZ_{p^r}$ of integers
modulo $p^r$. For $r=1$, $\msR$ reduces to the Galois field $\bbF_{p^d}$ with $p^d$ elements.
For $r=1$ and $d=1$, $\msR$ reduces further to the familiar Galois field $\bbF_p$ of the integer modulo $p$. 

A Galois ring $\msR=\GR(p^r,d)$ contains precisely one Galois ring $\msR'=\GR(p^r,d')$ as a subring
whenever $d'$ divides $d$. In particular, $\msR$ always contains a minimal subring, the prime subring
$\msP=\GR(p^r,1)\simeq\bbZ_{p^r}$. $\msP$ is the subring generated by the multiplicative unity $1$. 
The ring $\msR$ itself is a Galois extension of $\msP$. 

By definition, the zero divisors of a Galois ring $\msR=\GR(p^r,d)$ are the elements constituting the principal
ideal $p\msR$. All zero divisors of $\msR$ are nilpotent. No zero divisor is idempotent except for $0$. 

A Galois ring $\msR$ is a free module over its prime subring $\msP$ of dimension equal to its degree 
$d$. This allows to treat $\msR$ much as a vector space, a fact that has important implications in the following analysis.

Multiplication by $x\in\msR$ yields a free module morphism $L_x\in\End_\msP(\msR)$. The trace mapping
$\tr:\msR\rightarrow\msP$ is defined by \hphantom{xxxxxx}
\begin{equation}
\label{galois1}
\tr(x)=\Tr(L_x)   
\end{equation}
with $x\in\msR$. $\tr$ has several relevant properties.
In particular, $\tr$ is $\msP$ linear and surjective. When $\msR=\msP$, $\tr$ reduces to the identity map $\id_\msP$. 
Further, $\tr$ is non singular in the sense that for $x\in\msR$ one has $\tr(xy)=0$ for all $y\in\msR$ if and only if
$x=0$ (cf. prop. \cref{prop:rt3}). The trace map $\tr$ is an essential element of a number of constructions carried
out in this section and in sect. \cref{sec:grstt}.

Ref. \!\!\ccite{Kibler:2017gre} provides a rich collection of illustrative examples of Galois rings and fields.
Here, we shall present a few very well--known ones.

\begin{exa} \label{exa:gr22}
The Galois ring $\msR=\GR(2,2)$. {\rm $\msR$ is actually a field, $\msR=\bbF_4$, the smallest field of non prime
cardinality. $\msR$ is the extension field $\bbF_2[\theta]$ of the field $\bbF_2$
by a root $\theta$ of the primitive irreducible monic degree $2$ polynomial $h_2(\sfx)=1+\sfx+\sfx^2$.
The prime subfield of $\msR$ is therefore $\msP=\bbF_2$. Further, 
The elements of $\msR$ can be written as  %
\begin{equation}
\label{gr22ex1}
x=x_0+x_1\theta
\end{equation}
with $x_0,x_1\in\bbZ_2$. Formally, the ring operation of $\msR$ are those of the polynomial ring $\bbF_2[\sfx]$ 
combined with evaluation at $\theta$ and systematic use of the relation 
\begin{equation}
\label{gr22ex2}
1+\theta+\theta^2=0.
\end{equation}
The trace map of $\msR$ is given by \hphantom{xxxxxxxxxxxx}
\begin{equation}
\label{gr22ex3}
\tr(x)=x_1. 
\end{equation}
}
\end{exa} 

\begin{exa} \label{exa:gr42}
The Galois ring $\msR=\GR(4,2)$. 
{\rm In this case, $\msR$ is the extension ring $\bbZ_4[\theta]$
of the ring $\bbZ_4$ by a root $\theta$ of the basic primitive irreducible
monic degree $2$ polynomial $h_2(\sfx)=1+\sfx+\sfx^2$.
The prime subring of $\msR$ is therefore $\msP=\bbZ_4$. The elements of $\msR$ are so of the form 
\begin{equation}
\label{gr42ex1}
x=x_0+x_1\theta
\end{equation}
with $x_0,x_1\in\bbZ_4$. Formally, the ring operation of $\msR$ are those of the polynomial ring $\bbZ_4[\sfx]$ 
combined with evaluation at $\theta$ and recast in the form above by using that 
\begin{equation}
\label{gr42ex2}
1+\theta+\theta^2=0.
\end{equation}
The trace map of $\msR$ is given by 
\begin{equation}
\label{gr42ex3}
\tr(x)=2x_0+3x_1.
\end{equation}
}
\end{exa}

\begin{exa} \label{exa:gr43}
The Galois ring $\msR=\GR(4,3)$. 
{\rm $\msR$ is the extension ring $\bbZ_4[\theta]$ of the ring $\bbZ_4$ by a root $\theta$ of a basic primitive
irreducible monic degree $3$ polynomial, usually taken to be $h_3(\sfx)=3+\sfx+2\sfx^2+\sfx^3$. The prime subring of
$\msR$ is again $\msP=\bbZ_4$. Its elements are hence of the form \hphantom{xxxxxxxxxxxx}
\begin{equation}
\label{gr43ex1}
x=x_0+x_1\theta+x_2\theta^2
\end{equation}
with $x_0,x_1,x_2\in\bbZ_4$. Again, the ring operation of $\msR$ are those of the polynomial ring $\bbZ_4[\sfx]$ 
with evaluation at $\theta$ and the relation 
\begin{equation}
\label{gr43ex2}
3+\theta+2\theta^2+\theta^3=0
\end{equation}
implemented. The trace map of $\msR$ is given by 
\begin{equation}
\label{gr43ex3}
\tr(x)=3x_0+2x_1+2x_2.
\end{equation}
}
\end{exa}


\subsection{\textcolor{blue}{\sffamily The cyclicity monoid of a Galois ring}}\label{subsec:cycmongal}

The cyclicity monoid of the relevant Galois ring constitutes a key element
of the construction of qudit calibrated hypergraph states in sect. \cref{sec:grstt}.
In the present subsection, we briefly review the theory of  Galois ring cyclicity monoids, 
since this topic is mostly unfamiliar. Most of the results presented actually hold
for any finite ring. Standard references on cyclic submonoids of a finite
monoid are \ccite{Howie:1995sgt,Clifford:1961sgt}. 

Below, we suppose first that $\msR$ is any finite commutative ring with unity. We shall assume 
$\msR$ to be a Galois ring only later on. 

Let $x\in\msR$. The cyclic monoid $\msC_x$ of $x$ is the submonoid of the commutative multiplicative monoid
of $\msR$ generated by $x$. i.e. the smallest submonoid containing $x$. $\msC_x$ is 
the set of all non negative powers of $x$, $\msC_x=\{x^u\hfpt|\hfpt u\in\bbN\}$ with multiplication law
$x^ux^v=x^{u+v}$ and unity $x^0=1$.

The monoid $\msC_x$ is finite for every $x\in\msR$, since the ambient ring $\msR$ is. Thus, there must be integers
$u,v\in\bbN$ with $u\neq v$ such that $x^u=x^v$. The following result holds. 
For any element $x\in\msR$, there exist unique integers
$\iota_x\in\bbN$ and $\pi_x\in\bbN_+$ enjoying the following properties.
{\it
\begin{enumerate}

\item \label{item:gr1} {\rm The elements $x^t\in\msR$ with $t\in\bbN$, $t<\iota_x+\pi_x$ are all distinct.}

\item \label{item:gr2} {\rm For every integer $u\in\bbN$, there is precisely one integer 
$h_x(u)\in\bbN$ with the property that $h_x(u)<\iota_x+\pi_x$ and $x^u=x^{h_x(u)}$.}

\item \label{item:gr3} {\rm For every $u\in\bbN$ with $u\geq\iota_x$ and $q\in\bbN$, one has $x^{u+q\pi_x}=x^u$.}

\end{enumerate}
}

\noindent
The parameters $\iota_x$ and $\pi_x$ are called respectively the index and the period of $x$. 
The function $h_x:\bbN\rightarrow\bbN$ defined in item \cref{item:gr2} above is a key structural element
of the monoid $\msC_x$. It can be explicitly computed: for $u\in\bbN$, $h_x(u)$ is given by 
\begin{equation}
\label{cycmongal1}
h_x(u)=\Bigg\{
\begin{array}{ll}
u&\text{if $u<\iota_x+\pi_x$},\\
\iota_x+\rst(u-\iota_x,\pi_x)&\text{if $u\geq\iota_x+\pi_x$},
\end{array}
\end{equation}
where $\rst(k,l)$ denotes the remainder of the division of $k$ by $l\leq k$.

The ring $\msR$ contains two kinds of elements: the elements whose product with some
non zero element of the ring vanishes, called zero divisors, which form a subset $\msD$ of $\msR$, 
and the elements admitting a multiplicative inverse in the multiplicative monoid of $\msR$, called units,
which form a subgroup $\msR^\times$ of such monoid
\footnote{$\vphantom{\dot{dot{\dot{f}}}}$ We count $x=0$ as a trivial zero divisor.}. 
$\msD$ and $\msR^\times$ are disjoint and $\msR=\msD\sqcup\msR^\times$.
The zero divisors $x\in\msD$ are precisely the elements with index $\iota_x>0$. 
The units $x\in\msR^\times$ are correspondingly the elements whose index $\iota_x=0$. The period $\pi_x$
of either kinds of elements is instead not restricted. 

Items {\it \cref{item:gr1}--\cref{item:gr3}} indicate that for $x\in\msR$ the cardinality of the monoid $\msC_x$ is 
$|\msC_x|=\iota_x+\pi_x$ and that as a set $\msC_x$ can be identified with the von Neumann ordinal
$[\iota_x+\pi_x]$. By virtue of this, the multiplication law and the multiplicative unity of $\msC_x$ can 
be expressed through an addition law and an additive unity on the set $\msZ_x=[\iota_x+\pi_x]$.
The addition law  $+_x$ is defined by the property that for $u,v\in\msZ_x$,
the sum $u+_xv\in\msZ_x$ satisfies
\begin{equation}
\label{cycmongal2}
x^{u+_xv}=x^{u+v}.
\end{equation}
By item \cref{item:gr2}, $u+_xv=$ can be expressed through the function $h_x$ as 
\begin{equation}
\label{cycmongal3}
u+_xv=h_x(u+v).
\end{equation}
The additive unity $0_x\in\msZ_x$ is given simply by $0_x=0$. $\msZ_x$ with the addition law and additive unity just defined 
is evidently a monoid isomorphic to the monoid $\msC_x$, indeed it is its additive form, as can be
verified also directly from \ceqref{cycmongal2}.
Thus, $\msZ_x$ is called the cyclic monoid of $x$ too. 

For any for $x\in\msR$ and $q\in\bbN$, $\msC_{x^q}$ is evidently a submonoid of $\msC_x$. It follows that $\msZ_x$
contains a submonoid isomorphic to $\msZ_{x^q}$. There therefore exists a monoid monomorphism
$i_{x,q};\msZ_{x^q}\rightarrow\msZ_x$. $i_{x,q}$ is given explicitly by 
\begin{equation}
\label{cycmongal4}
i_{x,q}(u)=h_x(qu)
\end{equation}
for $u\in\msZ_{x^q}$ in terms of the function $h_x$.
The cyclic monoid structure of the ring $\msR$ is in this way determined by
that of its primitive elements, that is the elements $x\in\msR$ such that $x\neq y^q$ for every $y\in\msR$
and $q\in\bbN$ with $q\geq2$.

For a unit $x\in\msR^\times$, the cyclic monoids $\msC_x$ and $\msC_{x^{-1}}$ are isomorphic.
So are consequently also $\msZ_x$ and $\msZ_{x^{-1}}$.

The complete cyclic monoid structure of the ring $\msR$ is subsumed by the cyclicity monoid of $\msC$ of
$\msR$ and its additive form $\msZ$. $\msC$ is defined as the direct sum of the cyclic monoids $\msC_x$
of all elements $x\in\msR$, 
\begin{equation}
\label{cycmongal5}
\msC=\ddd_{x\in\msR}\msC_x.
\end{equation}
$\msZ$ is correspondingly given by \hphantom{xxxxxxxx}
\begin{equation}
\label{cycmongal6}
\msZ=\ddd_{x\in\msR}\msZ_x.
\end{equation}   

We describe next $\msZ$ in greater detail because of its relevance in our analysis.
The cardinality of $\msZ$ is $\sum_{x\in\msR}(\iota_x+\pi_x)$ and is always at least $3$. 
An element $u\in\msZ$ is a tuple of the form $u=(u_x)_{x\in\msR}$ with $u_x\in\msZ_x$. 
The sum $u+v\in\msZ$ of a pair of elements $u,v\in\msZ$ is defined to be the tuple $((u+v)_x)$ with 
$(u+v)_x=u_x+_xv_x$. The additive unity $0\in\msZ$ is the tuple $(0_x)$.

To specify the elements of $\msZ$ in practice, one must choose beforehand an ordering of the ring $\msR$.
We shall employ only orderings of the form $\msR=(0,1,x_0,\ldots,x_{|R|-3})$, where 
$x_0,\ldots,x_{|R|-3}$ are the distinct elements of $\msR\,\backslash\,\{0,1\}$. Once one such ordering is
selected, every element $u\in\msZ$ is represented as an $|R|$--tuple $(u_0,u_1,u_{x_0},\ldots,u_{x_{|R|-3}})$,
with $u_0\in\msZ_0$, $u_1\in\msZ_1$ and $u_{x_0}\in\msZ_{x_0}$, ..., $u_{x_{|R|-3}}\in\msZ_{x_{|R|-3}}$.
Notice that one always has $u_0=0,1$ and $u_1=0$. 

It is possible to give a natural meaning to the exponentiation $x^u\in\msR$ of an element $x\in\msR$
to the power $u\in\msZ$. Explicitly, $x^u$ is given by 
\begin{equation}
\label{cycmongal7}
x^u=x^{u_x}.
\end{equation}
In this way, the integer exponent to which $x$ is raised depends on $x$ itself!
The usual properties of exponentiation hold. For $u,v\in\msZ$, one has $x^{u+ v}=x^ux^v$;
furthermore, $x^{0}=1$. For explicitness, we shall 
write $x^u$ as $x^{(u_0,u_1,u_{x_0},\ldots,u_{x_{|R|-3}})}$ whenever appropriate.

The generalized exponentiation operation of the ring $\msR$
introduced in the previous paragraph is key in the construction
of qudit calibrated hypergraph states elaborated in sect. \cref{sec:grstt}. For this reason,
the determination of the cyclic monoids $\msZ_x$ for all elements $x\in\msR$ 
is indispensable. From \ceqref{cycmongal1} and \ceqref{cycmongal3}, it ensues that
\begin{equation}
\label{cycmongal18}
\msZ_x=\bbH_{\iota_x,\pi_x}.
\end{equation}
The two parameter family $\bbH_{\lambda,\mu}$, $\lambda\in\bbN$, $\mu\in\bbN_+$ of monoids
appearing in the right hand side of this identity is defined as follows. For given $\lambda,\mu$, 
the underlying set of $\bbH_{\lambda,\mu}$ is the von Neumann ordinal $[\lambda+\mu]$, the additive unity
of $\bbH_{\lambda,\mu}$ is $0_{\lambda,\mu}=0$ and the addition operation of $\bbH_{\lambda,\mu}$ reads
for $u,v\in\bbH_{\lambda,\mu}$ as 
\begin{equation}
\label{cycmongal19}
u+_{\lambda,\mu}v=\Bigg\{
\begin{array}{ll}
u+v&\text{if $u+v<\lambda+\mu$},\\
\lambda+\rst(u+v-\lambda,\mu)&\text{if $u+v\geq\lambda+\mu$}.
\end{array}
\end{equation}

Owing to \ceqref{cycmongal18}, since $\iota_0=1$, $\pi_0=1$ and $\iota_1=0$, $\pi_1=1$ in every finite ring $\msR$, 
the cyclic monoids $\msZ_0$ and $\msZ_1$ have the same universal form for every such ring, viz
$\msZ_0=\bbH_{1,1}$, $\msZ_1=\bbH_{0,1}$. For $x\in\msR$, $x\neq 0,1$, the monoids $\msZ_x=\bbH_{\iota_x,\pi_x}$
are instead characteristic of the ring $\msR$ considered.


For a nilpotent element $x\in\msR$,  
$\iota_x$ is the index of nilpotency of $x$ and $\pi_x=1$. In this case so, by \ceqref{cycmongal18},
$\msZ_x=\bbH_{\iota_x,1}$. For a unit $x\in\msR$, $\iota_x=0$ and $\pi_x$ is the group theoretic
order of $x$ in the unit group $\msR^\times$. So, again by \ceqref{cycmongal18}, $\msZ_x=\bbH_{0,\pi_x}$.
Two special cases of the monoid $\bbH_{\lambda,\mu}$, $\lambda\in\bbN$, $\mu\in\bbN_+$,
introduced above appear here: $\bbH_{\lambda,1}$ and $\bbH_{0,\mu}$.
$\bbH_{\lambda,1}$ is characterized by its addition operation being of the simple form
given by $u+_{\lambda,1}v=\min(u+v,\lambda)$ for all $u,v\in\bbH_{\lambda,1}$. 
$\bbH_{0,\mu}$ is just the monoid of the integers modulo $\mu$, as is easily checked,
that is the monoid underlying the order $\mu$ cyclic group $\bbZ_\mu$. 
For a Galois ring $\msR$, every zero divisor is a nilpotent; so, every element
is either a nilpotent or a unit. Therefore, for such an $\msR$, the cyclic monoids $\msZ_x$
are all of the two kinds just indicated.

Here are the cyclicity monoids of a few Galois rings.

\begin{exa} \label{exa:gr41cyc}
The cyclicity monoid $\msZ$ of the Galois ring $\msR=\GR(4,1)$.
{\rm In this case $\msR=\bbZ_4$, the modulo $4$ integer ring. The cyclicity monoid $\msZ$ is therefore
\begin{equation}
\label{}
\msZ=\ddd_{x\in\bbZ_4}\msZ_x.
\end{equation}
The only non universal cyclic monoids featured by $\msR$ are $\msZ_2=\bbH_{2,1}$ and $\msZ_3=\bbH_{0,2}$.
reflecting $2$, $3$  being respectively a nilpotent and a unit of $\bbZ_4$. 
}
\end{exa}

\begin{exa} \label{exa:gr22cyc}
The cyclicity monoid $\msZ$ of the Galois ring $\msR=\GR(2,2)$.
{\rm In this case $\msR=\bbF_4$, the cardinality $4$ Galois field described in ex. \cref{exa:gr22}.
The cyclicity monoid $\msZ$ is 
\begin{equation}
\label{}
\msZ=\ddd_{x_0,x_1\in\bbZ_2}\msZ_{x_0+x_1\theta}.
\end{equation}
The only non universal cyclic monoids characterizing $\msR$ are $\msZ_\theta=\msZ_{1+\theta}=\bbH_{0,3}$,
the order $3$ cyclic group monoid, as expected from $\theta$, $1+\theta$
being reciprocally inverse units of $\msR$. 
}
\end{exa}

\begin{exa} \label{exa:gr42cyc}
The cyclicity monoid $\msZ$ of the Galois ring $\msR=\GR(4,2)$.
{\rm The ring $\msR$ is described in ex. \cref{exa:gr42}.
The cyclicity monoid $\msZ$ is 
\begin{equation}
\label{}
\msZ=\ddd_{x_0,x_1\in\bbZ_4}\msZ_{x_0+x_1\theta}.
\end{equation}
The only non universal cyclic monoids characterizing $\msR$ are
$\msZ_2=\msZ_{2\theta}=\msZ_{2+2\theta}=\bbH_{2,1}$ for the nilpotent elements of $\msR$ and 
$\msZ_3=\msZ_{1+2\theta}=\msZ_{3+2\theta}=\bbH_{0,2}$, $\msZ_{\theta}=\msZ_{3+3\theta}=\bbH_{0,3}$,
$\msZ_{1+\theta}=\msZ_{2+\theta}=\msZ_{3+\theta}=\msZ_{3\theta}
=\msZ_{1+3\theta}=\msZ_{2+3\theta}=\bbH_{0,6}$ for the units of $\msR$. 
}
\end{exa}


\subsection{\textcolor{blue}{\sffamily Multi Galois cdit configuration sets as Galois ring modules}}\label{subsec:qdgalois}

In this subsection, we examine the main implications for the multi cdit configuration $\varOmega$ monad $E\varOmega$
ensuing from $\msR$ being the additive monoid a Galois ring $\msR=\GR(p^r,d)$. 

\begin{prop}
For every $l\in\bbN$, $E[l]\in\Obj_{\Vect_\msP}$ is a free module over the prime subring $\msP$.
\end{prop}

\begin{proof}
Indeed, $\msR$ is a free module over $\msP$ and $E[l]=\msR^l$ by eq. (4.1.1) of I.
\end{proof}

\noindent
In what follows, we shall denote by $0_l$ the zero of $E[l]$. However, we shall often
write $0_0=0$ for simplicity.

\begin{defi}
For $l\in\bbN$, the trace pairing of $E[l]$ is the map $\langle\cdot,\cdot\rangle:E[l]\times E[l]\rightarrow\msP$ given 
by the formula \hphantom{xxxxxxxxxx}
\begin{equation}
\label{qdinner1}
\langle x,y\rangle=\mycom{{}_\sss}{{}_{r\in[l]}}\tr(x_ry_r)
\end{equation}
with $x,y\in E[l]$.
\end{defi}

\noindent
Note that $\langle x,y\rangle=0$ when $l=0$.
As hinted by the notation used, the trace pairing is an inner product.

\begin{prop}
For every $l\in\bbN$, the trace pairing $\langle\cdot,\cdot\rangle$ is a non singular inner product on
the module $E[l]$.
\end{prop}

\begin{proof}
The $\msP$ linearity of $\langle\cdot,\cdot\rangle$ follows readily from the property of $\msP$ linearity
of the trace map $\tr$ recalled earlier. The non singularity of of $\langle\cdot,\cdot\rangle$ follows likewise
from the non singularity of $\tr$ (cf. subsect. \cref{subsec:galoisrev}).
\end{proof}

\noindent
Note that $\langle\cdot,\cdot\rangle$ is non definite since $\msP$
does not have any ordering.

\begin{prop}
For $l,m\in\bbN$, the morphism $Ef\in\Hom_{E\varOmega}(E[l],E[m])$ associated with
$f\in\Hom_\varOmega([l],[m])$ 
corresponds to a module morphism $Ef\in\Hom_{\Mod_\msP}(E[l],E[m])$. 

\end{prop}

\begin{proof}
From eq. (4.1.2) of I, $Ef(x)_s=\sum_{r\in[l],f(r)=s}x_r$ for $x\in E[l]$. 
The $\msP$ linearity of $Ef$ is apparent from this formula. 
\end{proof}

Given that  $Ef\in\Hom_{\Mod_\msP}(E[l],E[m])$, $Ef$ admits a unique transposed
morphism $Ef^t\in\Hom_{\Mod_p}(E[m],E[l])$ fully characterized  by the property that  
\begin{equation}
\label{qdinner2}
\langle y,Ef(x)\rangle=\langle Ef^t(y),x\rangle
\end{equation}
for $x\in E[l]$, $y\in E[m]$, by the non singularity of the trace pairing.

\noindent
The transposed morphism $Ef^t$ has a simple expression. 

\begin{prop}
For $l,m\in\bbN$ and $f\in\Hom_\varOmega([l],[m])$, one has
\begin{equation}
\label{qdinner3}
Ef^t(y)_r=y_{f(r)}
\end{equation}
with $r\in[l]$ if $l>0$ and $Ef^t(y)=0$ if $l=0$ for $y\in E[m]$.   
\end{prop}

\begin{proof}
Let $x\in E[l]$, $y\in E[m]$. Then, we have
{\allowdisplaybreaks
\begin{align}
\label{}
\langle Ef^t(y),x\rangle&=\langle y,Ef(x)\rangle
\\
\nonumber
&=\mycom{{}_\sss}{{}_{s\in[m]}}\tr(y_sEf(x)_s)
\\
\nonumber
&=\mycom{{}_\sss}{{}_{s\in[m]}}\mycom{{}_\sss}{{}_{r\in[l]f(r)=s}}\tr(y_sx_r)
\\
\nonumber
&=\mycom{{}_\sss}{{}_{r\in[l],s\in[m]}}\delta_{f(r),s}\tr(y_sx_r)=\mycom{{}_\sss}{{}_{r\in[l]}}\tr(y_{f(r)}x_r).
\end{align}
}
\!\!From here, \ceqref{qdinner3} is evident. 
\end{proof}

The trace pairing enjoys the following further property relating to the multi cdit monadic 
multiplication. 

\begin{prop}
For $l,m\in\bbN$, $x,u\in E[l]$, $y,v\in E[m]$, one has
\begin{equation}
\label{qdinner4}
\langle x\smallsmile y, u\smallsmile v\rangle=\langle x, u\rangle+\langle y, v\rangle.
\end{equation}
\end{prop}

\noindent
The monadic product $\smallsmile$ is defined in (4.1.8) of I. 

\begin{proof}
Indeed, from \ceqref{qdinner1}, we have 
{\allowdisplaybreaks
\begin{align}
\label{}
\langle x\smallsmile y, u\smallsmile v\rangle
&=\mycom{{}_\sss}{{}_{r\in[l+m]}}\tr(x\smallsmile y_r u\smallsmile v_r)
\\
\nonumber
&=\mycom{{}_\sss}{{}_{r\in[l]}}\tr(x\smallsmile y_r u\smallsmile v_r)
+\mycom{{}_\sss}{{}_{r\in[m]+l}}\tr(x\smallsmile y_r u\smallsmile v_r)
\\
\nonumber
&=\mycom{{}_\sss}{{}_{r\in[l]}}\tr(x_ru_r)
+\mycom{{}_\sss}{{}_{r\in[m]+l}}\tr(y_{r-l} v_{r-l})
=\langle x, u\rangle+\langle y, v\rangle,
\end{align}
}
\!\!showing \ceqref{qdinner4}. 
\end{proof}

We conclude this subsection presenting a few examples. The expressions of the trace pairings shown
are cast in diagonal form for convenience. 

\begin{exa} \label{exa:gr22tr}
The case of the Galois ring $\msR=\GR(2,2)$.
{\rm Recall that $\msR=\bbF_4$, the cardinality $4$ Galois field described in ex. \cref{exa:gr22},
whose prime subfield is $\msP=\bbF_2$. $E[l]$ is therefore a $\bbF_2$ module. Its trace pairing 
is given by
\begin{equation}
\label{}
\langle x,y\rangle=\mycom{{}_\sss}{{}_{r\in[l]}}\big[x_{r0}y_{r0}+(x_{r0}+x_{r1})(y_{r0}+y_{r1})\big]
\end{equation}
with $x,y\in E[l]$, where the representation \ceqref{gr22ex1} of the elements of $\msR$ is used. 
}
\end{exa} 

\begin{exa} \label{exa:gr42tr}
The case of Galois ring $\msR=\GR(4,2)$. 
{\rm The ring $\msR$ is described in ex. \cref{exa:gr42}. Its prime subring is $\msP=\bbZ_4$.
$E[l]$ is thus a $\bbZ_4$ module. Its trace pairing reads as 
\begin{equation}
\label{}
\langle x,y\rangle=\mycom{{}_\sss}{{}_{r\in[l]}}\big[3x_{r0}y_{r0}+3(x_{r0}+x_{r1})(y_{r0}+y_{r1})\big]
\end{equation}
with $x,y\in E[l]$, where the elements of $\msR$ are expressed as in \ceqref{gr42ex1}. 
}
\end{exa}

\begin{exa} \label{exa:gr43tr}
The case of the Galois ring $\msR=\GR(4,3)$. 
{\rm The ring $\msR$ is described in ex. \cref{exa:gr43}. Its prime subring is $\msP=\bbZ_4$.
$E[l]$ is so also a $\bbZ_4$ module; its trace pairing takes the form
\begin{align}
\label{}
\langle x,y\rangle&=\mycom{{}_\sss}{{}_{r\in[l]}}
\big[(x_{r0}+x_{r1})(y_{r0}+y_{r1})
\\
\nonumber  
&\hspace{2.5cm}+(x_{r0}+x_{r2})(y_{r0}+y_{r2})+(x_{r0}+x_{r1}+x_{r2})(y_{r0}+y_{r1}+y_{r2})\big]
\end{align}
for $x,y\in E[l]$, the elements of $\msR$ being parametrized as in \ceqref{gr43ex1}. 
}
\end{exa}


\subsection{\textcolor{blue}{\sffamily Galois qudit Pauli groups}}\label{subsec:qdpauli}

The qubit Pauli group has been extensively studied in the literature in relation to
stabilizer error correcting codes. Graph states are stabilizer states, i.e
one--dimensional error correcting codes. The Pauli group enters therefore their
description in an essential way. For qudit hypergraph states, appropriate
generalizations of the Pauli group, the qudit Pauli groups, are required. 
In this subsection, we review basic results of its theory for Galois qudits.
This material is standard \ccite{Ashikhmin:2000nbc,Gheorghiu:2011qsg}. 
We thus provide no proofs. 

The basic datum of a Galois $l$ qudit Pauli group is a Galois ring $\msR=\GR(p^r,d)$ with prime subring
$\msP=\bbZ_{p^r}$. The primitive $p$--th root of unity  
\begin{equation}
\label{qdpauli0}
\omega=\exp(2\pi i/p^r)
\end{equation}
is a key element of its constructions. 
We note that the phase $\omega^a$ is defined for all $a\in\msP$
and that the map $a\rightarrow\omega^a$ is an additive character,
a morphism from the the additive group underlying $\msP$ to the multiplicative group
$\msU(1)\simeq\bbT$.


Let us fix $l\in\bbN$ and work in the Hilbert space $\scH_E[l]$.

\begin{defi}
The Galois $l$ qudit Pauli operators $Z_l(a), X_l(a)\in\End_{\bfsfH}(\scH_E[l])$, where $a\in E[l]$,
are given by \hphantom{xxxxxxxxxxxxx}
{\allowdisplaybreaks
\begin{align}
\label{qdpauli1}
Z_l(a)&=\mycom{{}_\sss}{{}_{x\in E[l]}}\ket{x+a}\hfpt\bra{x},
\\
\label{qdpauli2}
X_l(a)&=\mycom{{}_\sss}{{}_{x\in E[l]}}\ket{x}\,\omega^{\langle a, x\rangle}\hfpt\bra{x}.
\end{align}
} 
\end{defi}

\noindent
The orthonormal basis $\ket{x}$, $x\in E[l]$ is therefore the qudit Hadamard basis in which 
$X_l(a)$ is diagonal while $Z_l(a)$ is not. This is to be contrasted with the
conventional formulation in which the orthonormal basis used is the qudit computational basis in which
$Z_l(a)$ is diagonal instead. 

The Galois $l$ qudit Pauli operators $Z_l(a)$, $X_l(a)$ are to begin with unitary operators, so that 
$Z_l(a), X_l(a)\in\msU(\scH_E[l])$ for $a\in E[l]$. They further obey the relations 
{\allowdisplaybreaks
\begin{align}
\label{qdpauli3}
Z_l(a)Z_l(b)&=Z_l(a+b),
\\
\label{qdpauli4}
Z_l(a){}^{-1}&=Z_l(-a),
\\
\label{qdpauli5}
Z_l(0_l)&=1_l,
\\
\label{qdpauli6}
X_l(a)X_l(b)&=X_l(a+b),
\\
\label{qdpauli7}
X_l(a)^{-1}&=X_l(-a),
\\
\label{qdpauli8}
X_l(0_l)&=1_l,
\\
\label{qdpauli9}
X_l(a)Z_l(b)&=\omega^{\langle a,b\rangle}Z_l(b)X_l(a),
\end{align}
}
\!\!where $a,b\in E[l]$. These follow straightforwardly form the defining relations
\ceqref{qdpauli1}, \ceqref{qdpauli2}. We note that owing to \ceqref{qdpauli3}, \ceqref{qdpauli6}
{\allowdisplaybreaks
\begin{align}
\label{qdpauli10}
Z_l(a)^{p^r}&=1_l,
\\
\label{qdpauli11}
X_l(a)^{p^r}&=1_l,
\end{align}
}
\!\!
since $p^ra=0$ for all $a\in E[l]$. \pagebreak More generally, for $a,b\in E[l]$, $Z_l(a)=Z_l(b)$ if and only if
$X_l(a)=X_l(b)$ if and only if $a=b$.

The operator collection  $\clE_l=\{X(a)Z(b)\hfpt|\hfpt a,b\in E[l]\}$ enjoys 
the following basic properties:
{\it
\begin{enumerate}

\item {\rm $\clE_l\subset\msU(\scH_E[l])$, $\clE_l$ consists of unitary operators in $\scH_E[l]$;}
  
\item {\rm $\clE_l$ contains the unit operator $1_l$;}

\item {\rm the product of two operators in $\clE_l$ is a scalar multiple of another operator in $\clE_l$;}

\item {\rm $\clE_l\setminus\{1_l\}$ consists of traceless operators in $\scH_E[l]$.}
  
\end{enumerate}
}
\noindent
A finite operator set $\clE_l$ with the above features is called an $l$ qudit nice error basis.
The set $\clG_l=\{\omega^cX(a)Z(b)|a,b\in E[l],c\in\msP\}$ is a subgroup of $\msU(\scH_E[l])$, called a 
Galois $l$ qudit Pauli group.

Suppose now that $l,m\in\bbN$ and that $a\in E[l]$, $b\in E[m]$. Then, 
{\allowdisplaybreaks
\begin{align}
\label{qdpauli14}
Z_{l+m}(a\smallsmile b)&=Z_l(a)\otimes Z_m(b),
\\
\label{qdpauli15}
X_{l+m}(a\smallsmile b)&=X_l(a)\otimes X_m(b).  
\end{align}
}
\!\!These relations follow straightforwardly from (4.2.12) of I and \ceqref{qdinner4}. They show that the 
general $l$ qudit Pauli operators can be obtained through  tensor multiplication of $1$ qudit ones. 

The results expounded above clearly indicate that the operators $Z_l(a)$, $X_l(a)$ enjoy the same kind of
algebraic properties. There is a reason for this: they are related by unitary quantum Fourier transformation.
The importance of this relationship motivates reviewing the Fourier transformation operator and its main
properties. 

Many of the results concerning the quantum Fourier transform are based on the following formula, which is also of
independent interest \ccite{Zhang:2009qft}:  
\begin{equation}
\label{char}
\mycom{{}_\sss}{{}_{z\in E[l]}}\omega^{\langle x,z\rangle}=q^l\delta_{x,0_l},
\end{equation}
where $l\in\bbN$ and $x\in E[l]$ and $q=|\msR|=p^{rd}$. 
The mapping $x\mapsto\omega^{\langle x,z\rangle}$ is indeed an additive character of the Abelian group $E[l]$.
As a consequence, the sum of all values of a character vanishes unless the character is trivial. 

Let us fix $l\in\bbN$ and work in the Hilbert space $\scH_E[l]$ again.

\begin{defi}
The $l$ qudit Fourier transform operator is the special automorphism $F_l\in\Aut_{\bfsfH}(\scH_q[l])$ given by \hphantom{xxxxx}
\begin{equation}
\label{four1}
F_l=\mycom{{}_\sss}{{}_{x,y\in E[l]}}\ket{x}\hfpt q^{-l/2}\hfpt\omega^{\langle x,y\rangle}\hfpt\bra{y}.
\end{equation}
\end{defi}

\noindent 
Note that in particular $F_0=1_0$.

From \ceqref{four1}, it appears that the $l$ qudit Fourier transform operator $F_l$ is unitary, 
so that $F_l\in\msU(\scH_E[l])$. 
Conjugation by $F_l$ relates the two types of Galois qudit Pauli operators
introduced earlier. Indeed, the relations 
{\allowdisplaybreaks
\begin{align}
\label{four3}
F_l{}^+X_l(a)F_l&=Z_l(a),
\\
\label{four4}
F_l{}^+Z_l(a)F_l&=X_l(-a)
\end{align}
}
\!\!hold for $a\in E[l]$. \ceqref{qdpauli1}, \ceqref{qdpauli2} and \ceqref{four3}, \ceqref{four4} together imply that
$Z_l(a)$ is diagonal in the orthonormal basis $F_l{}^+\ket{x}$, $x\in E[l]$, while $X_l(a)$ is not.
Such a basis is indeed nothing but the qudit computational basis.


We also note that for $l,m\in\bbN$ 
\begin{equation}
\label{four5}
F_l\otimes F_m=F_{l+m},
\end{equation}
as required also by \ceqref{qdpauli14}, \ceqref{qdpauli15}.
This relation is a simple consequence of \ceqref{four1}. It shows that the generic 
$l$ qudit Fourier operator can be expressed as the $l$--th tensor power of the $1$ qudit one. 

\begin{exa} \label{exa:gr22pg}
The case of the Galois ring $\msR=\GR(2,2)$.
{\rm The setting considered here is that of exs. \cref{exa:gr22}, \cref{exa:gr22tr}.
By virtue of \ceqref{gr22ex1}, each element $x\in\msR$ is represented as a pair of elements $x_0,x_1\in\bbF_2$.
Summation over $x\in E[1]=\msR$ is so reduced to one over $x_0,x_1\in\bbF_2$. For $x_0,x_1,y_0,y_1\in\bbF_2$, set 
\begin{equation}
\label{}
\kappa(x_0,x_1,y_0,y_1)=x_0y_0+(x_0+x_1)(y_0+y_1).
\end{equation}
Then, the Pauli operators $Z_1(a_0,a_1)$, $X_1(a_0,a_1)$ with
$a_0,a_1\in\bbF_2$ are 
{\allowdisplaybreaks 
\begin{align}
\label{}
Z_1(a_0,a_1)&=\mycom{{}_\sss}{{}_{x_0,x_1\in\bbF_2}}\ket{x_0+a_0,x_1+a_1}\hfpt\bra{x_0,x_1},
\\
\label{}
X_1(a_0,a_1)&=\mycom{{}_\sss}{{}_{x_0,x_1\in\bbF_2}}\ket{x_0,x_1}\,\ee^{i\pi\kappa(a_0,a_1,x_0,x_1)}\hfpt\bra{x_0,x_1}.
\end{align}
}
\!\!The Fourier transform operator $F_1$ reads as
\begin{equation}
\label{}
F_1=\mycom{{}_\sss}{{}_{x_0,x_1,y_0,y_1\in\bbF_2}}\ket{x_0,x_1}\,\ee^{i\pi\kappa(x_0,x_1,y_0,y_1)}\mhfpt/2\,\bra{y_0,y_1}.
\end{equation}
}
\end{exa}

\begin{exa} \label{exa:gr42pg}
The case of the Galois ring $\msR=\GR(4,2)$.
{\rm The setting we are considering is that of exs. \cref{exa:gr42}, \cref{exa:gr42tr}.
Owing to \ceqref{gr42ex1}, each element $x\in\msR$ is expressed through a pair of elements $x_0,x_1\in\bbZ_4$.
Summation over $x\in E[1]=\msR$ is so turned into one over $x_0,x_1\in\bbZ_4$. For $x_0,x_1,y_0,y_1\in\bbZ_4$, set 
\begin{equation}
\label{}
\kappa(x_0,x_1,y_0,y_1)=3x_0y_0+3(x_0+x_1)(y_0+y_1).
\end{equation}
Then, the Pauli operators $Z_1(a_0,a_1)$, $X_1(a_0,a_1)$ with $a_0,a_1\in\bbZ$ read as 
{\allowdisplaybreaks 
\begin{align}
\label{}
Z_1(a_0,a_1)&=\mycom{{}_\sss}{{}_{x_0,x_1\in\bbZ_4}}\ket{x_0+a_0,x_1+a_1}\hfpt\bra{x_0,x_1},
\\
\label{}
X_1(a_0,a_1)&=\mycom{{}_\sss}{{}_{x_0,x_1\in\bbZ_4}}\ket{x_0,x_1}\,\ee^{i\pi\kappa(a_0,a_1,x_0,x_1)/2}\hfpt\bra{x_0,x_1}.
\end{align}
}\!\!The Fourier transform operator $F_1$ takes the form 
\begin{equation}
\label{}
F_1=\mycom{{}_\sss}{{}_{x_0,x_1,y_0,y_1\in\bbZ_4}}\ket{x_0,x_1}\,\ee^{i\pi\kappa(x_0,x_1,y_0,y_1)/2}\mhfpt/4\,\bra{y_0,y_1}.
\end{equation}
}
\end{exa}

\begin{exa} \label{exa:gr43pg}
The case of the Galois ring $\msR=\GR(4,3)$.
{\rm The relevant setting is that of exs. \cref{exa:gr43}, \cref{exa:gr43tr}.
Because of \ceqref{gr43ex1}, every element $x\in\msR$ is represented as a triple of elements $x_0,x_1,x_2\in\bbZ_4$.
Summation over $x\in E[1]=\msR$ is so equivalent to one over $x_0,x_1,x_2\in\bbZ_4$. For $x_0,x_1,x_2,y_0,y_1,y_2\in\bbZ_4$, put 
{\allowdisplaybreaks
\begin{align}
\label{}
\kappa(x_0,x_1,x_2,y_0,y_1,y_2)&=(x_0+x_1)(y_0+y_1)
\\
\nonumber
&\hspace{1cm}+(x_0+x_2)(y_0+y_2)+(x_0+x_1+x_2)(y_0+y_1+y_2), 
\end{align}
}
\!\!Then, the Pauli operators $Z_1(a_0,a_1,a_2)$, $X_1(a_0,a_1,a_2)$ with $a_0,a_1,a_2\in\bbZ_4$ are
{\allowdisplaybreaks 
\begin{align}
\label{}
Z_1(a_0,a_1,a_2)&=\mycom{{}_\sss}{{}_{x_0,x_1,x_2\in\bbZ_4}}\ket{x_0+a_0,x_1+a_1,x_2+a_2}\hfpt\bra{x_0,x_1,x_2},
\\
\label{}
X_1(a_0,a_1,a_2)&=\mycom{{}_\sss}{{}_{x_0,x_1,x_2\in\bbZ_4}}\ket{x_0,x_1,x_2}
\,\ee^{i\pi\kappa(a_0,a_1,a_2,x_0,x_1,x_2)/2}\hfpt\bra{x_0,x_1,x_2}.
\end{align}
}
\!\!The Fourier \pagebreak transform operator $F_1$ is given by 
\begin{equation}
\label{}
F_1=\mycom{{}_\sss}{{}_{x_0,x_1,x_2,y_0,y_1,y_2\in\bbZ_4}}\ket{x_0,x_1,x_2}\,
\ee^{i\pi\kappa(x_0,x_1,x_2,y_0,y_1,y_2)/2}\mhfpt/8\,\bra{y_0,y_1,y_2}.
\end{equation}
}
\end{exa}
\vspace{-3mm}

\vfill\eject

\renewcommand{\sectionmark}[1]{\markright{\thesection\ ~~#1}}

\section{\textcolor{blue}{\sffamily Calibrated hypergraph states}}\label{sec:grstt}

In this section, we introduce and study qudit calibrated hypergraph states, which are the main
topic of the present endeavour. We shall do so relying in an essential way on the graded $\varOmega$ monadic framework
of hypergraphs and multi dit modes elaborated in sects. 3 and 4 of I respectively,
which are repeatedly referred to in the following treatment, properly adapted to Galois qudits
along the lines of sect. \cref{sec:galqud}.

Calibrated hypergraph states are introduced and studied in great detail in 
subsect. \cref{subsec:whgsts}. Subsect. \cref{subsec:whgstab} provides a full analysis of 
calibrated hypergraph states as stabilizer states. In subsect. \cref{subsec:chglme},
calibrated hypergraph states are shown to be locally maximally entangleable. 
The optimization of the calibrated hypergraph
state classification problem is discussed in some detail in subsect. \cref{subsec:hgststruc}. 
The relationship of weighted to calibrated hypergraph states is explained in subsect. 
\cref{subsec:hgscompar}, where it is also shown that calibrated states reduce to the
weighted ones in the familiar qubit case.
In the final subsect. \cref{subsec:cwhsmonad}, of a more abstract nature, we prove that
calibrated hypergraph states themselves organize in an $\varOmega$ monad.
All the material presented below is accompanied by illustrative examples and calculations.



\subsection{\textcolor{blue}{\sffamily Calibrated hypergraph states}}\label{subsec:whgsts}

In this subsection, we shall introduce and study Galois qudit calibrated hypergraph states. The theoretical construction
we present below extends and generalizes a number of related designs which have appeared in the
literature \ccite{Helwig:2013amq,Keet:2010qss,Steinhoff:2016:qhs,Xiong:2017qhp}.
Our formulation is somewhat lengthy and elaborated and so will proceed in a step by step
manner for the sake of clarity.

The calibrated hypergraph $\varOmega$ monad $G_C\varOmega$ studied in subsect. 3.2 of I provides
the graph theoretic framework for the construction of qudit calibrated hypergraph states.
We consider specifically the Galois qudits surveyed in sect. \cref{sec:galqud}.
The commutative monoids $\msA$ and $\msM$ entering the definition of $G_C\varOmega$
are thus related to the relevant Galois ring $\msR$ and are taken in this and the ensuing subsections
to be of the following form. 
\begin{enumerate}
{\it
\item \label{item:whgsts1}
{\rm $\msA$ is a submonoid of the additive cyclicity monoid $\msZ$ of the ring $\msR$ (cf. subsect. \cref{subsec:cycmongal}).}
  
\item \label{item:whgsts2}
{\rm $\msM$ is the commutative monoid underlying the additive group of the prime subring $\msP$
of $\msR$ (cf. subsect. \cref{subsec:galoisrev}).}
}
\end{enumerate}

\noindent
In this way, the scheme depends implicitly on the choice of $\msA$. However, 
we shall tacitly assume in the rest of this study that $\msA=\msZ$ 
unless otherwise stated. $\msM=\msP$ is instead fixed.
The symbols $\msA$, $\msZ$ and $\msM$, $\msP$ will so be used interchangeably as appropriate. 
 
There is one calibrated hypergraph state for each $l\in\bbN$ and calibrated hypergraph $(H,\varrho)\in G_C[l]$. 
Basic ingredients of the definition are hence exponent functions and hypergraph calibrations
(cf. defs. 3.2.1, 3.2.3 and 3.2.5 of I).  

The calibrated hypergraph operators introduced next are the basic structural elements of the expression of the
calibrated hypergraph states. For this reason, it is also important to study their formal properties.  

\begin{defi} \label{def:sigmahr}
Let $l\in\bbN$ and let $(H,\varrho)\in G_C[l]$ be a calibrated hypergraph.
The calibrated hypergraph operator $D_{(H,\varrho)}\in\End_{\bfsfH}(\scH_E[l])$
associated with $(H,\varrho)$ is 
\begin{equation}
\label{whgsts10}
D_{(H,\varrho)}=\mycom{{}_\sss}{{}_{x\in E[l]}}F_l{}^+\ket{x}\,\omega^{\sigma_{(H,\varrho)}(x)}\hfpt\bra{x}F_l,
\end{equation}
where $F_l$ is the Fourier transform operator (cf. eq. \ceqref{four1}) and 
the phase function $\sigma_{(H,\varrho)}:E[l]\rightarrow\msP$ is given by
\begin{equation}
\label{whgsts9}
\sigma_{(H,\varrho)}(x)=\mycom{{}_\sss}{{}_{X\in H}}\mycom{{}_\sss}{{}_{w\in\msA^X}}\varrho_X(w)
\tr\left(\mycom{{}_\ppp}{{}_{r\in X}}x_r{}^{w(r)}\right)
\end{equation}
for $x\in E[l]$. 
\end{defi}

\noindent
The ring trace $\tr$ and the power $x^w$ with $x\in\msR$ and $w\in\msZ$ were defined
in subsect. \cref{subsec:galoisrev} and \cref{subsec:cycmongal}, respectively.
For this reason, here and in the following, the addition operation and additive unity of the exponents $w$
will be tacitly understood to be those of $\msZ$ as a monoid. 

A first basic property of calibrated hypergraph operators $D_{(H,\varrho)}$ is their unitarity. 

\begin{prop} \label{prop:dhuni}
Let $l\in\bbN$ and let $(H,\varrho)\in G_C[l]$ be a calibrated hypergraph. Then, $D_{(H,\varrho)}$ is a unitary operator,
so that $D_{(H,\varrho)}\in\msU(\scH_E[l])$. 
\end{prop}

\begin{proof}
The quantum Fourier transform operator $F_l$ is unitary in $\scH_E[l]$.
Inspection of expression \ceqref{whgsts10} shows then that the operator
$D_{(H,\varrho)}$ is also unitary in $\scH_E[l]$ by virtue of its diagonal form
with unit absolute value coefficients. 
\end{proof}

The calibrated hypergraph operators are compatible with the morphism structures of the calibrated hypergraph 
and multi qudit state $\varOmega$ monads $G_C\varOmega$ and $\scH_E\varOmega$
(cf. subsects. 3.2 and 4.2 of I) 
in the sense stated in the following proposition. 

\begin{prop} \label{prop:shefdgcfd}
Let $l,m\in\bbN$ and let $f\in\Hom_\varOmega([l],[m])$ be a morphism. Further, let $(H,\varrho)\in G_C[l]$
be a calibrated hypergraph. Then, the relation
\begin{equation}
\label{whgsts13}
\scH_EfD_{(H,\varrho)}=D_{G_Cf(H,\varrho)}\scH_Ef
\end{equation}
is satisfied
respectively.
\end{prop}

\begin{proof} The demonstration of the proposition is based a couple of lemmas which we prove preliminarily.

\begin{lemma} \label{lemma:whslm1}
For any $l,m\in\bbN$ and $f\in\Hom_\varOmega([l],[m])$,
\begin{equation}
\label{whgsts1}
\scH_Ef=\mycom{{}_\sss}{{}_{y\in E[m]}}F_m{}^+\ket{y}\hfpt q^{(l-m)/2}\hfpt\bra{Ef^t(y)}F_l,
\end{equation}
where $Ef^t(y)$ given by \ceqref{qdinner3}. 
\end{lemma}

\begin{proof} Employing expressions (4.3.1) of I and \ceqref{four1} for $\scH_Ef$ and $F_l$, $F_m$, 
through simple manipulations we obtain 
{\allowdisplaybreaks 
\begin{align}
\label{}
F_m\scH_EfF_l{}^+&=\mycom{{}_\sss}{{}_{x,y\in E[m],x',y',z\in E[l]}}
\ket{y}\hfpt q^{-m/2}\omega^{\langle y,x\rangle}\hfpt\braket{x}{Ef(z)}\hfpt
\braket{z}{x'}\hfpt q^{-l/2}\omega^{-\langle x',y'\rangle}\hfpt\bra{y'}
\\
\nonumber
&=\mycom{{}_\sss}{{}_{x,y\in E[m],x',y',z\in E[l]}}\ket{y}\hfpt q^{-(l+m)/2}\omega^{\langle y,x\rangle-\langle x',y'\rangle}
\delta_{x,Ef(z)}\delta_{z,x'}\hfpt\bra{y'}
\\
\nonumber
&=\mycom{{}_\sss}{{}_{y\in E[m],y',z\in E[l]}}\ket{y}\hfpt q^{-(l+m)/2}\omega^{-\langle z,y'-Ef^t(y)\rangle}\hfpt\bra{y'}
\\
\nonumber
&=\mycom{{}_\sss}{{}_{y\in E[m],y',\in E[l]}}\ket{y}\hfpt q^{(l-m)/2}\delta_{y',Ef^t(y)}\hfpt\bra{y'}
=\mycom{{}_\sss}{{}_{y\in E[m]}}\ket{y}\hfpt q^{(l-m)/2}\hfpt\bra{Ef^t(y)},
\end{align}
}
\!\!where relation \ceqref{char} was used. 
\ceqref{whgsts1} follows immediately from the above computation  by virtue of the unitarity of $F_l$, $F_m$. 
\end{proof}

\begin{lemma} \label{lemma:whslm2}
Let $l,m\in\bbN$ and $f\in\Hom_\varOmega([l],[m])$. Let moreover $A\in\End_{\bfsfH}(\scH_E[l])$,
$B\in\End_{\bfsfH}(\scH_E[m])$ be operators of the form
{\allowdisplaybreaks
\begin{align}
\label{whgsts3}
A&=\mycom{{}_\sss}{{}_{x\in E[l]}}F_l{}^+\ket{x}\,\omega^{\alpha(x)}\hfpt\bra{x}F_l,
\\
\label{whgsts4}
B&=\mycom{{}_\sss}{{}_{y\in E[m]}}F_m{}^+\ket{y}\,\omega^{\beta(y)}\hfpt\bra{y}F_m,
\end{align}
}
\!\!where $\alpha:E[l]\rightarrow\msP$, $\beta:E[m]\rightarrow\msP$ are certain phase functions.
Then, 
\begin{equation}
\label{whgsts7}
\scH_EfA=B\scH_Ef 
\end{equation}
if and only if for every $y\in E[m]$ 
\begin{equation}
\label{whgsts8}
\beta(y)=\alpha(Ef^t(y)).
\end{equation}
\end{lemma}

\begin{proof} Using \ceqref{whgsts1} and \ceqref{whgsts3}, \ceqref{whgsts4}, we obtain 
{\allowdisplaybreaks  
\begin{align}
\label{}
\scH_EfA&=\mycom{{}_\sss}{{}_{x\in E[l],y\in E[m]}}F_m{}^+\ket{y}\hfpt q^{(l-m)/2}
\hfpt\bra{Ef^t(y)}F_lF_l{}^+\ket{x}\,\omega^{\alpha(x)}\hfpt\bra{x}F_l
\\
\nonumber 
&=\mycom{{}_\sss}{{}_{x\in E[l],y\in E[m]}}F_m{}^+\ket{y}\hfpt q^{(l-m)/2}
\hfpt\braket{Ef^t(y)}{x}\,\omega^{\alpha(x)}\hfpt\bra{x}F_l
\\
\nonumber
&=\mycom{{}_\sss}{{}_{x\in E[l],y\in E[m]}}F_m{}^+\ket{y}\hfpt q^{(l-m)/2}\delta_{Ef^t(y),x}\,\omega^{\alpha(x)}\hfpt\bra{x}F_l
\\
\nonumber
&\hspace{4cm}
=\mycom{{}_\sss}{{}_{y\in E[m]}}F_m{}^+\ket{y}\hfpt q^{(l-m)/2}\omega^{\alpha(Ef^t(y))}\hfpt\bra{Ef^t(y)}F_l,
\\
%
\label{}
B\scH_Ef&=\mycom{{}_\sss}{{}_{y,y'\in E[m]}}F_m{}^+\ket{y}\,\omega^{\beta(y)}
\hfpt\bra{y}F_mF_m{}^+\ket{y'}\hfpt q^{(l-m)/2}\hfpt\bra{Ef^t(y')}F_l
\\
\nonumber
&=\mycom{{}_\sss}{{}_{y,y'\in E[m]}}F_m{}^+\ket{y}\,\omega^{\beta(y)}
\hfpt\braket{y}{y'}\hfpt q^{(l-m)/2}\hfpt\bra{Ef^t(y')}F_l
\\
\nonumber
&=\mycom{{}_\sss}{{}_{y,y'\in E[m]}}F_m{}^+\ket{y}\,\omega^{\beta(y)}\hfpt\delta_{y,y'}\hfpt q^{(l-m)/2}\hfpt\bra{Ef^t(y')}F_l
\\
\nonumber
&\hspace{4.8cm}
=\mycom{{}_\sss}{{}_{y\in E[m]}}F_m{}^+\ket{y}\hfpt q^{(l-m)/2}\omega^{\beta(y)}\hfpt\bra{Ef^t(y)}F_l.
\end{align}
}
\!\!Hence, $\scH_EfA=B\scH_Ef$ if and only if $\beta(y)=\alpha(Ef^t(y))$ for $y\in E[l]$, as claimed. 
\end{proof}

We now show \ceqref{whgsts13}. From \ceqref{qdinner3} and \ceqref{whgsts9}, recalling (3.2.2) of I, we have 
{\allowdisplaybreaks 
\begin{align}
\label{si2si}
\sigma_{(H,\varrho)}(Ef^t(y))&=\mycom{{}_\sss}{{}_{X\in H}}\mycom{{}_\sss}{{}_{w\in\msA^X}}\varrho_X(w)
\tr\left(\mycom{{}_\ppp}{{}_{r\in X}}y_{f(r)}{}^{w(r)}\right)
\\
\nonumber
&=\mycom{{}_\sss}{{}_{X\in H}}\mycom{{}_\sss}{{}_{w\in\msA^X}}\varrho_X(w)
\tr\left(\mycom{{}_\ppp}{{}_{s\in f(X)}}y_s{}^{\sum_{r\in X,f(r)=s}w(r)}\right)
\\
\nonumber
&=\mycom{{}_\sss}{{}_{Y\in Gf(H)}}\mycom{{}_\sss}{{}_{v\in\msA^Y}} \mycom{{}_\sss}{{}_{X\in H}}\mycom{{}_\sss}{{}_{w\in\msA^X}}
\delta_{Y,f(X)}\delta_{v,f_\star(w)}\hfpt\varrho_X(w)
\tr\left(\mycom{{}_\ppp}{{}_{s\in f(X)}}y_s{}^{f_\star(w)(s)}\right)
\\
\nonumber
&=\mycom{{}_\sss}{{}_{Y\in Gf(H)}}\mycom{{}_\sss}{{}_{v\in\msA^Y}}
\mycom{{}_\sss}{{}_{X\in H,f(X)=Y}}\mycom{{}_\sss}{{}_{w\in\msA^X,f_\star(w)=v}}\varrho_X(w)
\tr\left(\mycom{{}_\ppp}{{}_{s\in Y}}y_s{}^{v(s)}\right)
\\
\nonumber
&=\mycom{{}_\sss}{{}_{Y\in Gf(H)}}\mycom{{}_\sss}{{}_{v\in\msA^Y}}f_{H*}(\varrho)_Y(v)
\tr\left(\mycom{{}_\ppp}{{}_{s\in Y}}y_s{}^{v(s)}\right)
\\
\nonumber
&=\sigma_{(Gf(H),f_{H*}(\varrho))}(y)=\sigma_{G_Cf(H,\varrho)}(y).
\end{align}
}
\!\!Above, we used expressions (3.2.2) of I in the third line and (3.2.22) of I in the fifth line
to recast the relevant combination of calibrations in the final form shown. \ceqref{whgsts13} follows now readily
from applying lemma \cref{lemma:whslm2}.
\end{proof}

The qudit calibrated hypergraph operators are compatible also with the monadic multiplicative structure of the 
$\varOmega$ monad $G_C\varOmega$ (cf. subsect. 3.2 of I).


\begin{prop} \label{prop:jntdhr}
Let $l,m\in\bbN$ and let $(H,\varrho)\in G_C[l]$, $(K,\varsigma)\in G_C[m]$ be calibrated hypergraphs. Then,
one has \hphantom{xxxxxxxxxxxxx}  
\begin{equation}
\label{dhjntdh1}
D_{(H,\varrho)\smallsmile(K,\varsigma)}=D_{(H,\varrho)}\otimes D_{(K,\varsigma)}.
\end{equation}
Further, for the special calibrated hypergraph $(O,\varepsilon)\in G_C[0]$, 
\begin{equation}
\label{dhjntdh2}
D_{(O,\varepsilon)}=1_0.  
\end{equation}
\end{prop}

\begin{proof}
From expression \ceqref{whgsts10}, recalling relation \ceqref{four5}, we find 
{\allowdisplaybreaks
\begin{align}
\label{whgsts18p2}
D_{(H,\varrho)}\otimes D_{(K,\varsigma)}
&=\mycom{{}_\sss}{{}_{x\in E[l],y\in E[m]}}F_l{}^+\ket{x}
\,\omega^{\sigma_{(H,\varrho)}(x)}\hfpt\bra{x}F_l\otimes F_m{}^+\ket{y}\,\omega^{\sigma_{(K,\varsigma)}(y)}\hfpt\bra{y}F_m
\\
\nonumber
&=\mycom{{}_\sss}{{}_{x\in E[l],y\in E[m]}}F_l{}^+\otimes F_m{}^+\ket{x}\otimes\ket{y}
\,\omega^{\sigma_{(H,\varrho)}(x)+\sigma_{(K,\varsigma)}(y)}\hfpt\bra{x}\otimes\bra{y}F_l\otimes F_m
\\
\nonumber
&=\mycom{{}_\sss}{{}_{z\in E[l+m]}}F_{l+m}{}^+\ket{z}\,\omega^{\sigma_{(H,\varrho)}(P_l(z))+\sigma_{(K,\varsigma)}(P'{}_m(z))}\hfpt\bra{z}F_{l+m},
\end{align}
}
\!\!where $P_l:E[l+m]\rightarrow E[l]$, $P'{}_m:E[l+m]\rightarrow E[m]$ are given by
$P_l(z)_r=z_r$, $r\in[l]$, and $P'{}_m(z)_s=z_{s+l}$, $s\in[m]$, for $z\in E[l+m]$. 
$\sigma_{(H,\varrho)}(P_l(z))$, $\sigma_{(K,\varsigma)}(P'{}_m(z))$ read as
{\allowdisplaybreaks
\begin{align}
\label{whgsts18p3}
\sigma_{(H,\varrho)}(P_l(z))&=\mycom{{}_\sss}{{}_{X\in H}}\mycom{{}_\sss}{{}_{w\in\msA^X}}\varrho_X(w)
\tr\left(\mycom{{}_\ppp}{{}_{r\in X}}P_l(z)_r{}^{w(r)}\right)
\\
\nonumber
&\hspace{3.5cm}=\mycom{{}_\sss}{{}_{X\in H}}\mycom{{}_\sss}{{}_{u\in\msA^X}}\varrho_X(u)
\tr\left(\mycom{{}_\ppp}{{}_{t\in X}}z_t{}^{u(t)}\right),
\\
\label{whgsts18p4}
\sigma_{(K,\varsigma)}(P'{}_m(z))
&=\mycom{{}_\sss}{{}_{X\in K}}\mycom{{}_\sss}{{}_{v\in\msA^X}}\varsigma_X(v)
\tr\left(\mycom{{}_\ppp}{{}_{s\in X}}P'{}_m(z)_s{}^{v(s)}\right)
\\
\nonumber
&=\mycom{{}_\sss}{{}_{X\in K+l}}\mycom{{}_\sss}{{}_{v\in\msA^{X-l}}}\varsigma_{X-l}(v)
\tr\left(\mycom{{}_\ppp}{{}_{s\in X-l}}z_{s+l}{}^{v(s)}\right)
\\
\nonumber
&=\mycom{{}_\sss}{{}_{X\in K+l}}\mycom{{}_\sss}{{}_{u\in\msA^X}}\varsigma_{X-l}(t_{Xl\star}(u))
\tr\left(\mycom{{}_\ppp}{{}_{s\in X}}z_s{}^{t_{Xl\star}(u)(s-l)}\right)
\\
\nonumber
&\hspace{3.5cm}=\mycom{{}_\sss}{{}_{X\in K+l}}\mycom{{}_\sss}{{}_{u\in\msA^X}}\varsigma_{X-l}\circ t_{Xl\star}(u)
\tr\left(\mycom{{}_\ppp}{{}_{t\in X}}z_t{}^{u(t)}\right),
\end{align}
}
\!\!where we expressed each $v\in\msA^{X-l}$ as $v=t_{Xl\star}(u)$ for a unique $u\in\msA^X$ (cf.
subsect. 3.2 and eq. (3.2.38) of I) and used that $t_{Xl\star}(u)(t)=u(t+l)$ for $t\in X$. 
Now, by (3.1.4) of I, $H\smallsmile K=H\cup(K+l)$, where $H\cap(K+l)=\emptyset$. 
So, recalling (3.2.41) of I, we have 
{\allowdisplaybreaks
\begin{align}
\label{whgsts18p5}
\sigma_{(H,\varrho)}(P_l(z))+\sigma_{(K,\varsigma)}(P'{}_m(z))
&=\mycom{{}_\sss}{{}_{X\in H\smallsmile K}}\mycom{{}_\sss}{{}_{u\in\msA^X}}(\varrho\smallsmile\varsigma)_X(u)
\tr\left(\mycom{{}_\ppp}{{}_{t\in X}}z_t{}^{u(t)}\right)
\\
\nonumber
&=\sigma_{(H\smallsmile K,\varrho\smallsmile\varsigma)}(z)=\sigma_{(H,\varrho)\smallsmile(K,\varsigma)}(z),
\end{align}
}
\!\!where (3.2.52) of I was used in the last step. 
Inserting \ceqref{whgsts18p5} into \ceqref{whgsts18p2}, we find
\begin{equation}
\label{whgsts18p6}
D_{(H,\varrho)}\otimes D_{(K,\varsigma)}
=\mycom{{}_\sss}{{}_{z\in E[l+m]}}F_{l+m}{}^+\ket{z}\,\omega^{\sigma_{(H,\varrho)\smallsmile(K,\varsigma)}(z)}\hfpt\bra{z}F_{l+m}
=D_{(H,\varrho)\smallsmile(K,\varsigma)}, 
\end{equation}
showing \ceqref{dhjntdh1}.

For the hypergraph $(O,\varepsilon)$, we have
\begin{equation}
\label{whgsts18p8}
\sigma_{(O,\varepsilon)}=0.
\end{equation}
It follows from \ceqref{whgsts18p8} and the identity $F_0=1_0$ that
\begin{equation}
\label{whgsts18p10}
D_{(O,\varepsilon)}=F_0{}^+\ket{0}\,\omega^{\sigma_{(O,\varepsilon)}}\hfpt\bra{0}F_0=1_0, 
\end{equation}
showing \ceqref{dhjntdh2}. 
\end{proof}

We are now ready to define qudit calibrated hypergraph states.

\begin{defi} \label{def:whgst}
Let $l\in\bbN$. Further, let $(H,\varrho)\in G_C[l]$ be a calibrated hypergraph. The qudit calibrated hypergraph
state of $(H,\varrho)$ is the ket $\ket{(H,\varrho)}\in\scH_E[l]$ given by
\begin{equation}
\label{whgsts14}
\ket{(H,\varrho)}=D_{(H,\varrho)}\ket{0_l}.
\end{equation}
\end{defi}

\noindent
We notice that, by virtue of \ceqref{whgsts10}, \ceqref{whgsts9}, $D_{(H,\varrho)}$
factorizes in a product of commuting terms, one
for each pair formed by a hyperedge $X$ of $H$ and an exponent function $w$ of $X$.
The above definition therefore generalizes
various definitions of hypergraph states which have appeared in the literature
\ccite{Helwig:2013amq,Keet:2010qss,Steinhoff:2016:qhs,Xiong:2017qhp},
but does so in a rather non trivial manner. 


\begin{prop}
Let $l\in\bbN$ and let $(H,\varrho)\in G_C[l]$ be a calibrated hypergraph. Then, $\ket{(H,\varrho)}$ is normalized. 
\end{prop}

\begin{proof}
This follows immediately from the unitarity of $D_{(H,\varrho)}$, see prop. \cref{prop:dhuni},
and the normalization of $\ket{0_l}$.
\end{proof}

The calibrated hypergraph states function covariantly under the joint morphism actions of the calibrated
hypergraph and multi qudit state $\varOmega$ monads $G_C\varOmega$ and $\scH_E\varOmega$, a property
that indeed determines to a considerable extent their formal structure.


\begin{prop} \label{prop:hefgcf}
Let $l,m\in\bbN$ and let $f\in\Hom_\varOmega([l],[m])$ be a morphism. Moreover,
let $(H,\varrho)\in G_C[l]$ be a calibrated hypergraph.
Then, one has  \hphantom{xxxxxxx}
\begin{equation}
\label{whgsts16}
\scH_Ef\ket{(H,\varrho)}=\ket{G_Cf(H,\varrho)}.
\end{equation}
\end{prop}

\begin{proof} Using expression (4.3.1) of I, we find that 
\begin{equation}
\label{whgsts15} 
\scH_Ef\ket{0_l}=\ket{Ef(0_l)}=\ket{0_m}.
\end{equation}
From \ceqref{whgsts14}, combining \ceqref{whgsts13} and \ceqref{whgsts15}, we obtain 
{\allowdisplaybreaks 
\begin{align}
\label{}
\scH_Ef\ket{(H,\varrho)}&=\scH_EfD_{(H,\varrho)}\ket{0_l}
\\
\nonumber
&=D_{G_Cf(H,\varrho)}\scH_Ef\ket{0_l}
\\
\nonumber
&=D_{G_Cf(H,\varrho)}\ket{0_m}=\ket{G_Cf(H,\varrho)},
\end{align}
}
\!\!showing \ceqref{whgsts16}.
\end{proof}

\begin{cor}
Let $l\in\bbN$ and let $h\in\Hom_\varOmega([l],[l])$ be a bijective morphism. Further, let $(H,\varrho)\in G_C[l]$ be a
calibrated hypergraph such that $G_Ch(H,\varrho)=(H,\varrho)$. Then, 
\begin{equation}
\label{whgsts17}
\scH_Eh\ket{(H,\varrho)}=\ket{(H,\varrho)}.
\end{equation}
\end{cor}

\noindent
As $G_Ch(H,\varrho)=(H,\varrho)$, $G_Ch$ represents a symmetry of $(H,\varrho)$. By the above corollary,
$\scH_Eh$ is then a symmetry of $\ket{(H,\varrho)}$. Note that $\scH_Eh$ is unitary by prop. 4.3.1 of I
as appropriate for a quantum symmetry. 

In addition,  calibrated hypergraph states behave compatibly with respect to the monadic multiplications
of the $\varOmega$ monads $G_C\varOmega$ and $\scH_E\varOmega$ 
as summarized by the following proposition. 

\begin{prop}
Let $l,m\in\bbN$ and let $(H,\varrho)\in G_C[l]$, $(K,\varsigma)\in G_C[m]$ be calibrated hypergraphs. Then,
\hphantom{xxxxxxxxxxxxx}  
\begin{equation}
\label{whgsts18}
\ket{(H,\varrho)\smallsmile(K,\varsigma)}=\ket{(H,\varrho)}\smallsmile\ket{(K,\varsigma)}.
\end{equation}
Further, it holds that \hphantom{xxxxxxxxxxxxxxxxxxx}
\begin{equation}
\label{whgsts19}
\ket{(O,\varepsilon)}=\ket{0}.
\end{equation}
\end{prop}

\begin{proof}
Recall $\ket{(H,\varrho)}\smallsmile\ket{(K,\varsigma)}=\ket{(H,\varrho)}\otimes\ket{(K,\varsigma)}$ by (4.2.11) of I.
By \ceqref{whgsts14}, so, 
\begin{equation}
\label{whgsts18p0}
\ket{(H,\varrho)}\smallsmile\ket{(K,\varsigma)}=D_{(H,\varrho)}\ket{0_l}\otimes D_{(K,\varsigma)}\ket{0_m}
=D_{(H,\varrho)}\otimes D_{(K,\varsigma)}\ket{0_l}\otimes\ket{0_m}
\end{equation}
By relation \ceqref{dhjntdh1}, the operator $D_{(H,\varrho)}\otimes D_{(K,\varsigma)}$
in the right hand side is given by
\begin{equation}
\label{whgsts18p0/1}
D_{(H,\varrho)}\otimes D_{(K,\varsigma)}=D_{(H,\varrho)\smallsmile(K,\varsigma)}.
\end{equation}
We notice in addition that, from  relation (4.2.12) of I, 
\begin{align}
\label{whgsts18p7/1}
\ket{0_l}\otimes\ket{0_m}&=\ket{0_l}\smallsmile\ket{0_m}
\\
\nonumber
&=\ket{0_l\smallsmile 0_m}=\ket{0_{l+m}}. 
\end{align}
Substituting \ceqref{whgsts18p0/1}, \ceqref{whgsts18p7/1} into \ceqref{whgsts18p0}, we find 
\begin{equation}
\label{whgsts18p7/2}
\ket{(H,\varrho)}\smallsmile\ket{(K,\varsigma)}
=D_{(H,\varrho)\smallsmile(K,\varsigma)}\ket{0_{l+m}}
=\ket{(H,\varrho)\smallsmile(K,\varsigma)},
\end{equation}
leading to \ceqref{whgsts18}. 

From \ceqref{whgsts14} again, using \ceqref{dhjntdh2}, we find 
\begin{equation}
\label{whgsts18p11}
\ket{(O,\varepsilon)}=D_{(O,\varepsilon)}\ket{0}=\ket{0}, 
\end{equation}
as stated in \ceqref{whgsts19}. 
\end{proof}



The analysis carried out up this point shows that the calibrated hypergraph state construction 
 enjoys a number of nice properties. These properties can be elegantly
characterized as the calibrated hypergraph state map $\ket{-}$
being a morphism of the calibrated hypergraph and multi qudit state
$\varOmega$ monads $G_C\varOmega$ and $\scH_E\varOmega$ (cf. defs. 2.3.1, 2.3.2
and props. 3.2.9 and 4.2.3 of I). 
The next proposition formulates the above in precise terms. 

\begin{prop} \label{prop:whsmapmor}
The specification for each $l\in\bbN$ of the function $\ket{-}:G_C[l]\rightarrow \scH_E[l]$
assigning to every calibrated hypergraph $(H,\varrho)\in G_C[l]$ its corresponding
hypergraph state $\ket{(H,\varrho)}\in\scH_E[l]$
defines a distinguished morphism $\ket{-}\in\Hom_{\ul{\rm GM}_\varOmega}(G_C\varOmega,\scH_E\varOmega)$
of the objects $G_C\varOmega$, $\scH_E\varOmega\in\Obj_{\ul{\rm GM}_\varOmega}$ in $\ul{\rm GM}_\varOmega$.
\end{prop}

\begin{proof}
The proof consists in showing that the functions $\ket{-}$ satisfy relations 
(2.3.4)--(2.3.6) of I in accordance with def. 2.3.2 of I.
(2.3.4)--(2.3.6) of I indeed are fulfilled on account of relations
\ceqref{whgsts16}, \ceqref{whgsts18}, \ceqref{whgsts19}, respectively.   
\end{proof}

The following examples serve as an illustration of the theory on one hand and
make the reach of its scope manifest on the other. 

\begin{exa} \label{exa:bell} The Bell states. 
{\rm The familiar Bell states are simple instances of 2--qubit calibrated hypergraph states
as we show now. 
Since one works with qubits, the relevant Galois ring is the field $\msR=\GR(2,1)=\bbF_2$.
Its cyclicity monoid $\msZ$ consists of the cyclic monoids $\msZ_0=\bbH_{1,1}$, $\msZ_1=\bbH_{0,0}$
and so has cardinality $2$ (cf. subsect. \cref{subsec:cycmongal}).
We have further $\msM=\msP=\msR$.  
Consider the calibrated hypergraphs $(H,\varrho_{a_0,a_1})\in G_C[2]$, $a_0,a_1\in\bbF_2$,
so defined. $H=\{X\}$ where $X=\{0,1\}$ and $\varrho_{a_0,a_1}=\{\varrho_{a_0,a_1X}\}$ where
\begin{equation}
\label{bell1}
\varrho_{a_0,a_1X}=(a_0,a_1,1,0)
\end{equation}
with respect to the indexing  $\msA^X=(w_0,\ldots,w_3)$ of $\msA^X$ specified by 
{\allowdisplaybreaks
\begin{align}
\label{bell2}
&w_0=((1,0),(0,0)), &w_1=((0,0),(1,0)),
\\
\nonumber
&w_2=((1,0),(1,0)), &w_3=((0,0),(0,0)),
\end{align}
}
\!\!the inner pairs being elements of $\msZ$. 
The calibrated hypergraph states $\ket{(H,\varrho_{a_0,a_1})}$ can be computed straightforwardly \pagebreak 
using the general expressions \ceqref{whgsts10}, \ceqref{whgsts9} and \ceqref{whgsts14}. 
Recalling that $x^{(0,0)}=1$ and $x^{(1,0)}=x$ for $x\in\msR$, the phase functions $\sigma_{(H,\varrho_{a_0,a_1})}$
can be expressed as
\begin{equation}
\label{bell3}
\sigma_{(H,\varrho_{a_0,a_1})}(x_0,x_1)=a_0x_0+a_1x_1+x_0x_1
\end{equation}
with $x_0,x_1\in\msR$. 
The resulting calibrated hypergraph states are
\begin{equation}
\label{bell4}
\ket{(H,\varrho_{a_0,a_1})}=\mycom{{}_\sss}{{}_{x_0,x_1\in\msR}}\ket{x_0,x_1}\,2^{-1}(-1)^{\sigma_{(H,\varrho_{a_0,a_1})}(x_0,x_1)},
\end{equation}
where the kets $\ket{x_0,x_1}$ shown are those of the 2--qubit computational basis. Explicitly, 
{\allowdisplaybreaks
\begin{align}
\label{bell5}
\ket{(H,\varrho_{00})}&=\big(\ket{00}+\ket{01}+\ket{10}-\ket{11}\big)2^{-1},
\\
\nonumber
\ket{(H,\varrho_{01})}&=\big(\ket{00}-\ket{01}+\ket{10}+\ket{11}\big)2^{-1},
\\
\nonumber
\ket{(H,\varrho_{10})}&=\big(\ket{00}+\ket{01}-\ket{10}+\ket{11}\big)2^{-1},
\\
\nonumber
\ket{(H,\varrho_{11})}&=\big(\ket{00}-\ket{01}-\ket{10}-\ket{11}\big)2^{-1}.
\end{align}
}
\!\!As is well--known, these states can be
turned into the standard Bell states $\ket{\varPhi_+}$, $\ket{\varPhi_-}$, $\ket{\varPsi_+}$, $\ket{\varPsi_-}$
respectively up to a sign by acting with the local unitary operator $F_1\otimes 1_1$. 
}
\end{exa}

\vspace{-1mm}

\begin{exa} \label{exa:3qtchg} Three qutrit calibrated hypergraph states.
{\rm In ref. \!\!\ccite{Giri:2024qtt}, a teleportation protocol for qutrit states in a noisy environment
is proposed, that leverages five 3--qutrit hypergraph states as quantum channels. 
These states are in fact examples of calibrated hypergraph states.
As qutrits are considered, the proper Galois ring is the field $\msR=\GR(3,1)=\bbF_3$.
Its cyclicity monoid $\msZ$ consists of the cyclic monoids $\msZ_0=\bbH_{1,1}$, $\msZ_1=\bbH_{0,0}$, $\msZ_1=\bbH_{0,2}$
and so has cardinality $4$ (cf. subsect. \cref{subsec:cycmongal}).
We have  further $\msM=\msP=\msR$.
The relevant calibrated 3--qutrits hypergraphs are all constructed out of the following elements. 
The basic hyperedges are 
$X^0=\{0,1\}$, $X^1=\{1,2\}$, $X^2=\{0,2\}$ and $X^3=\{0,1,2\}$. The basic calibrations 
$\varrho^0,\varrho^1,\varrho^2,\varrho^3$ of $X^0,X^1,X^2,X^3$ are given by 
\begin{align}
\label{3qtchgp1}
&\varrho^0=(1,2,0,\ldots,0), &\varrho^1=(1,2,0,\ldots,0), &&\varrho^2=(1,2,0,\ldots,0),
\\
\nonumber
&\varrho^2=(1,2,2,1,0,\ldots,0)
\end{align}
for the indexings $\msA^{X^0}=(w^0{}_0,\ldots,w^0{}_{15})$, $\msA^{X^1}=(w^1{}_0,\ldots,w^1{}_{15})$, 
$\msA^{X^2}=(w^2{}_0,\ldots,w^2{}_{15})$, $\msA^{X^3}=(w^3{}_0,\ldots,w^3{}_{63})$ of $\msA^{X^0}$, $\msA^{X^1}$, 
$\msA^{X^2}$, $\msA^{X^3}$ for which \pagebreak 
{\allowdisplaybreaks
\begin{align}
\label{3qtchgp2}
&w^0{}_0=((0,0,1),(1,0,1)), &&w^0{}_1=((0,0,0),(1,0,1)),
\\
\nonumber
&w^1{}_0=((0,0,1),(1,0,1)), &&w^1{}_1=((0,0,0),(1,0,1)), 
\\
\nonumber
&w^2{}_0=((1,0,1),(0,0,1)), &&w^2{}_1=((1,0,1),(0,0,0)), 
\\
\nonumber
&w^3{}_0=((0,0,1),(0,0,1),(1,0,1)), &&w^3{}_1=((0,0,0),(0,0,1),(1,0,1)),
\\
\nonumber
&w^3{}_2=((0,0,1),(0,0,0),(1,0,1)), &&w^3{}_3=((0,0,0),(0,0,0),(1,0,1)), 
\end{align}
}
\!\!the inner triples being all elements of $\msZ$. Consider now the calibrated hypergraphs 
$(H^i,\varrho^i)\in G_C[3]$, $i=\sfa,\sfb,\sfc,\sfd,\sfe$, specified by 
{\allowdisplaybreaks
\begin{align}
\label{3qtchgp3}
&H^\sfa=\{X^3\},&&\varrho^\sfa=\{\varrho^3\},
\\
\nonumber
&H^\sfb=\{X^0,X^1\},&&\varrho^\sfb=\{\varrho^0,\varrho^1\},
\\
\nonumber 
&H^\sfc=\{X^0,X^1,X^2\},&&\varrho^\sfc=\{\varrho^0,\varrho^1,\varrho^2\},
\\
\nonumber
&H^\sfd=\{X^0,X^1,X^3\},&&\varrho^\sfd=\{\varrho^0,\varrho^1,\varrho^3\},
\\
\nonumber
&H^\sfe=\{X^0,X^1,X^2,X^3\},&&\varrho^\sfe=\{\varrho^0,\varrho^1,\varrho^2,\varrho^3\}. 
\end{align}
}
\!\!The associated calibrated hypergraph states $\ket{(H^i,\varrho^i)}$ can be computed using 
the general expressions \ceqref{whgsts10}, \ceqref{whgsts9} and \ceqref{whgsts14}.
Their phase functions $\sigma_{(H^i,\varrho^i)}$ read as 
{\allowdisplaybreaks
\begin{align}
\label{3qtchgp4}
\sigma_{(H^\sfa,\varrho^\sfa)}(x_0,x_1,x_2)&=x_0{}^{(0,0,1)}x_1{}^{(0,0,1)}x_2{}^{(1,0,1)}
\\
\nonumber
&\hspace{2.7cm}+2x_0{}^{(0,0,1)}x_2{}^{(1,0,1)}+2x_1{}^{(0,0,1)}x_2{}^{(1,0,1)}+x_2{}^{(1,0,1)},
\\
\nonumber
\sigma_{(H^\sfb,\varrho^\sfb)}(x_0,x_1,x_2)
&=x_0{}^{(0,0,1)}x_1{}^{(1,0,1)}+x_1{}^{(0,0,1)}x_2{}^{(1,0,1)}+2x_1{}^{(1,0,1)}+2x_2{}^{(1,0,1)},
\\
\nonumber
\sigma_{(H^\sfc,\varrho^\sfc)}(x_0,x_1,x_2)&=x_0{}^{(0,0,1)}x_1{}^{(1,0,1)}
+x_1{}^{(0,0,1)}x_2{}^{(1,0,1)}+x_2{}^{(0,0,1)}x_0{}^{(1,0,1)}
\\
\nonumber
&\hspace{5.cm}+2x_1{}^{(1,0,1)}+2x_2{}^{(1,0,1)}+2x_0{}^{(1,0,1)},
\\
\nonumber
\sigma_{(H^\sfd,\varrho^\sfd)}(x_0,x_1,x_2)&=x_0{}^{(0,0,1)}x_1{}^{(0,0,1)}x_2{}^{(1,0,1)}+x_0{}^{(0,0,1)}x_1{}^{(1,0,1)}
\\
\nonumber
&\hspace{5.7cm}+2x_0{}^{(0,0,1)}x_2{}^{(1,0,1)}+2x_1{}^{(1,0,1)},
\\
\nonumber
\sigma_{(H^\sfe,\varrho^\sfe)}(x_0,x_1,x_2)&=x_0{}^{(0,0,1)}x_1{}^{(0,0,1)}x_2{}^{(1,0,1)}+x_0{}^{(0,0,1)}x_1{}^{(1,0,1)}
\\
\nonumber
&\hspace{.6cm}+x_0{}^{(1,0,1)}x_2{}^{(0,0,1)}+2x_0{}^{(0,0,1)}x_2{}^{(1,0,1)}+2x_0{}^{(1,0,1)}+2x_1{}^{(1,0,1)},
\end{align}
}
\!\!where $x_0,x_1,x_2\in\msR$. The states $\ket{(H^i,\varrho^i)}$ are given by
\begin{equation}
\label{3qtchgp4/1}
\ket{(H^i,\varrho^i)}=\mycom{{}_\sss}{{}_{x_0,x_1,x_2\in\msR}}\ket{x_0,x_1,x_2}\,3^{-3/2}\omega^{\sigma_{(H^i,\varrho^i)}(x_0,x_1,x_2)},
\end{equation}
where $\omega=\exp(2\pi i/3)$ and
the kets $\ket{x_0,x_1,x_2}$ shown belong to the 3--qutrit computational basis.
These are precisely the five hypergraph states studied in \ccite{Giri:2024qtt}.
We provide the explicit expressions of the states 
$|\ket{(H^\sfc,\varrho^\sfc)}$, $\ket{(H^\sfe,\varrho^\sfe)}$,
{\allowdisplaybreaks
\begin{align}
\label{3qtchgp5}
|(H^\sfc,\varrho^\sfc)\rangle&=
\big(|000\rangle+|001\rangle+|002\rangle+|010\rangle+|011\rangle+|012\rangle+|020\rangle+|021\rangle\,\omega
\\
\nonumber
&+|022\rangle\,\omega^2+|100\rangle+|101\rangle+|102\rangle\,\omega+|110\rangle+|111\rangle+|112\rangle\,\omega+|120\rangle
\\
\nonumber
&+|121\rangle\,\omega+|122\rangle+|200\rangle+|201\rangle+|202\rangle\,\omega^2
+|210\rangle\,\omega+|211\rangle\,\omega+|212\rangle
\\
\nonumber
&+|220\rangle\,\omega^2+|221\rangle+|222\rangle\big)3^{-3/2},
\\
\nonumber
|(H^\sfe,\varrho^\sfe)\rangle&=
\big(|000\rangle+|001\rangle+|002\rangle+|010\rangle+|011\rangle+|012\rangle+|020\rangle+|021\rangle\,\omega
\\
\nonumber
&+|022\rangle\,\omega^2+|100\rangle+|101\rangle+|102\rangle\,\omega+|110\rangle+|111\rangle+|112\rangle\,\omega+|120\rangle
\\
\nonumber
&+|121\rangle\,\omega+|122\rangle+|200\rangle+|201\rangle+|202\rangle\,\omega^2
+|210\rangle\,\omega+|211\rangle\,\omega+|212\rangle
\\
\nonumber
&+|220\rangle\,\omega^2+|221\rangle\,\omega+|222\rangle\,\omega^2\big)3^{-3/2}.
\end{align}
}
\!\!The states $\ket{(H^i,\varrho^i)}$ are important for us, because they are instances of calibrated hypergraph states which
are not weighted, as will be shown later in subsect. \cref{subsec:tech}. 
}
\end{exa}


\subsection{\textcolor{blue}{\sffamily Calibrated hypergraph states as stabilizer states}}\label{subsec:whgstab}

Calibrated hypergraph states turn out to be stabilizer states, that is one dimensional stabilizer
error correcting codes, just as other hypergraph states \ccite{Gottesman:1997scq,Garcia:2014gss}.
In this subsection, we shall prove such a property and analyze its mathematical ramifications. 

\begin{defi}
Let $l\in\bbN$ and let $(H,\varrho)\in G_C[l]$ be a calibrated hypergraph. The stabilizer group operators associated with
$(H,\varrho)$ are the operators $K_{(H,\varrho)}(a)\in\End_{\bfsfH}(\scH_e[l])$, $a\in E[l]$, given by the formula
\begin{equation}
\label{whggen0}
K_{(H,\varrho)}(a)=D_{(H,\varrho)}X_l(a)D_{(H,\varrho)}{}^+,
\end{equation}
where the operator $X_l(a)$ is given by \ceqref{qdpauli2}. 
\end{defi}

\noindent The next proposition provides another expression of such operators. 

\begin{prop} \label{prop:altexprkh}
Let $l\in\bbN$ and let $(H,\varrho)\in G_C[l]$ be a calibrated hypergraph. Then,
\begin{equation}
\label{whggen1}
K_{(H,\varrho)}(a)=X_l(a)L_{(H,\varrho)}(a)
\end{equation}
for $a\in E[l]$, where $L_{(H,\varrho)}(a)\in\End_{\bfsfH}(\scH_E[l])$ is the operator 
\begin{equation}
\label{whggen1/1}
L_{(H,\varrho)}(a)
=\mycom{{}_\sss}{{}_{x\in E[l]}}F_l{}^+\ket{x}\,\omega^{\sigma_{(H,\varrho)}(x-a)-\sigma_{(H,\varrho)}(x)}\hfpt\bra{x}F_l,
\end{equation}
the phase function $\sigma_{(H,\varrho)}$ being given by \ceqref{whgsts9}.
\end{prop}

\begin{proof}
Inserting the combination of \ceqref{qdpauli1} and \ceqref{four4} and \ceqref{whgsts10} into \ceqref{whggen0}, we get
{\allowdisplaybreaks
\begin{align}
\label{whggen3p2}
K_{(H,\varrho)}(a)&=
\mycom{{}_\sss}{{}_{x,y,z\in E[l]}}F_l{}^+\ket{x}\,\omega^{\sigma_{(H,\varrho)}(x)}\hfpt\bra{x}F_l
F_l{}^+\ket{z-a}\bra{z}F_l F_l{}^+\ket{y}\,\omega^{-\sigma_{(H,\varrho)}(y)}\hfpt\bra{y}F_l
\\
\nonumber
&=\mycom{{}_\sss}{{}_{x,y,z\in E[l]}}F_l{}^+\ket{x}\,\omega^{\sigma_{(H,\varrho)}(x)-\sigma_{(H,\varrho)}(y)}
\hfpt\braket{x}{z-a}\hfpt\braket{z}{y}\bra{y}F_l
\\
\nonumber
&=\mycom{{}_\sss}{{}_{x,y,z\in E[l]}}F_l{}^+\ket{x}\,\omega^{\sigma_{(H,\varrho)}(x)-\sigma_{(H,\varrho)}(y)}\hfpt
\delta_{x,z-a}\hfpt\delta_{z,y}\hfpt\bra{y}F_l
\\
\nonumber
&=\mycom{{}_\sss}{{}_{y\in E[l]}}F_l{}^+\ket{y-a}\,\omega^{\sigma_{(H,\varrho)}(y-a)-\sigma_{(H,\varrho)}(y)}\hfpt\bra{y}F_l
\\
\nonumber
&=\mycom{{}_\sss}{{}_{y,z\in E[l]}}F_l{}^+\ket{z-a}\hfpt\delta_{z,y}\,\omega^{\sigma_{(H,\varrho)}(y-a)-\sigma_{(H,\varrho)}(y)}\hfpt\bra{y}F_l
\\
\nonumber
&=\mycom{{}_\sss}{{}_{y,z\in E[l]}}F_l{}^+\ket{z-a}\hfpt\bra{z}F_l{}^+ F_l\ket{y}\,
\omega^{\sigma_{(H,\varrho)}(y-a)-\sigma_{(H,\varrho)}(y)}\hfpt\bra{y}F_l
=X_l(a)L_{(H,\varrho)}(a).
\end{align}
}
\!\!This shows \ceqref{whggen1}.
\end{proof}

\noindent
We notice that, by \ceqref{whgsts9} and \ceqref{whggen1/1}, 
$L_{(H,\varrho)}$ factorizes in a product of commuting terms, one
for each pair formed by a hyperedge $X$ of $H$ and an exponent function $w$ of $X$. The above expression extends
in this way similar expressions of stabilizer group operators of
hypergraph states which have appeared in the literature. Its graph theoretic interpretation 
is not however as obvious. 

The stabilizer group operators of a calibrated hypergraph are pairwise distinct.

\begin{prop} \label{prop:khdiff}
Let $l\in\bbN$ and let $(H,\varrho)\in G_C[l]$ be a calibrated hypergraph. Then, for $a,b\in E[l]$,
$K_{(H,\varrho)}(a)=K_{(H,\varrho)}(b)$ if and only if $a=b$.
\end{prop}

\begin{proof}
By \ceqref{whggen0}, $K_{(H,\varrho)}(a)=K_{(H,\varrho)}(b)$ if and only if $X_l(a)=X_l(b)$.
In turn, as reviewed in subsect. \cref{subsec:qdpauli}, $X_l(a)=X_l(b)$ if and only if $a=b$.
The statement follows.   
\end{proof}

\noindent
The number of stabilizer group operators is therefore $q^l$. 

A basic property characterizing the stabilizer group operators of a calibrated hypergraph is their unitarity. 

\begin{prop} \label{prop:khuni}
Let $l\in\bbN$ and let $(H,\varrho)\in G_C[l]$ be a calibrated hypergraph. Then, for $a\in E[l]$
the operator $K_{(H,\varrho)}(a)$ is unitary, so that $K_{(H,\varrho)}(a)\in\msU(\scH_E[l])$. 
\end{prop}

\begin{proof}
This property is an immediate consequence of relation \ceqref{whggen0} and the unitarity of
$X_l(a)$ (cf. subsect. \cref{subsec:qdpauli}) and $D_{(H,\varrho)}$ (cf. prop. \cref{prop:dhuni}). 
\end{proof}

\noindent
What is more, the stabilizer group operators obey a set of algebraic identities  
implying that they form an Abelian subgroup of the group of unitaries. 

\begin{prop}
Let $l\in\bbN$ and let $(H,\varrho)\in G_C[l]$ be a calibrated hypergraph. Then,
{\allowdisplaybreaks
\begin{align}
\label{whggen4}
K_{(H,\varrho)}(a)K_{(H,\varrho)}(b)&=K_{(H,\varrho)}(a+b),
\\
\label{whggen5}
K_{(H,\varrho)}(a)^{-1}&=K_{(H,\varrho)}(-a),
\\
\label{whggen6}
K_{(H,\varrho)}(0_l)&=1_l
\end{align}
}
\!\!for $a,b\in E[l]$. Hence, $\{K_{(H,\varrho)}(a)|a\in E[l]\}$ is an Abelian subgroup of $\msU(\scH_E[l])$.
\end{prop}

\begin{proof}
By virtue of relation \ceqref{whggen0} and the unitarity of $D_{(H,\varrho)}$ (cf. prop. \cref{prop:dhuni}),
the identities \ceqref{whggen4}--\ceqref{whggen6} are immediate consequences of relations 
\ceqref{qdpauli6}--\ceqref{qdpauli8}.  
\end{proof}

\begin{defi}
The group $\msS_{(H,\varrho)}=\{K_{(H,\varrho)}(a)\hfpt|\hfpt a\in E[l]\}\subset\msU(\scH_E[l])$
is the stabilizer group of $\ket{(H,\varrho)}$. 
\end{defi}

The stabilizer group operators cohere with the morphism structures of the calibrated hypergraph 
and multi qudit state $\varOmega$ monads $G_C\varOmega$ and $\scH_E\varOmega$
(cf. subsects. 3.2, 4.2 of I) in the manner formulated in the next proposition.

\begin{prop} \label{prop:catkhr1}
Let $l,m\in\bbN$ and let $f\in\Hom_\varOmega([l],[m])$ be a morphism. Let further
$(H,\varrho)\in G_C[l]$ be a calibrated hypergraph. Then, for all $a\in E[l]$ one has
\begin{equation}
\label{catkhr1}
\scH_EfK_{(H,\varrho)}(a)\scH_Ef^+=q^{l-m}\mycom{{}_\sss}{{}_{b\in E[m],Ef^t(b)=a}}K_{G_Cf(H,\varrho)}(b).
\end{equation}
\end{prop}  

\begin{proof} 
The proof leverages the following lemma. 

\begin{lemma}
Let $l,m\in\bbN$, $f\in\Hom_\varOmega([l],[m])$, Then, 
\begin{equation}
\label{catkhr2}
\scH_EfX_l(a)\scH_Ef^+=q^{l-m}\mycom{{}_\sss}{{}_{b\in E[m],Ef^t(b)=a}}X_m(b)
\end{equation}
for $a\in E[l]$
\end{lemma}

\begin{proof}
From (4.3.1) of I and \ceqref{qdpauli2}, 
{\allowdisplaybreaks
\begin{align}
\label{catkhr2p1}
\scH_EfX_l(a)\scH_Ef^+&=\mycom{{}_\sss}{{}_{x,y,z\in E[l]}}\ket{Ef(x)}\braket{x}{z}\,\omega^{\langle a,z\rangle}
\hfpt\braket{z}{y}\bra{Ef(y)}
\\
\nonumber
&=\mycom{{}_\sss}{{}_{x,y,z\in E[l]}}\ket{Ef(x)}\delta_{x,z}\,\omega^{\langle a,z\rangle}\hfpt\delta_{z,y}\bra{Ef(y)}
\\
\nonumber
&=\mycom{{}_\sss}{{}_{x\in E[l]}}\ket{Ef(x)}\,\omega^{\langle a,x\rangle}\hfpt\bra{Ef(x)}
=\mycom{{}_\sss}{{}_{x\in E[l],y\in E[m]}}\ket{y}\,\omega^{\langle a,x\rangle}\delta_{y,Ef(x)}\hfpt\bra{y}.
\end{align}
}
\!\!Using the character formula \ceqref{char}, we can recast the sum over $x\in E[l]$ as follows, 
{\allowdisplaybreaks
\begin{align}
\label{catkhr2p2}
\mycom{{}_\sss}{{}_{x\in E[l]}}\omega^{\langle a,x\rangle}\delta_{y,Ef(x)}
&=q^{-m}\mycom{{}_\sss}{{}_{x\in E[l],b\in E[m]}}\omega^{\langle a,x\rangle}\omega^{\langle y-Ef(x),b\rangle}
\\
\nonumber
&=q^{-m}\mycom{{}_\sss}{{}_{x\in E[l],b\in E[m]}}\omega^{\langle a-Ef^t(b),x\rangle}\omega^{\langle b,y\rangle}
=q^{l-m}\mycom{{}_\sss}{{}_{b\in E[m]}}\omega^{\langle b,y\rangle}\delta_{a,Ef^t(b)}.
\end{align}
}
\!\!Inserting \ceqref{catkhr2p2} into \ceqref{catkhr2p1}, we obtain 
{\allowdisplaybreaks
\begin{align}
\label{catkhr2p3}
\scH_EfX_l(a)\scH_Ef^+&=\mycom{{}_\sss}{{}_{y,b\in E[m]}}\ket{y}\,q^{l-m}\,\omega^{\langle b,y\rangle}
\hfpt\delta_{a,Ef^t(b)}\hfpt\bra{y}
\\
\nonumber
&=q^{l-m}\mycom{{}_\sss}{{}_{b\in E[m]}}\delta_{a,Ef^t(b)}X_m(b)=q^{l-m}\mycom{{}_\sss}{{}_{b\in E[m],Ef^t(b)=a}}X_m(b),
\end{align}
}
\!\!proving formula \ceqref{catkhr2}. 
\end{proof}

\noindent
We can now show \ceqref{catkhr1}. By \ceqref{whgsts13} and \ceqref{catkhr2}, 
{\allowdisplaybreaks
\begin{align}
\label{}
\scH_EfK_{(H,\varrho)}(a)\scH_Ef^+&=\scH_EfD_{(H,\varrho)}X_l(a)D_{(H,\varrho)}{}^+\scH_Ef^+
\\
\nonumber
&=D_{G_Cf(H,\varrho)}\scH_EfX_l(a)\scH_Ef^+D_{G_Cf(H,\varrho)}{}^+
\\
\nonumber
&=q^{l-m}\mycom{{}_\sss}{{}_{b\in E[m],Ef^t(b)=a}}D_{G_Cf(H,\varrho)}X_m(b)D_{G_Cf(H,\varrho)}{}^+
\\
\nonumber
&\hspace{4.5cm}=q^{l-m}\mycom{{}_\sss}{{}_{b\in E[m],Ef^t(b)=a}}K_{G_Cf(H,\varrho)}(b),
\end{align}
}
\!\!which is \ceqref{catkhr1}. 
\end{proof}

\noindent 
When $l=m$ and $f$ is invertible, relation \ceqref{catkhr1} takes a simpler form, 
\begin{equation}
\label{catkhr3}
\scH_EfK_{(H,\varrho)}(a)\scH_Ef^+=K_{G_Cf(H,\varrho)}(Ef(a)), 
\end{equation}
as by (4.1.2) of I, \ceqref{qdinner3} $Ef^t(b)=a$ if and only if $Ef(a)=b$ for $a\in E[l]$, $b\in E[m]$. \pagebreak 


The stabilizer group operators are compatible also with the monadic multiplicative structure of the 
$\varOmega$ monad $G_C\varOmega$ (cf. subsect. 3.2 of I). 

\begin{prop}
Let $l,m\in\bbN$ and let $(H,\varrho)\in G_C[l]$, $(K,\varsigma)\in G_C[m]$ be calibrated hypergraphs. Then,
one has \hphantom{xxxxxxxxxxxxx}  
\begin{equation}
\label{catkhr4}
K_{(H,\varrho)\smallsmile(K,\varsigma)}(a\smallsmile b)=K_{(H,\varrho)}(a)\otimes K_{(K,\varsigma)}(b).
\end{equation}
Moreover, for the special calibrated hypergraph $(O,\varepsilon)\in G_C[0]$,
it holds that 
\begin{equation}
\label{catkhr5}
K_{(O,\varepsilon)}(0)=1_0. 
\end{equation}
\end{prop}  

\begin{proof}
From \ceqref{whggen0}, owing to \ceqref{qdpauli15} and \ceqref{dhjntdh1}, we find
{\allowdisplaybreaks
\begin{align}
\label{}
K_{(H,\varrho)\smallsmile(K,\varsigma)}(a\smallsmile b)
&=D_{(H,\varrho)\smallsmile(K,\varsigma)}X_{l+m}(a\smallsmile b)D_{(H,\varrho)\smallsmile(K,\varsigma)}{}^+
\\
\nonumber
&=D_{(H,\varrho)}\otimes D_{(K,\varsigma)}\,X_l(a)\otimes X_m(b)\,D_{(H,\varrho)}{}^+\otimes D_{(K,\varsigma)}{}^+
\\
\nonumber
&=D_{(H,\varrho)}X_l(a)D_{(H,\varrho)}{}^+\otimes D_{(K,\varsigma)}X_m(b)D_{(K,\varsigma)}{}^+
\\
\nonumber
&\hspace{6.5cm}=K_{(H,\varrho)}(a)\otimes K_{(K,\varsigma)}(b).
\end{align}
}
\!\!Relation \ceqref{catkhr4} is so proven. 

From \ceqref{whggen0}, using \ceqref{dhjntdh2} and recalling that $X_0(0)=1_0$, we find
\begin{equation}
\label{}
K_{(O,\varepsilon)}(0)=D_{(O,\varepsilon)}X_0(0)D_{(O,\varepsilon)}{}^+=1_0,
\end{equation}
showing also \ceqref{catkhr5}. 
\end{proof}



The stabilizer group of a given calibrated hypergraph state owes its name to its elements 
leaving invariant that state. 

\begin{prop}
Let $l\in\bbN$ and let $(H,\varrho)\in G_C[l]$ be a calibrated hypergraph. Then, 
\begin{equation}
\label{whggen3}
K_{(H,\varrho)}(a)\ket{(H,\varrho)}=\ket{(H,\varrho)} 
\end{equation}
for $a\in E[l]$.
\end{prop}

\begin{proof}
We begin by noting that by \ceqref{qdpauli2} $X_l(a)\ket{0_l}=\ket{0_l}\,\omega^{\langle a,0_l\rangle}=\ket{0_l}$. 
Using \ceqref{whgsts14}, \ceqref{whggen0} and the relation just recalled, we find 
{\allowdisplaybreaks
\begin{align}
\label{}
\ket{(H,\varrho)}&=D _{(H,\varrho)}\ket{0_l}
\\
\nonumber
&=D _{(H,\varrho)}X_l(a)\ket{0_l}
\\
\nonumber
&=K_{(H,\varrho)}(a)D_{(H,\varrho)}\ket{0_l}=K_{(H,\varrho)}(a)\ket{(H,\varrho)},
\end{align}
}
\!\!showing \ceqref{whggen3}.
\end{proof}

\noindent
As it contains $q^l=\dim\scH_E[l]$ elements, the stabilizer group of a hypergraph state
not only characterizes this latter but also determines  it uniquely up to a phase factor.

With each calibrated hypergraph there is associated a full $q^l$ element family 
of states containing the  associated hypergraph state. 

\begin{defi}
Let $l\in\bbN$. The calibrated hypergraph state basis associated with the calibrated hypergraph
$(H,\varrho)\in G_C[l]$ consists of the kets 
\begin{equation}
\label{whgstab1}
\ket{(H,\varrho),a}=D_{(H,\varrho)}\ket{a}
\end{equation}
with $a\in E[l]$. 
\end{defi}

\noindent
The name given to such a collection of states reflects indeed their being a distinguished orthonormal basis
of the ambient state Hilbert space. 

\begin{prop} \label{prop:hrkon}
Let $l\in\bbN$ and let $(H,\varrho)\in G_C[l]$ be a calibrated hypergraph. Then, the kets $\ket{(H,\varrho),a}$, $a\in E[l]$,
constitute an orthonormal basis of $\scH_E[l]$.
\end{prop}

\begin{proof}
By prop. \cref{prop:dhuni}, $D_{(H,\varrho)}$ is a unitary operator in $\scH_E[l]$.
Therefore, given that the kets $\ket{a}$, $a\in E[l]$, constitute an orthonormal basis of $\scH_E[l]$, the kets
$\ket{(H,\varrho),a}$, $a\in E[l]$, also do. 
\end{proof}


The hypergraph state $\ket{(H,\varrho)}$ is a kind of ground state generating all the
basis states $\ket{(H,\varrho),a}$ under the action of Pauli operators $Z_l(a)$ operating as
creation operators (cf. eq. \ceqref{qdpauli1}). In this perspective, the stabilizer group operators $K_{(H,\varrho)}(a)$
are to be regarded as a kind of discrete annihilation operators. 

\begin{prop}
Let $l\in\bbN$ and let $(H,\varrho)\in G_C[l]$ be a calibrated hypergraph. Then, 
\begin{equation}
\label{whgstab2}
\ket{(H,\varrho),a}=Z_l(a)\ket{(H,\varrho)}
\end{equation}
for $a\in E[l]$. 
\end{prop}

\begin{proof}
The kets $F_l{}^+\ket{x}$ with $x\in E[l]$ constitute an orthonormal basis of $\scH_E[l]$
by the unitarity of $F_l$. By \ceqref{qdpauli2}, \ceqref{four3} and \ceqref{whgsts10}, the operators
$Z_l(a)$ with $a\in E[l]$ and $D_{(H,\varrho)}$ are diagonal in such a basis.
Therefore, they commute 
\begin{equation}
\label{whgstab2p1}
Z_l(a)D_{(H,\varrho)}=D_{(H,\varrho)}Z_l(a).
\end{equation}  
We observe next that, on account of \ceqref{qdpauli1}, $Z_l(a)\ket{0_l}=\ket{0_l+a}=\ket{a}$. 
From this identity, recalling further \ceqref{whgsts14}, we find
{\allowdisplaybreaks
\begin{align}
\label{whgstab2p3}
Z_l(a)\ket{(H,\varrho)}&=Z_l(a)D_{(H,\varrho)}\ket{0_l}
\\
\nonumber
&=D_{(H,\varrho)}Z_l(a)\ket{0_l}
\\
\nonumber
&=D_{(H,\varrho)}\ket{a}=\ket{(H,\varrho),a}, 
\end{align}
}
\!\!as stated in \ceqref{whgstab2}. 
\end{proof}

The basis $\ket{(H,\varrho),a}$ consists of common eigenkets of the stabilizer group operators
$K_{(H,\varrho)}(a)$, allowing so for the simultaneous diagonalization of these latter. This property is expected 
from their commutativity. 

\begin{prop}
Let $l\in\bbN$ and let $(H,\varrho)\in G_C[l]$ be a calibrated hypergraph. Then, 
\begin{equation}
\label{whgstab3}
K_{(H,\varrho)}(a)\ket{(H,\varrho),b}=\ket{(H,\varrho),b}\,\omega^{\langle a,b\rangle}
\end{equation}
for $a,b\in E[l]$. Consequently, 
\begin{equation}
\label{whgstab4}
K_{(H,\varrho)}(a)=\mycom{{}_\sss}{{}_{b\in E[l]}}\ket{(H,\varrho),b}\,\omega^{\langle a,b\rangle}\hfpt\bra{(H,\varrho),b}.
\end{equation}
\end{prop}

\begin{proof}
As already noticed, the kets $F_l{}^+\ket{x}$ with $x\in E[l]$ constitute an orthonormal basis of $\scH_E[l]$.
By \ceqref{qdpauli2}, \ceqref{four3} and \ceqref{whggen1/1},  the operators $L_{(H,\varrho)}(a)$, $a\in E[l]$,
and $Z_l(b)$, $b\in E[l]$, are diagonal in such a basis. Therefore, they commute 
\begin{equation}
\label{whgstab3p1}
L_{(H,\varrho)}(a)Z_l(b)=Z_l(b)L_{(H,\varrho)}(a).
\end{equation}
From \ceqref{whggen1}, using \ceqref{qdpauli9} and \ceqref{whgstab3p1}, it follows that 
{\allowdisplaybreaks
\begin{align}
\label{whgstab3p3}
K_{(H,\varrho)}(a)Z_l(b)&=X_l(a)L_{(H,\varrho)}(a)Z_l(b)
\\
\nonumber
&=\omega^{\langle a,b\rangle}Z_l(b)X_l(a)L_{(H,\varrho)}(a)=\omega^{\langle a,b\rangle}Z_l(b)K_{(H,\varrho)}(a).
\end{align}
}
\!\! Using \ceqref{whggen3}, \ceqref{whgstab2} and \ceqref{whgstab3p3}, we obtain then
{\allowdisplaybreaks
\begin{align}
\label{whgstab3p4}
K_{(H,\varrho)}(a)\ket{(H,\varrho),b}&=K_{(H,\varrho)}(a)Z_l(b)\ket{(H,\varrho)}
\\
\nonumber
&=Z_l(b)K_{(H,\varrho)}(a)\ket{(H,\varrho)}\,\omega^{\langle a,b\rangle}
\\
\nonumber
&=Z_l(b)\ket{(H,\varrho)}\,\omega^{\langle a,b\rangle}=\ket{(H,\varrho),b}\,\omega^{\langle a,b\rangle},
\end{align}
}
\!\! 
showing \ceqref{whgstab3}. \ceqref{whgstab4} follows readily from the spectral theorem.
\end{proof}


The orthonormal basis states of a calibrated hypergraph function covariantly under the joint morphism actions of the calibrated
hypergraph and multi qudit state $\varOmega$ monads $G_C\varOmega$ and $\scH_E\varOmega$.

\begin{prop}
Let $l,m\in\bbN$ and let $f\in\Hom_\varOmega([l],[m])$ be a morphism. Moreover, let $(H,\varrho)\in G_C[l]$ be a calibrated hypergraph.
Then, 
\begin{equation}
\label{covkhr1}
\scH_Ef\ket{(H,\varrho),a}=\ket{G_Cf(H,\varrho),Ef(a)}
\end{equation}
for $a\in E[l]$.
\end{prop}

\begin{proof}
From \ceqref{whgstab1}, using relations (4.2.2) of I and \ceqref{whgsts13}, we find 
{\allowdisplaybreaks
\begin{align}
\label{}
\scH_Ef\ket{(H,\varrho),a}&=\scH_EfD_{(H,\varrho)}\ket{a}
\\
\nonumber
&=D_{G_Cf(H,\varrho)}\scH_Ef\ket{a}
\\
\nonumber
&=D_{G_Cf(H,\varrho)}\ket{Ef(a)}=\ket{G_Cf(H,\varrho),Ef(a)},
\end{align}
}
\!\!yielding the desired relation \ceqref{covkhr1}. 
\end{proof}

The orthonormal basis states of a hypergraph behave compatiblyly also with respect to the monadic multiplications
of the $\varOmega$ monads $G_C\varOmega$ and $\scH_E\varOmega$. 

\begin{prop}
Let $l,m\in\bbN$ and let $(H,\varrho)\in G_C[l]$, $(K,\varsigma)\in G_C[m]$ be calibrated hypergraphs. Then,
\begin{equation}
\label{covkhr2}
\ket{(H,\varrho)\smallsmile(K,\varsigma),a\smallsmile b}
=\ket{(H,\varrho),a}\smallsmile\ket{(K,\varsigma),b}.
\end{equation}
Further, it holds that \hphantom{xxxxxxxxxxxxxxxxxxx}
\begin{equation}
\label{covkhr3}
\ket{(O,\varepsilon),0}=\ket{0}.
\end{equation}
\end{prop}

\begin{proof}
We preliminarily observe that 
\begin{equation}
\label{covkhr2p1}
\ket{a\smallsmile b}=\ket{a}\smallsmile\ket{b}=\ket{a}\otimes\ket{b},
\end{equation}
on account of (4.1.11) and (4.2.12) of I. 
From \ceqref{whgstab1}, using relations \ceqref{dhjntdh1} and \ceqref{covkhr2p1}, we obtain then
\hphantom{xxxxxxxxxxx}
{\allowdisplaybreaks
\begin{align}
\label{covkhr2p2}
\ket{(H,\varrho)\smallsmile(K,\varsigma),a\smallsmile b}
&=D_{(H,\varrho)\smallsmile(K,\varsigma)}\ket{a\smallsmile b}
\\
\nonumber
&=D_{(H,\varrho)}\otimes D_{(K,\varsigma)}\ket{a}\otimes\ket{b}
\\
\nonumber
&=D_{(H,\varrho)}\ket{a}\otimes D_{(K,\varsigma)}\ket{b}
\\
\nonumber
&=\ket{(H,\varrho),a}\otimes\ket{(K,\varsigma),b}=\ket{(H,\varrho),a}\smallsmile\ket{(K,\varsigma),b}.
\end{align}
}
\!\!\ceqref{covkhr2} is in this way demonstrated. 

From \ceqref{whgstab1}, employing \ceqref{dhjntdh2}, we find
\begin{equation}
\label{covkhr3p1}
\ket{(O,\varepsilon),0}=D_{(O,\varepsilon)}\ket{0}=\ket{0},
\end{equation}
proving \ceqref{covkhr3}. 
\end{proof}


We present a few sample computations. 

\begin{exa} \label{exa:bellort} The Bell states.
{\rm In ex. \cref{exa:bell}, working in the qubit case $\msR=\bbF_2$, we derived the four 2-qubit Bell states as certain
calibrated hypergraph states $\ket{(H,\varrho_{a_0,a_1})}$, $a_0,a_1\in\msR$, given by \ceqref{bell3}, \ceqref{bell4}
and more explicitly by \ceqref{bell5}. The associated orthonormal bases $\ket{(H,\varrho_{a_0,a_1}),b_0,b_1}$, 
$b_0,b_1\in\msR$, are gotten by acting with the diagonal operators $Z_2(b_0,b_1)$ on the
states $\ket{(H,\varrho_{a_0,a_1})}$ in accordance to 
\ceqref{whgstab1}, 
\begin{equation}
\label{}
\ket{(H,\varrho_{a_0,a_1}),b_0,b_1}
=\mycom{{}_\sss}{{}_{x_0,x_1\in\msR}}\ket{x_0,x_1}\,2^{-1}(-1)^{b_0x_0+b_1x_1+\sigma_{(H,\varrho_{a_0,a_1})}(x_0,x_1)}.
\end{equation}
Examining \ceqref{bell4} shows that 
$b_0x_0+b_1x_1+\sigma_{(H,\varrho_{a_0,a_1})}(x_0,x_1)=\sigma_{(H,\varrho_{a_0+b_0,a_1+b_1})}(x_0,x_1)$. 
As a consequence, one has 
\begin{equation}
\label{}
\ket{(H,\varrho_{a_0,a_1}),b_0,b_1}=\ket{(H,\varrho_{a_0+b_0,a_1+b_1})}. 
\end{equation}
Therefore, if we take any one of the four Bell states \ceqref{bell5} as ground hypergraph state,
then the associated orthonormal basis consists of the chosen state and the remaining three ones.
This fact is well--known. 
}
\end{exa}

\begin{exa} \label{exa:3qtchgort} Three qutrit calibrated hypergraph states.
{\rm In ex. \cref{exa:3qtchg}, concerned with the qutrit case $\msR=\bbF_3$, we recovered
the five 3--qutrit hypergraph states $\ket{(H^i,\varrho^i)}$, $i=\sfa,\sfb,\sfc,\sfd,\sfe$, originally 
studied in ref. \!\!\ccite{Giri:2024qtt}. These are shown in \ceqref{3qtchgp4}, \ceqref{3qtchgp4/1} and
more explicitly in two special cases in \ceqref{3qtchgp5}. The associated orthonormal bases $\ket{(H,\varrho^i),a_0,a_1,a_2}$, 
$a_0,a_1,a_2\in\msR$, are provided by the action of the diagonal operators $Z_3(a_0,a_1,a_2)$
on the states according to \ceqref{whgstab1}, obtaining 
\begin{equation}
\label{}
\ket{(H^i,\varrho^i),a_0,a_1,a_2}
=\mycom{{}_\sss}{{}_{x_0,x_1,x_2\in\msR}}\ket{x_0,x_1,x_2}\,3^{-3/2}\omega^{a_0x_0+a_1x_1+a_2x_2+\sigma_{(H^i,\varrho^i)}(x_0,x_1,x_2)},
\end{equation}
where $\omega=\exp(2\pi i/3)$. For fixed $i$ there twenty seven such states. Upon explicitly writing down 
the expansion \ceqref{3qtchgp4/1} of $\ket{(H^i,\varrho^i)}$ as in e.g. \ceqref{3qtchgp5}, 
they can be easily obtained by multiplying the expansion's term proportional to $\ket{x_0,x_1,x_2}$ by the phase
$\omega^{a_0x_0+a_1x_1+a_2x_2}$ for all the twenty seven triples $x_1,x_2,x_3$. 
}
\end{exa}

\vspace{-2mm}

\subsection{\textcolor{blue}{\sffamily Calibrated hypergraphs states and local maximal entangleability}}\label{subsec:chglme}

In this subsection, we show that calibrated hypergraphs states are locally maximally entangleable. 
The notion of local maximal entangleability was originally introduced in ref. \!\!\ccite{Kruszynska:2008lem} for multi qubit
states. It extends straightforwardly to multi qudit states.

We recall a few basic notions. Let $A$, $B$ be two sets of parties with Hilbert spaces $\scH_A$, $\scH_B$
of dimensions $d_A$, $d_B$, respectively. Let further $d_A\leq d_B$ for definiteness.
The total Hilbert space of all parties is $\scH_{AB}=\scH_A\otimes\scH_B$.
A state $\ket{\varPsi}\in\scH_{AB}$ is said to be locally maximally entangled if the condition 
\begin{equation}
\label{chglme1}
\Tr_B(\ket{\varPsi}\bra{\varPsi})=1_A/d_A 
\end{equation}
is satisfied. 
(When $d_A=d_B$, the above identity and the analogous one with $A$, $B$ interchanged
can be shown to be equivalent.) In other words, the marginal state $\rho_A$ of the parties $A$,
given by the left hand side of the above identity, is maximally mixed. 

The following can be proven. A state $\ket{\varPsi}\in\scH_{AB}$ is locally maximally entangled if 
and only if there are an orthonormal basis $\ket{a}$ of $\scH_A$ and an orthonormal set $\ket{\varPsi(a)}$
of $\scH_B$ both indexed by the same labels $a\in I$ such that 
\begin{equation}
\label{chglme2}
\ket{\varPsi}=\mycom{{}_\sss}{{}_{a\in I}}\ket{a}\otimes\ket{\varPsi(a)}\,d_A{}^{-1/2}.
\end{equation}

We now revert back to the study of $l$--qudit states of $\scH_E[l]$. 
With any unitary operator collection $\clU=\{U_{rx}\hfpt|\hfpt r\in [l],x\in\msR\}\subset\msU(\scH_1)$, associate 
the operator
\begin{equation}
\label{chglme3}
U_{\clU}=\mycom{{}_\sss}{{}_{x\in E[l]}}U_x\otimes F_l{}^+\ket{x}\bra{x}F_l,
\end{equation}
where for $x\in E[l]$ the operator $U_x$ is given by
\begin{equation}
\label{chglme4}
U_x=\mycom{{}_\ooo}{{}_{r\in [l]}}U_{rx_r}.
\end{equation}
Notice that $U_{\clU}\in\msU(\scH_E[2l])$. 

Let $l\in\bbN$. Generalizing the original analysis of ref. \!\!\ccite{Kruszynska:2008lem} for qubits, we introduce the
following definition of local maximal entangleability for $l$--qudit states of $\scH_E[l]$. 

\begin{defi}
A state $\ket{\varPsi}\in\scH_E[l]$ is said to be locally maximally entangleable 
if there exists a unitary operator collection $\clU=\{U_{rx}\}$ with the property that the state 
$U_\clU\ket{\varPsi}\,\otimes\,\ket{0_l}\in\scH_E[2l]$ is 
locally maximally entangled with reference to the multipartite partition
$\scH_E[2l]=\scH_E[l]\otimes\scH_E[l]$. 
\end{defi}

An equivalent condition for local maximal entangleability is established by the following proposition 
also demonstrated in ref. \!\!\ccite{Kruszynska:2008lem} for qubits and extended here to qudits.

\begin{prop} \label{prop:lmecnd}
A state $\ket{\varPsi}\in\scH_E[l]$ is locally maximallys entangleable 
if and only if there exists a unitary operator collection $\clU=\{U_{rx}\}$ such that
the states $U_x\ket{\varPsi}$, $x\in E[l]$, constitute an orthonormal basis
of $\scH_E[l]$.
\end{prop}

\begin{proof}
By virtue of \ceqref{chglme3}, we have 
\begin{equation}
\label{lmecndp1}
U_\clU\ket{\varPsi}\otimes\ket{0_l}=\mycom{{}_\sss}{{}_{x\in E[l]}}U_x\ket{\varPsi}\otimes F_l{}^+\ket{x}\,q^{-l/2},
\end{equation}
where we employed the relation $\exval{x}{F_l}{0_l}=q^{-l/2}\omega^{\langle x,0_l\rangle}=q^{-l/2}$ following 
from \ceqref{four1}. Denote by $\Tr_{|l}$ the tracing operation on the second tensor factor of the tensor product 
$\scH_E[2l]=\scH_E[l]\otimes\scH_E[l]$. From \ceqref{lmecndp1}, we have 
{\allowdisplaybreaks
\begin{align}
\label{lmecndp2}
\Tr_{|l}\big(U_\clU\ket{\varPsi}\otimes\ket{0_l}&\,\bra{\varPsi}\otimes\bra{0_l}U_\clU{}^+\big)
\\
\nonumber
&=\mycom{{}_\sss}{{}_{x,y\in E[l]}}\Tr_{|l}\left(U_x\ket{\varPsi}\,q^{-l}\,\bra{\varPsi}U_y{}^+\otimes F_l{}^+\ket{x}
\,\bra{y}F_l\right)
\\
\nonumber
&=\mycom{{}_\sss}{{}_{x,y\in E[l]}}U_x\ket{\varPsi}\,q^{-l}\Tr\left(F_l{}^+\ket{x}\,\bra{y}F_l\right)\bra{\varPsi}U_y{}^+
\\
\nonumber
&=\mycom{{}_\sss}{{}_{x,y\in E[l]}}U_x\ket{\varPsi}\,q^{-l}\delta_{x,y}\,\bra{\varPsi}U_y{}^+
=\mycom{{}_\sss}{{}_{x\in E[l]}}U_x\ket{\varPsi}\,q^{-l}\,\bra{\varPsi}U_x{}^+,
\end{align}
}
\!\!since $\Tr(F_l{}^+\ket{x}\bra{y}F_l)=\delta_{x,y}$. 
Suppose that $\ket{\varPsi}$ is locally maximally entangleable. Then,
\begin{equation}
\label{lmecndp3}
\Tr_{|l}\big(U_\clU\ket{\varPsi}\otimes\ket{0_l}\,\bra{\varPsi}\otimes\bra{0_l}U_\clU{}^+\big)=q^{-l}1_l.
\end{equation}
Owing to \ceqref{lmecndp2}, in such a case we have
\begin{equation}
\label{lmecndp4}
\mycom{{}_\sss}{{}_{x\in E[l]}}U_x\ket{\varPsi}\,\bra{\varPsi}U_x{}^+=1_l.
\end{equation}
It follows that the kets $U_x\ket{\varPsi}$, $x\in E[l]$, constitute an orthonormal basis of $\scH_E[l]$.
Suppose conversely that the kets $U_x\ket{\varPsi}$, $x\in E[l]$, constitute an orthonormal basis of $\scH_E[l]$.
Then, \ceqref{lmecndp4} holds and so by \ceqref{lmecndp2} also \ceqref{lmecndp3} does. 
Hence, $\ket{\varPsi}$ is locally maximally entangleable. 
\end{proof}

This result has the following immediate application.

\begin{prop}
Let $l\in\bbN$ and let $(H,\varrho)\in G_C[l]$ be a calibrated hypergraph. Then, the hypergraph state
$\ket{(H,\varrho)}$ is locally maximally entangleable. 
\end{prop}

\begin{proof} Consider the unitary collection $\clU=\{U_{rx}\}$, where $U_{rx}=Z_1(x)$
with $r\in [l]$ and $x\in\msR$, where the operators $Z_l(a)$ are given by
\ceqref{qdpauli1}. Then, from \ceqref{chglme4}, for $x\in E[l]$
\begin{equation}
\label{}
U_x=\mycom{{}_\ooo}{{}_{r\in [l]}}Z_1(x_r)=Z_l(x),
\end{equation}
where we used \ceqref{qdpauli14}. Consequently, we have 
\begin{equation}
\label{}
U_x\ket{(H,\varrho)}=Z_l(x)\ket{(H,\varrho)}=\ket{(H,\varrho),x}
\end{equation}
owing to \ceqref{whgstab2}. The states $\ket{(H,\varrho),x}$, $x\in E[l]$,
constitute an orthonormal basis of $\scH_E[l]$ by prop. \cref{prop:hrkon}. 
Prop. \cref{prop:lmecnd} implies then that the state $\ket{(H,\varrho)}$ is locally maximally entangleable. 
\end{proof}

In ref. \!\!\ccite{Kruszynska:2008lem}, other results about locally maximally entangled qubit states were proven.
These do not extend to qudit states, but yet hold for our calibrated hypergraph states and boil down
to these having the structure \ceqref{whgsts10}. \ceqref{whgsts9}, \ceqref{whgsts14}.


\subsection{\textcolor{blue}{\sffamily Optimizing the classification of calibrated hypergraph states}}\label{subsec:hgststruc}

The classification program of calibrated hypergraph states consists in identifying the
hypergraph states which are truly distinct given a certain notion of `sameness'. 
In particular, categorizing 
hypergraph states into suitably defined entanglement classes is of key importance.
However, this will not be done in this paper.

The number of calibrations of a single hyperedge $X$ is $|\msM|^{|\msA|^{|X|}}$ and so it is characterized 
by an exponential of exponential growth with the cardinality of the hyperedge. The number of calibrated hypergraphs
with the same underlying bare hypergraph can therefore be staggeringly large. For this reason, the optimization of
the classification scheme employed is a compelling issue.

Different calibrated hypergraphs can yield the same calibrated hypergraph state up to a phase factor.
This fact indicates that a part of the information encoded in a hypergraph carries no weight in the
determination of the associated hypergraph state and is so redundant. An optimized classification scheme
for hypergraph states should take this in due account. 
Effective calibrated hypergraphs are streamlined hypergraphs with a limited amount of redundant data.


\begin{defi} Let $l\in\bbN$. A calibrated hypergraph $(H,\varrho)\in G_C[l]$ is said to be effective
if it enjoys the following properties. 

\begin{enumerate}

\item For all $X\in H$, $\varrho_X\neq 0_{\msA^X}$, where $0_{\msA^X}$ is the vanishing calibration of $X$. 

\item For all $X\in H$ and $w\in\msA^X$ with $\supp w\neq X$, one has $\varrho_X(w)=0$. 
  
\end{enumerate}

\end{defi} 

\noindent
For a finite subset $X\subset\bbN$, the support of an exponent function $w\in\msA^X$ is the set
$\supp w=\{r\hfpt|\hfpt r\in X, w(r)\neq 0\}\subseteq X$. 
Note that the hypergraph $(\emptyset,e_\emptyset)\in G_C[l]$ (see subsect. 3.2 of I) 
is trivially effective. 

The relevance of the notion of hypergraph effectiveness is validated by the following proposition.
Its proof is constructive and so also of considerable practical significance. 

\begin{prop} \label{prop:effhgs}
Let $l\in\bbN$. For every calibrated hypergraph $(H,\varrho)\in G_C[l]$, there exists an effective
calibrated hypergraph $(K,\varsigma)\in G_C[l]$ such that the associated calibrated hypergraph states obey
$\ket{(H,\varrho)}=\ket{(K,\varsigma)}\,\omega^a$ for some $a\in\msP$. 
\end{prop}

\begin{proof}
We recall that $z^0=1$ for $z\in\msR$, where $0\in\msZ$ denotes the additive unit of the cyclicity monoid. 
From \ceqref{whgsts9}, the phase function $\sigma_{(H,\varrho)}$ can thus be expressed as 
\begin{equation}
\label{effhgsp1}
\sigma_{(H,\varrho)}(x)=
\mycom{{}_\sss}{{}_{X\in H}}\mycom{{}_\sss}{{}_{D\subseteq X}}\mycom{{}_\sss}{{}_{w\in\msA^X,\,\supp w=D}}\varrho_X(w)
\tr\left(\mycom{{}_\ppp}{{}_{r\in D}}x_r{}^{w(r)}\right)
\end{equation}  
with $x\in E[l]$. We can split the sum in the right hand side in two parts,
\begin{equation}
\label{effhgsp2}
\sigma_{(H,\varrho)}(x)=\sigma^{(0)}{}_{(H,\varrho)}(x)+\sigma^{(1)}{}_{(H,\varrho)}(x). 
\end{equation}  
$\sigma^{(0)}{}_{(H,\varrho)}(x)$ is the contribution of all terms in \ceqref{effhgsp1} with 
$D=\emptyset$. Since the product in the trace equals $1$ in this case 
\begin{equation}
\label{effhgsp3}
\sigma^{(0)}{}_{(H,\varrho)}(x)=\mycom{{}_\sss}{{}_{X\in H}}\varrho_X(0_X)\tr(1)=a,
\end{equation}  
where $0_X\in\msA^X$ is the vanishing exponent function of $X\in H$ and $a\in\msM$
denotes the value independent from $x$ of this expression. 
$\sigma^{(1)}{}_{(H,\varrho)}(x)$ is the resultant of all terms in \ceqref{effhgsp1} with 
$D\neq\emptyset$,
\begin{equation}
\label{effhgsp4}
\sigma^{(1)}{}_{(H,\varrho)}(x)=\mycom{{}_\sss}{{}_{X\in H}}\mycom{{}_\sss}{{}_{D\subseteq X,D\neq\emptyset}}
\mycom{{}_\sss}{{}_{w\in\msA^X,\,\supp w=D}}\varrho_X(w)\tr\left(\mycom{{}_\ppp}{{}_{r\in D}}x_r{}^{w(r)}\right)\!.
\end{equation}  
Set $L=\{Z\hfpt|\hfpt Z\subseteq X, Z\neq\emptyset~\text{for some $X\in H$}\}$. Then,
$\sigma^{(1)}{}_{(H,\varrho)}(x)$ can be reexpressed in the form
{\allowdisplaybreaks
\begin{align}
\label{effhgsp5}
\sigma^{(1)}{}_{(H,\varrho)}(x)
&=\mycom{{}_\sss}{{}_{Z\in L}}\mycom{{}_\sss}{{}_{u\in\msA^Z,\,\supp u=Z}}
\mycom{{}_\sss}{{}_{X\in H}}\mycom{{}_\sss}{{}_{D\subseteq X,D\neq\emptyset}}
\mycom{{}_\sss}{{}_{w\in\msA^X,\,\supp w=D}}
\\
\nonumber
&\hspace{6.95cm}\varrho_X(w)\hfpt\delta_{Z,D}\,\delta_{u,w|_D}
\tr\left(\mycom{{}_\ppp}{{}_{r\in D}}x_r{}^{w|_D(r)}\right)
\\
\nonumber
&=\mycom{{}_\sss}{{}_{Z\in L}}\mycom{{}_\sss}{{}_{u\in\msA^Z,\,\supp u=Z}}
\mycom{{}_\sss}{{}_{X\in H,\,Z\subseteq X}}\mycom{{}_\sss}{{}_{w\in\msA^X,\,\supp w=Z,\,w|_Z=u}}
\varrho_X(w)\tr\left(\mycom{{}_\ppp}{{}_{t\in Z}}x_t{}^{u(t)}\right)
\\
\nonumber
&\hspace{7.3cm}=\mycom{{}_\sss}{{}_{Z\in L}}\mycom{{}_\sss}{{}_{u\in\msA^Z}}
\varphi_Z(u)\tr\left(\mycom{{}_\ppp}{{}_{t\in Z}}x_t{}^{u(t)}\right)\!, 
\end{align}
}
\!\!where we set 
\begin{equation}
\label{effhgsp6}
\varphi_Z(u)=
\Bigg\{
\begin{array}{ll}
\sum_{X\in H,\,Z\subseteq X}\sum_{w\in\msA^X,\,\supp w=Z,\,w|_Z=u}
\varrho_X(w)&\text{if $\supp u=Z$},\\
0&\text{else}
\end{array}
\end{equation}  
for $Z\in L$ and $u\in\msA^Z$. Let us set next $K=\{Y\hfpt|\hfpt Y\in L~\text{with}~\varphi_Y\neq 0_{\msA_Y}\}$ 
and $\varsigma_Y=\varphi_Y$ for $Y\in K$. Then, $K\in G[l]$ is a hypergraph and $\varsigma\in C(K)$
is a calibration of $K$ so that $(K,\varsigma)\in G_C[l]$ is calibrated hypergraph.
By construction, $(K,\varsigma)$ is effective. Furthermore, by the calculation \ceqref{effhgsp5} above, 
\begin{equation}
\label{effhgsp7}
\sigma^{(1)}{}_{(H,\varrho)}(x)=\sigma_{(K,\varsigma)}(x).
\end{equation}  
Relations \ceqref{effhgsp2}, \ceqref{effhgsp3} and \ceqref{effhgsp7} together imply that 
\begin{equation}
\label{effhgsp8}
\sigma_{(H,\varrho)}(x)=a+\sigma_{(K,\varsigma)}(x)
\end{equation}  
for $x\in E[l]$. The statement now follows readily from \ceqref{whgsts10} and \ceqref{whgsts14}.   
\end{proof} 


The following definition is now apposite. 

\begin{defi}
A calibrated hypergraph state is called effectively represented when it is expressed as 
$\ket{(H,\varrho)}$ for some $l\in\bbN$ and effective calibrated hypergraph
$(H,\varrho)\in G_C[l]$.   
\end{defi}

\noindent
It is important to have clear in one's mind that being effectively represented is a property of the
way a hypergraph state is encoded by a hypergraph and not of the state itself,
since the encoding hypergraph is not unique in general.

\begin{exa} \label{exa:effhyp} Effectively represented three qutrit calibrated hypergraph states.
{\rm Consider again the three qutrit set--up of ex. \cref{exa:3qtchg}.
Our purpose is to cast the calibrated hypergraph states $\ket{(H^i,\varrho^i)}$,
$i=\sfa,\sfb,\sfc,\sfd,\sfe$, in effectively represented form.
To begin with, we need to list all possible hyperedges that can be generated from the
ordinal $[2]$. There are seven of these altogether, viz
$Y^0=\{0\}$, $Y^1=\{1\}$, $Y^2=\{2\}$, $Y^3=\{0,1\}$, $Y^4=\{1,2\}$, $Y^5=\{0,2\}$, $Y^6=\{0,1,2\}$. 
To express a suitable set $\varsigma^0,\ldots,\varsigma^6$ of basic calibrations 
of $Y^0,\ldots,Y^6$ we have to specify the indexings of the exponent monoids
$\msA^{Y^0},\ldots,\msA^{Y^6}$ of $Y^0,\ldots,Y^6$ used. 
We use indexings
$\msA^{Y^0}=(v^0{}_0,\ldots,v^0{}_3)$, $\msA^{Y^1}=(v^1{}_0,\ldots,v^1{}_3)$, $\msA^{Y^2}=(v^2{}_0,\ldots,v^2{}_3)$
of $\msA^{Y^0},\msA^{Y^1},\msA^{Y^2}$ for which 
\begin{align}
\label{effhypex1}
&v^0{}_0=((1,0,1)), &v^1{}_0=((1,0,1)), &&v^2{}_0=((1,0,1)).
\end{align}
The indexings of $\msA^{Y^3}$, $\msA^{Y^4}$, $\msA^{Y^5}$, $\msA^{Y^6}$, viz
$\msA^{Y^3}=(v^3{}_0,\ldots,v^3{}_{15})$, $\msA^{Y^4}=(v^4{}_0,\ldots,v^4{}_{15})$, 
$\msA^{Y^5}=(v^5{}_0,\ldots,v^5{}_{15})$, $\msA^{Y^6}=(v^6{}_0,\ldots,v^6{}_{63})$ 
coincide with the indexings $\msA^{X^0}$, $\msA^{X^1}$, $\msA^{X^2}$, $\msA^{X^3}$, viz
$\msA^{X^0}=(w^0{}_0,\ldots,w^0{}_{15})$, $\msA^{X^1}=(w^1{}_0,\ldots,w^1{}_{15})$, 
$\msA^{X^2}=(w^2{}_0,\ldots,w^2{}_{15})$, $\msA^{X^3}=(w^3{}_0,\ldots,w^3{}_{63})$
respectively, given by expressions \ceqref{3qtchgp2}, by virtue of
the identities $Y^3=X^0$, $Y^4=X^1$, $Y^5=X^2$, $Y^6=X^3$. 
With reference to these indexings, the calibrations $\varsigma^0,\ldots,\varsigma^6$
read as
{\allowdisplaybreaks
\begin{align}
\label{effhypex2}
&\varsigma^0=(1,0,0,0), &&\varsigma^1=(1,0,0,0), &&\varsigma^2=(1,0,0,0),\hspace{.265cm}
\\
\nonumber
&\varsigma^3=(1,0,\ldots,0), &&\varsigma^4=(1,0,\ldots,0), &&\varsigma^5=(1,0,\ldots,0),
\\
\nonumber
&\varsigma^6=(1,0,\ldots,0) &&~ &&~   
\end{align}
}
\!\!(N.B. three $4$--tuples, three $16$--tuples and one $64$--tuple, respectively).
We can now express the given hypergraph states $\ket{(H^i,\varrho^i)}$
in effectively represented form $\ket{(K^i,\varsigma^i)}$ for each $i=\sfa,\sfb,\sfc,\sfd,\sfe$.
The effective calibrated hypergraphs $(K^i,\varsigma^i)$ appearing here are all expressible in terms
of the hyperedges $Y^0,\ldots,Y^6$ and calibrations $\varsigma^0,\ldots,\varsigma^6$ thereof. 
Their structure can be read off by inspecting the phase functions $\sigma_{(H^i,\varrho^i)}$ 
of the states $(H^i,\varrho^i)$ shown in \ceqref{3qtchgp4}. We find 
{\allowdisplaybreaks
\begin{align}
\label{effhypex3}
&K^\sfa=\{Y^2,Y^4,Y^5,Y^6\},&&\varsigma^\sfa=\{\varsigma^2,2\varsigma^4,2t|_{Y^5*}\varsigma^5,\varsigma^6\},
\\
\nonumber
&K^\sfb=\{Y^1,Y^2,Y^3,Y^4\},&&\varsigma^\sfb=\{2\varsigma^1,2\varsigma^2,\varsigma^3,\varsigma^4\},
\\
\nonumber
&K^\sfc=\{Y^0,Y^1,Y^2,Y^3,Y^4,Y^5\},
&&\varsigma^\sfc=\{2\varsigma^0,2\varsigma^1,2\varsigma^2,\varsigma^3,\varsigma^4,\varsigma^5\},
\\
\nonumber
&K^\sfd=\{Y^1,Y^3,Y^5,Y^6\},&&\varsigma^\sfd=\{2\varsigma^1,\varsigma^3,2t|_{Y^5*}\varsigma^5,\varsigma^6\},
\\
\nonumber
&K^\sfe=\{Y^0,Y^1,Y^3,Y^5,Y^6\},
&&\varsigma^\sfe=\{2\varsigma^0,2\varsigma^1,\varsigma^3,\varsigma^5+2t|_{Y^5*}\varsigma^5,\varsigma^6\},
\end{align}
}
\!\!where $t\in\Hom_\varOmega([2],[2])$ is the transposition of $0,2$. 
}
\end{exa}

Prop. \cref{prop:effhgs} states that every hypergraph state can be effectively
represented after a physically irrelevant phase redefinition.
An optimized classification program of calibrated hypergraph states
might therefore safely be restricted to the effectively represented ones.
However, even after doing so, there still remain redundancies which need to be disposed of.
We shall show how next. 

Let $l\in\bbN$ and let $H\in G[l]$ be a hypergraph. The support $\sigma H$ of $H$ is the set
\begin{equation}
\label{hgsupind1}
\sigma H=\mycom{{}_\uuu}{{}_{X\in H}}X.
\end{equation}
The index of $H$ is the integer \hphantom{xxxxxxxxxxx} 
\begin{equation}
\label{hgsupind2}
\iota H=|\sigma H|. 
\end{equation}
The support $\sigma H$ is therefore the set of the vertices of the 
hyperedges of $H$. The index $\iota H$ is the number of such vertices. Clearly, $\sigma H\subseteq [l]$
and $\iota H\leq l$.

A hypergraph $H\in G[l]$ is said to be primitive if $\iota H=l$. So, $H$ is primitive 
if  each vertex of $[l]$ is contained in at least one hyperedge of $H$.

\begin{defi}
An effective calibrated hypergraph $(H,\varrho)\in G_C[l]$ is said to be primitive if $H$ is primitive. 
\end{defi}

\noindent
The following proposition indicates that with each effective calibrated hypergraph there is associated 
uniquely a primitive effective calibrated hypergraph, which in the appropriate sense stated codifies it.
We recall here that for $l,m\in\bbN$ a morphism $f\in\Hom_\varOmega([l],[m])$ is increasing if for any $r,s\in[l]$ with $r<s$,
one has that $f(r)<f(s)$. 

\begin{prop} \label{prop:projexst}
Let $l\in\bbN$ and let $(H,\varrho)\in G_C[l]$ be an effective calibrated hypergraph. Then,
there are a unique increasing injective morphism $z\in\Hom_\varOmega([\iota H],[l])$ such that $z([\iota H])=\sigma H$ 
and a unique primitive effective calibrated hypergraph $(\bar H,\bar\varrho)\in G_C[\iota H]$
with the property that $G_Cz(\bar H,\bar\varrho)=(H,\varrho)$ (cf. eq. (3.2.33) of I). 
\end{prop}

\begin{proof} Since $\sigma H\subseteq [l]$ and $\iota H\leq l$, there is an injective function 
$z:[\iota H]\rightarrow [l]$ such that $z([\iota H])=\sigma H$. 
If $z$ is required to be increasing, then $z$ is unique with this property.

For all $X\in H$, $X\subseteq\sigma H=z([\iota H])$ and $X\neq\emptyset$
and consequently $z^{-1}(X)\subseteq[\iota H]$ and $z^{-1}(X)\neq\emptyset$. 
Therefore, if we set 
\begin{equation}
\label{hgststrucp1}
\bar H=\{z^{-1}(X)\hfpt|\hfpt X\in H\},
\end{equation}
then $\bar H\in G[\iota H]$ is a hypergraph. It is simple to see that $\bar H$ contains
the same number of hyperedges as $H$ does
and that $Gz(\bar H)=H$, as a consequence of the fact that $z$ maps $[\iota H]$ injectively onto $\sigma H$
and therefore $z(z^{-1}(X))=X$ for $X\subseteq\sigma H$. 

We note that for all $X\in H$, $z|_{z^{-1}(X)}:z^{-1}(X)\rightarrow X$ is a bijection. 
Since $\varrho_X\in \msM^{\msA^X}$ is a calibration of $X$, 
its  push--forward 
\begin{equation}
\label{hgststrucp2}
\bar\varrho_{z^{-1}(X)}=(z|_{z^{-1}(X)})^{-1}{}_*(\varrho_X)\in \msM^{\msA^{z^{-1}(X)}}
\end{equation}
is a calibration of $z^{-1}(X)$ (cf. def. 3.2.4, eq. (3.2.13) of I). 
Owing to \ceqref{hgststrucp1}, the above expression defines a calibration $\bar\varrho\in C(\bar H)$
of $\bar H$. We thus constructed a calibrated hypergraph $(\bar H,\bar\varrho)\in G_C[\iota H]$. 
We have now to show that $(\bar H,\bar\varrho)$ is effective. This requires proving a few preliminary results.

Let $X,Y\subset\bbN$ be non empty finite subsets and let $h:X\rightarrow Y$ be a bijection.
If $w\in\msA^X$ is an exponent function of $X$, then its push--forward $h_\star(w)\in\msA^Y$ (cf. def.
3.2.2 of I) is given by $h_\star(w)=w\circ h^{-1}$, as follows readily from (3.2.2) of I and the
invertibility of $h$. From here, it is evident that $\supp h_\star(w)=Y$ if and only if $\supp w=X$. 

Let $X,Y\subset\bbN$ and $h:X\rightarrow Y$ be as above non empty finite subsets and a bijection, respectively.
By the bijectivity of $h$ and (3.2.8) of I, the exponent function push--forward map
$h_\star:\msA^X\rightarrow\msA^Y$ is a bijection with inverse $h_\star{}^{-1}=h^{-1}{}_\star$.
Consequently, if $\varpi\in\msM^{\msA^X}$ is a calibration of $X$, its 
push--forward $h_*(\varpi)\in\msM^{\msA^Y}$ (cf. def. 3.2.4 of I) is given by
$h_*(\varpi)=\varpi\circ h^{-1}{}_\star$, 
by the remark just below eq. (3.2.13) of I. 
It is evident from this expression that $h_*(\varpi)\neq 0_{\msA^Y}$ if and only if
$\varpi\neq 0_{\msA^X}$. It also appears by the findings of the previous paragraph that
$h_*(\varpi)(v)=0$ for all $v\in\msA^Y$ with $\supp v\neq Y$ if and only if
$\varpi(w)=0$ for all $w\in\msA^X$ with $\supp w\neq X$.

By \ceqref{hgststrucp1} and the injectivity of $h$ again, expression \ceqref{hgststrucp2} can be cast as 
\begin{equation}
\label{hgststrucp4}
\bar\varrho_{\bar X}=(z|_{\bar X})^{-1}{}_*(\varrho_{z(\bar X)})
\end{equation}
with $\bar X\in\bar H$, where $z|_{\bar X}:\bar X\rightarrow z(\bar X)$ is a bijection.
The findings at the end of the previous paragraph ensure now that the calibrated hypergraph
$(\bar H,\bar\varrho)$ is effective owing to $(H,\varrho)$ being so, as required. 

The hypergraph $(\bar H,\bar\varrho)$ is in addition primitive.
Indeed, because of \ceqref{hgsupind1} and \ceqref{hgststrucp1}, we have that $\sigma\bar H=z^{-1}(\sigma H)$
and so, by \ceqref{hgsupind1}, $\iota\bar H=\iota H$. Therefore, $\bar H$ is primitive implying that 
$(\bar H,\bar\varrho)$ is too. 

From \ceqref{hgststrucp1}, we find that $Gz(\bar H)=H$. From (3.2.22) of I, we have further that 
{\allowdisplaybreaks
\begin{align}
\label{hgststrucp3}
z_{\bar H*}(\bar\varrho)_X&=\mycom{{}_\sss}{{}_{\bar X\in\bar H,z(\bar X)=X}}z|_{\bar X*}(\bar\varrho_{\bar X})
\\
\nonumber
&=z|_{z^{-1}(X)*}(\bar\varrho_{z^{-1}(X)})
\\
\nonumber
&=z|_{z^{-1}(X)*}((z|_{z^{-1}(X)})^{-1}{}_*(\varrho_X))=\varrho_X,
\end{align}
}
\!\!where we used that only $z^{-1}(X)\in\bar H$ contributes to the sum by the injectivity of $z$ 
and in the last step we relied on (3.2.16), (3.2.17) of I. 
Consequently, $z_{\bar H*}(\bar\varrho)=\varrho$.
From (3.2.33) of I and by what we have just obtained, $G_Cz(\bar H,\bar\varrho)=(H,\varrho)$ as stated. 
$(\bar H,\bar\varrho)$ is unique with this property, since $G_Cz$ is injective
because of $z$ being so and the functoriality of $G_C$. This completes the proof. 
\end{proof}

Prop. \cref{prop:projexst} justifies the following definition

\begin{defi}
Let $l\in\bbN$ and let $(H,\varrho)\in G_C[l]$ be an effective calibrated hypergraph.
The chart of $(H,\varrho)$ is the unique increasing injective morphism
$z_H\in\Hom_\varOmega([\iota H],[l])$ such that $z_H([\iota H])=\sigma H$. The primitive core
$\pro(H,\varrho)\in G_C[\iota H]$ of $(H,\varrho)$ is the unique primitive effective calibrated 
hypergraph such that
\begin{equation}
\label{chroot}
G_Cz_H(\pro(H,\varrho))=(H,\varrho).
\end{equation}
\end{defi}

\noindent
We note also that when $(H,\varrho)$ is primitive $\pro(H,\varrho)=(H,\varrho)$, since $\iota H=l$ and
so $z_H=\id_{[l]}$ in that case.

The above analysis suggests the following definition. 

\begin{defi}
A calibrated hypergraph state is said to be primitively effectively represented if it is expressed as 
$\ket{(H,\varrho)}$ for some $l\in\bbN$ and primitive effective calibrated hypergraph
$(H,\varrho)\in G_C[l]$.   
\end{defi}

\noindent
Again, recall that being primitively effectively represented is a property of the
way a hypergraph state is encoded by a hypergraph and not of the state itself.

\begin{exa} \label{exa:pefhyp} Primitively effectively represented three qutrit calibrated hypergraph states.
{\rm The calibrated hypergraph states $\ket{(K^i,\varsigma^i)}$, $i=\sfa,\sfb,\sfc,\sfd,\sfe$, 
examined in ex. \cref{exa:effhyp} are all primitively effectively represented.
This is evident in particular for the states $\ket{(K^\sfc,\varsigma^\sfc)}$, $\ket{(K^\sfe,\varsigma^\sfe)}$
upon inspecting the hypergraphs $K^\sfc$, $K^\sfe$ shown in \ceqref{effhypex3}.
}
\end{exa}

Owing to \ceqref{whgsts16} and \ceqref{chroot}, if $l\in\bbN$ and $(H,\varrho)\in G_C[l]$ is an
effective calibrated hypergraph,  then the hypergraph states $\ket{(H,\varrho)}$ and $\ket{\pro(H,\varrho)}$
are related as 
\begin{equation}
\label{}
\ket{(H,\varrho)}=\ket{G_Cz_H(\pro(H,\varrho))}=\scH_Ez_H\ket{\pro(H,\varrho)}.
\end{equation}
Recall that $\scH_Ez_H\in\msI(\scH_E[\iota H],\scH_E[l])$ is isometric because $z_H$
is injective (cf. prop. 4.3.1 of I). Consequently, in any sensible classification scheme, the hypergraph states
$\ket{(H,\varrho)}$, $\ket{\pro(H,\varrho)}$ need not be considered independently. Notice however that while 
$\ket{(H,\varrho)}$ is merely effectively represented, $\ket{\pro(H,\varrho)}$, more restrictively, is
primitively so. An optimized classification program of calibrated hypergraph states may therefore be
narrowed to those states which are primitively effectively represented. 

For $l\in\bbN$, the invertible morphisms $f\in\Hom_\varOmega([l],[l])$ are called automorphisms. They
form a subgroup $\Aut_\varOmega([l])$ of the monoid $f\in\Hom_\varOmega([l],[l])$.
As $\Hom_\varOmega([l],[l])$ consists of all functions $f:[l]\rightarrow[l]$, $\Aut_\varOmega([l])$ is isomorphic
to the degree $l$ symmetric group $\msS(l)$. 

For any fixed $l\in\bbN$, the calibrated hypergraph functor $G_C$ establishes an action of the
automorphism group $\Aut_\varOmega([l])$ on the calibrated hypergraph set $G_C[l]$. This action is compatible
with hypergraph effectiveness. 

\begin{prop} \label{prop:pfrspctef}
Let $l\in\bbN$. Let $(H,\varrho)\in G_C[l]$ be an effective calibrated hypergraph and $f\in\Aut_\varOmega([l])$
be an automorphism. Then, the hypergraph $G_Cf(H,\varrho)\in G_C[l]$ is also effective. Further,
if $(H,\varrho)$ is primitive, then $G_Cf(H,\varrho)$ also is.
\end{prop}

\begin{proof}
In the course of the proof of prop. \cref{prop:projexst}, the following was found.
Let $X,Y\subset\bbN$ and $h:X\rightarrow Y$ be non empty finite subsets and a bijection, respectively.
Let further $\varpi\in\msM^{\msA^X}$ be a calibration of $X$.
Then, $h_*(\varpi)\neq 0_{\msA^Y}$ if and only if $\varpi\neq 0_{\msA^X}$.
Moreover, $h_*(\varpi)(v)=0$ for all $v\in\msA^Y$ with $\supp v\neq Y$
if and only if $\varpi(w)=0$ for all $w\in\msA^X$ with $\supp w\neq X$.
  

Consider now the calibrated hypergraph $G_Cf(H,\varrho)=(Gf(H),f_{H*}(\varrho))\in G_C[l]$
(cf. eq. (3.2.33) of I). The push--forward $f_{H*}(\varrho)\in C(Gf(H))$ of the calibration $\varrho$
(cf. def. 3.2.6 of I) appearing here is given by 
\begin{equation}
\label{pfrspctefp1}
f_{H*}(\varrho)_Y=f|_{f^{-1}(Y)*}(\varrho_{f^{-1}(Y)})
\end{equation}  
with $Y\in G_C(H)$ by (3.2.22) of I and the invertibility of $f$. From \ceqref{pfrspctefp1} and  the invertibility
of the functions $f|_X:X\rightarrow f(X)$ for all $X\in H$, the facts recalled in the previous paragraph
ensure now that the calibrated hypergraph $G_Cf(H,\varrho)$ is effective as a consequence of
$(H,\varrho)$ being so, as claimed.

Recall that $Gf(H)=\{f(X)\hfpt|\hfpt X\in H\}$. Thus, $\sigma Gf(H)=f(\sigma H)$,
by \ceqref{hgsupind1}, and so $\iota Gf(H)=\iota H$, by \ceqref{hgsupind2}. Then, if $(H,\varrho)$ is
primitive, then $G_Cf(H,\varrho)$ also is. 
\end{proof}

\begin{defi}
Let $l\in\bbN$ and let $(H,\varrho),(K,\varsigma)\in G_C[l]$ be primitive effective calibrated hypergraphs.
$(H,\varrho)$, $(K,\varsigma)$ are said to be congruent if there exists an automorphism
$f\in\Aut_\varOmega([l])$ such that $G_Cf(H,\varrho)=(K,\varsigma)$.
\end{defi}

\begin{prop}
For $l\in\bbN$, congruence of primitive effective calibrated hypergraphs of $G_C[l]$ is an equivalence relation. 
\end{prop}

\begin{proof}
This follows immediately from $\Aut_\varOmega([l])$ being a group acting on $G_C[l]$ through the functor $G_C$
and prop. \cref{prop:pfrspctef}. 
\end{proof}

\noindent
In this way, for each $l\in\bbN$, the subset of $G_C[l]$ of the primitive effective calibrated hypergraphs
is partitioned in pairwise disjoint congruence classes.



If $l\in\bbN$, $(H,\varrho),(K,\varsigma)\in G_C[l]$ are congruent primitive effective calibrated hypergraphs
and $f\in\Aut_\varOmega([l])$ is an automorphism such that $G_Cf(H,\varrho)=(K,\varsigma)$, then 
the associated congruent primitively effectively represented hypergraph states
$\ket{(H,\varrho)}$, $\ket{(K,\varsigma)}$ are related as 
\begin{equation}
\label{}
\ket{(K,\varsigma)}=\ket{G_Cf(H,\varrho)}=\scH_Ef\ket{(H,\varrho)}
\end{equation}
by virtue of \ceqref{whgsts16}. Now, $\scH_Ef\in\msU(\scH_E[l])$ is unitary because $f$
is invertible (cf. prop. 4.3.1 of I). Consequently, in in any sensible classification scheme, the hypergraph state
$\ket{(H,\varrho)}$, $\ket{(K,\varsigma)}$ need not be considered independently: they differ only by a relabelling
of the underlying $l$ qudits.
An optimized classification program of calibrated hypergraph states can therefore be reduced to cataloging 
the primitively effectively represented ones up to congruence.

Let $l\in\bbN$ and let $(H,\varrho)\in G_C[l]$ be a primitive effective calibrated hypergraph.
An automorphism $f\in\Aut_\varOmega([l])$ is said to fix $(H,\varrho)$ if $G_Cf(H,\varrho)=(H,\varrho)$.
The automorphisms fixing $(H,\varrho)$ constitute a subgroup $\msY_{(H,\varrho)}$ of $\Aut_\varOmega([l])$,
called the isotropy group of $(H,\varrho)$.

By general properties of group actions, in the case studied here the $\Aut_\varOmega([l])$ automorphism group action, 
if $(H,\varrho),(K,\varsigma)\in G_C[l]$ are congruent primitive effective calibrated hypergraphs, then their isotropy
groups $\msY_{(H,\varrho)}$, $\msY_{(K,\varsigma)}$ are conjugated in $\Aut_\varOmega([l])$.
We see so that every congruence class of primitive effective calibrated
hypergraphs $(H,\varrho)\in G_C[l]$ \pagebreak 
is characterized by an isotropy group conjugacy class in $\Aut_\varOmega([l])$.

From the findings of the previous paragraph, it is clear that every congruence class of primitively
effectively represented hypergraph states $\ket{(H,\varrho)}$ 
is characterized by an isotropy group conjugacy class in $\Aut_\varOmega([l])$. We note here that if
$f\in\msY_{(H,\varrho)}$ fixes $(H,\varrho)$, then the hypergraph state $\ket{(H,\varrho)}$ obeys
\begin{equation}
\label{}
\ket{(H,\varrho)}=\ket{G_Cf(H,\varrho)}=\scH_Ef\ket{(H,\varrho)}.
\end{equation}
$\scH_Ef$ thus fixes $\ket{(H,\varrho)}$.




\subsection{\textcolor{blue}{\sffamily Calibrated vs weighted hypergraph states}}\label{subsec:hgscompar}

An important issue of the theory developed in the previous subsections is the comparison of the calibrated 
hypergraph states it deals with and the weighted hypergraph states studied in earlier literature, e.g.
\ccite{Qu:2012eqs,Rossi:2012qhs,Steinhoff:2016:qhs,Xiong:2017qhp}.  

In the $\varOmega$ monadic framework, weighted hypergraph states are defined
analogously to calibrated hypergraph ones but relying on the weighted hypergraph $\varOmega$ monad
introduced in subsect. 3.3 of I rather than the calibrated one. The basic 
monoids $\msA$ and $\msM$ are assumed to be the same for both types of hypergraph states 
(see subsect. \cref{subsec:whgsts}). 

\begin{defi} \label{def:sigmaha}
Let $l\in\bbN$ and let $(H,\alpha)\in G_W[l]$ be a weighted hypergraph.
The qudit weighted hypergraph operator $D_{(H,\alpha)}\in\End_{\bfsfH}(\scH_E[l])$
associated with $(H,\alpha)$ is 
\begin{equation}
\label{hgscompar1}
D_{(H,\alpha)}=\mycom{{}_\sss}{{}_{x\in E[l]}}F_l{}^+\ket{x}\,\omega^{\sigma_{(H,\alpha)}(x)}\hfpt\bra{x}F_l,
\end{equation}
where $F_l$ is the Fourier transform operator (cf. eq. \ceqref{four1}) and 
the phase function $\sigma_{(H,\alpha)}:E[l]\rightarrow\msP$ is given by
\begin{equation}
\label{hgscompar3}
\sigma_{(H,\alpha)}(x)=\mycom{{}_\sss}{{}_{X\in H}}\alpha_X
\tr\left(\mycom{{}_\ppp}{{}_{r\in X}}x_r\right)
\end{equation}
for $x\in E[l]$. 
\end{defi}


\begin{defi} \label{def:ketha}
Let $l\in\bbN$. Further, let $(H,\alpha)\in G_W[l]$ be a weighted hypergraph. The qudit weighted hypergraph
state of $(H,\alpha)$ is the ket $\ket{(H,\alpha)}\in\scH_E[l]$ given by
\begin{equation}
\label{hgscompar4}
\ket{(H,\alpha)}=D_{(H,\alpha)}\ket{0_l}.
\end{equation}
\end{defi}

\noindent 
The notation we use for weighted hypergraph states is purposely analogous to that employed
for calibrated ones to make the comparison of the defs. \cref{def:sigmaha}, \cref{def:ketha} and 
\cref{def:sigmahr}, \cref{def:whgst} straightforward.
No confusion should arise.

The following result relates weighted to calibrated hypergraph states. 


\begin{prop} \label{prop:calweicomp}
Let $l\in\bbN$. Then, for each weighted hypergraph $(H,\alpha)\in G_W[l]$, there exists a calibrated
hypergraph $(H,\varrho)\in G_C[l]$ with the same underlying hypergraph $H\in G[l]$ with the property that
$\ket{(H,\alpha)}=\ket{(H,\varrho)}$.
\end{prop}

\begin{proof}
The cyclic submonoids $\msZ_x$ with $x\in\msR$ (cf. subsect. \cref{subsec:cycmongal}) have the 
property that $\msZ_x=[1]$ for $x=1$ and $\msZ_x\supseteq [2]$ for $x\neq 1$. We define a special
element $q=(q_x)_{x\in\msR}\in\msZ$ by $q_x=0$ for $x=1$ and $q_x=1$ for $x\neq 1$ and observe that 
$x^q=x$ for $x\in\msR$. 

For $X\in H$, there is a special element $w_X\in\msA^X$ defined by $w_X(r)=q$ for $r\in X$.
We define $\varrho\in C(H)$ by $\varrho_X(w)=\alpha_X\delta_{w,w_X}$ for $w\in\msA^X$. It is immediate to
check from \ceqref{whgsts9} and \ceqref{hgscompar3} that $\sigma_{(H,\varrho)}(x)=\sigma_{(H,\alpha)}(x)$
for $x\in E]l]$. From \ceqref{whgsts14} and \ceqref{hgscompar4}, it follows then that 
$\ket{(H,\alpha)}=\ket{(H,\varrho)}$, as stated. 
\end{proof}

The above proposition shows that calibrated hypergraph states constitute a broad generalization of
weighted hypergraph states. To appreciate the breadth of it, we introduce another class
of graph states, which are more general than weighted hypergraph states, in that
they allow for polynomial phase functions, but still
are instances of calibrated hypergraph states. Let us introduce the parameter
\begin{equation}
\label{}
\delta=\max_{x\in\msR}(\iota_x+\pi_x)\in\bbN_+, \vphantom{\bigg]}
\end{equation}
where $\iota_x$, $\pi_x$ are  the index and the period of $x$ (cf. subsect. \cref{subsec:cycmongal}).
$\delta$ has the property that for each $x\in\msR$ the cyclic monoid $\msC_x=\{x^k\hfpt|\hfpt k\in[\delta]\}$.
As a consequence, the finite set $\{x^k\hfpt|\hfpt x\in\msR,~k\in[\delta]\}$ includes all powers of all elements of $\msR$.
(Please, note that the $x^k$ need not be all distinct.)

\begin{prop} \label{prop:chgspoly}
Let $l\in\bbN$. Let $H\in G[l]$ be a hypergraph and for each hyperedge $X\in H$ let a function
$\tau_X:[\delta]^X\rightarrow\msM$ be assigned. Next, let $D_{(H,\tau)}\in\End_{\bfsfH}(\scH_E[l])$
be the operator with expression 
\begin{equation}
\label{chgspoly1}
D_{(H,\tau)}=\mycom{{}_\sss}{{}_{x\in E[l]}}F_l{}^+\ket{x}\,\omega^{\sigma_{(H,\tau)}(x)}\hfpt\bra{x}F_l,
\end{equation}
where the phase function $\sigma_{(H,\tau)}:E[l]\rightarrow\msP$ is given by
\begin{equation}
\label{chgspoly2}
\sigma_{(H,\tau)}(x)=\mycom{{}_\sss}{{}_{X\in H}}\mycom{{}_\sss}{{}_{a\in[\delta]^X}}\tau_X(a)
\tr\left(\mycom{{}_\ppp}{{}_{r\in X}}x_r{}^{a(r)}\right)
\end{equation}
for $x\in E[l]$. Let finally $\ket{(H,\tau)}\in\scH_E[l]$ be the state 
\begin{equation}
\label{chgspoly3}
\ket{(H,\tau)}=D_{(H,\tau)}\ket{0_l}.
\end{equation}
Then, there is a calibration $\varrho\in C(H)$ such that $\ket{(H,\tau)}=\ket{(H,\varrho)}$. 
\end{prop}

\begin{proof} Define a mapping 
  $h_\delta:[\delta]\rightarrow\msZ$ by setting $h_\delta(k)=(h_x(k))_{x\in\msR}$, $k\in[\delta]$, 
where the functions $h_x:\bbN\rightarrow\bbN$ are shown in \ceqref{cycmongal1}.
$h_\delta$ enjoys the property that $x^k=x^{h_\delta(k)}$ for all $x\in\msR$ and $k\in[\delta]$.
Consequently, expression \ceqref{chgspoly2} of $\sigma_{(H,\tau)}$ can be recast as 
{\allowdisplaybreaks
\begin{align}
\label{chgspolyp1}
\sigma_{(H,\tau)}(x)&=\mycom{{}_\sss}{{}_{X\in H}}\mycom{{}_\sss}{{}_{w\in\msA^X}}
\mycom{{}_\sss}{{}_{a\in[\delta]^X}}\tau_X(a)\delta_{w,h_\delta\hfpt\circ\hfpt a}
\tr\left(\mycom{{}_\ppp}{{}_{r\in X}}x_r{}^{h_\delta\hfpt\circ\hfpt a(r)}\right)
\\
\nonumber
&=\mycom{{}_\sss}{{}_{X\in H}}\mycom{{}_\sss}{{}_{w\in\msA^X}}\mycom{{}_\sss}{{}_{a\in[\delta]^X,w=h_\delta\hfpt\circ\hfpt a}}\tau_X(a)
\tr\left(\mycom{{}_\ppp}{{}_{r\in X}}x_r{}^{w(r)}\right)
\end{align}
}
\!\!for $x\in E[l]$. Let $\varrho\in C(H)$ be the calibration of the form
\begin{equation}
\label{chgspolyp2}
\varrho_X(w)=\mycom{{}_\sss}{{}_{a\in[\delta]^X,w=h_\delta\hfpt\circ\hfpt a}}\tau_X(a)
\end{equation}
for $X\in H$ and $w\in\msA^X$. Then, \ceqref{chgspolyp1} can be restated as 
\begin{equation}
\label{chgspolyp3}
\sigma_{(H,\tau)}(x)=\mycom{{}_\sss}{{}_{X\in H}}\mycom{{}_\sss}{{}_{w\in\msA^X}}\varrho_X(w)
\tr\left(\mycom{{}_\ppp}{{}_{r\in X}}x_r{}^{w(r)}\right)=\sigma_{(H,\varrho)}(x).
\end{equation}
From here, the statement follows from \ceqref{chgspoly1}--\ceqref{chgspoly3} and
\ceqref{whgsts10}, \ceqref{whgsts9}, \ceqref{whgsts14}. 
\end{proof}

In the familiar qubit case, however, weighted and calibrated hypergraph states turn out to be the same.  

\begin{prop} \label{pro:qudqubcomp}
Let $\msR=\bbF_2$. Let further $l\in\bbN$. Then, for every calibrated hypergraph
$(H,\varrho)\in G_C[l]$ there exists a weighted hypergraph $(L,\beta)\in G_W[l]$ with the property that 
$\ket{(L,\beta)}=\ket{(H,\varrho)}$ up to a sign.
\end{prop}

\begin{proof} The proof leverages a part of the computations carried out in the proof of 
prop. \cref{prop:effhgs} valid for a general Galois ring $\msR$. 
From relations \ceqref{effhgsp2}, \ceqref{effhgsp3}, \ceqref{effhgsp5}, we have 
\begin{equation}
\label{qudqubcompp1}
\sigma_{(H,\varrho)}(x)=a+\mycom{{}_\sss}{{}_{Z\in L}}\mycom{{}_\sss}{{}_{u\in\msA^Z,\,\supp u=Z}}
\varphi_Z(u)\tr\left(\mycom{{}_\ppp}{{}_{t\in Z}}x_t{}^{u(t)}\right)
\end{equation}
for $x\in E[l]=\bbF_2{}^l$, where $a\in\msM=\bbF_2$, $L=\{Z\hfpt|\hfpt Z\subset X, Z\neq\emptyset~\text{for some $X\in H$}\}$
and $\varphi_Z(u)$ is given by \ceqref{effhgsp6} for $Z\in L$ and $u\in\msA^Z$ and the property
that $\varphi_Z(u)=0$ if $\supp u\neq Z$ was used. Since $\msR=\bbF_2$, we have that $\msZ=\{(0,0),(1,0)\}$.
Hence, if $u\in\msA^Z$ with $\supp u=Z$, we have $u(t)=(1,0)$ and, consequently,
$x_t{}^{u(t)}=x_t$ for all $t\in Z$. Relation \ceqref{qudqubcompp1} can in this way be simplified as 
\begin{equation}
\label{qudqubcompp2}
\sigma_{(H,\varrho)}(x)=a+\mycom{{}_\sss}{{}_{Z\in L}}\beta_Z\tr\left(\mycom{{}_\ppp}{{}_{t\in Z}}x_t\right)\!,
\end{equation}
where $\beta_Z=\varphi_Z(u_Z)$, $u_Z\in\msA^Z$ being defined by $u_Z(t)=(1,0)$ for $t\in Z$.
Now, $L\in G[l]$ is a hypergraph and $\beta\in W(L)$ is a weighting of $L$. Thus,
$(L,\beta)\in G_W[l]$ is a weighted hypergraph. By \ceqref{hgscompar3}, we can so cast
relation \ceqref{qudqubcompp2} as 
\begin{equation}
\label{qudqubcompp3}
\sigma_{(H,\varrho)}(x)=a+\sigma_{(L,\beta)}(x)
\end{equation}
for $x\in E[l]$. The claim follows now readily by \ceqref{whgsts10}, \ceqref{whgsts14}
and \ceqref{hgscompar1}, \ceqref{hgscompar4}.
\end{proof}


\subsection{\textcolor{blue}{\sffamily The calibrated hypergraph state $\varOmega$ monad}}\label{subsec:cwhsmonad}

In this subsection, we shall show that calibrated hypergraph states organize in a certain graded monad.

The most salient properties of the calibrated hypergraph state construction
are displayed in \ceqref{whgsts16}, \ceqref{whgsts18} and \ceqref{whgsts19}. 
In accordance to prop. \cref{prop:whsmapmor}, they characterize the calibrated hypergraph state map $\ket{-}$
as a morphism of the calibrated hypergraph and multi qudit state graded 
$\varOmega$ monads $G_C\varOmega$ and $\scH_E\varOmega$.
The very same properties entail that the calibrated hypergraph states themselves
constitute a graded $\varOmega$ monad.

\begin{prop} \label{prop:chsmonad}
There exists a graded $\varOmega$ monad $\scC_E\varOmega$ with the following design.

\begin{enumerate}[label=\alph*)] %

\item \label{it:chsmonad1} For each $l\in\bbN$, $\scC_E[l]$ is the set of all calibrated hypergraph states contained in the Hilbert
space $\scH_E[l]$.

\item  \label{it:chsmonad2} For every $l,m\in\bbN$ and morphism $f\in\Hom_\varOmega([l],[m])$
\begin{equation}
\label{cwhsmonad1}
\scC_Ef(\ket{\varPhi})=\scH_Ef\ket{\varPhi}
\end{equation}
for $\ket{\varPhi}\in\scC_E[l]$.

\item  \label{it:chsmonad3} For every $l,m\in\bbN$ and states $\ket{\varPhi}\in\scC_E[l]$, $\ket{\varPsi}\in\scC_E[m]$
\begin{equation}
\label{cwhsmonad2}
\ket{\varPhi}\smallsmile\ket{\varPsi}=\ket{\varPhi}\otimes\ket{\varPsi}
\end{equation}
is the monadic product of $\ket{\varPhi}$, $\ket{\varPsi}$ in $\scC_E\varOmega$.

\item  \label{it:chsmonad4} $\ket{0}\in\scC_E[0]$ is the monadic unit of $\scC_E\varOmega$. 

\end{enumerate}

\end{prop}

\begin{proof}
The proof that $\scC_E\varOmega$ is an $\varOmega$ category is based on prop. 2.2.1 of I and involves a number of steps.

To begin with, we check that the definitions stated in items \cref{it:chsmonad1}, \cref{it:chsmonad2} are consistent.
By virtue of relation \ceqref{whgsts16}, we have that $\scH_Ef\hfpt\hfpt\scC_E[l]\subseteq\scC_E[m]$ for every morphism
\linebreak $f\in\Hom_\varOmega([l],[m])$. Therefore, \ceqref{cwhsmonad1} defines 
a mapping $\scC_Ef:\scC_E[l]\rightarrow\scC_E[m]$, as required. 
Next, we have to demonstrate that the maps $[l]\in\Obj_\varOmega\mapsto\scC_E[l]\in\Obj_{\ul{\rm Set}}$ and 
$f\in\Hom_\varOmega([l],[m])\mapsto\scC_Ef\in\Hom_{\ul{\rm Set}}(\scC_E[l],\scC_E[m])$, $l,m\in\bbN$,
are injective.

Let $l,m\in\bbN$ with $l\neq m$. We have that $\scC_E[l]\subset\scH_E[l]$, $\scC_E[m]\subset\scH_E[m]$. Since $\scH_E[l]$, $\scH_E[m]$
are vector space of different dimensions, $\scC_E[l]\neq\scC_E[m]$ necessarily. Consequently, 
the mapping $[l]\in\Obj_\varOmega\mapsto\scC_E[l]\in\Obj_{\ul{\rm Set}}$ is injective, as required.

Showing injectivity on morphisms requires more work. From \ceqref{whgsts10} and \ceqref{whgsts14}, 
the hypergraph state associated with a calibrated hypergraph $(H,\varrho)\in G_W[l]$ reads as 
\begin{equation}
\label{chsmonadp1}
\ket{(H,\varrho)}=\mycom{{}_\sss}{{}_{x\in E[l]}}F_l{}^+\ket{x}\,q^{-l/2}\omega^{\sigma_{(H,\varrho)}(x)},
\end{equation}
where we employed the relation
$\exval{x}{F_l}{0_l}=q^{-l/2}\omega^{\langle x,0_l\rangle}=q^{-l/2}$ ensuing from \ceqref{four1}. From \ceqref{chsmonadp1}, 
it follows that if $(H,\varrho),(K,\varsigma)\in G_W[l]$ are calibrated hypergraphs and $\ket{(H,\varrho)}=\ket{(K,\varsigma)}$,
then it holds that $\sigma_{(H,\varrho)}(x)=\sigma_{(K,\varsigma)}(x)$ for every $x\in E[l]$. Suppose next that 
$f,g\in\Hom_\varOmega([l],[m])$ and that for all hypergraphs $(H,\varrho)\in G_W[l]$ one has that 
$\scH_Ef\ket{(H,\varrho)}=\scH_Eg\ket{(H,\varrho)}$. Then, $\ket{G_Cf(H,\varrho)}=\ket{G_Cg(H,\varrho)}$
as a consequence of property \ceqref{whgsts16}. Consequently, 
$\sigma_{G_Cf(H,\varrho)}(x)=\sigma_{G_Cg(H,\varrho)}(x)$ for every $x\in E[l]$. The computation \ceqref{si2si} then implies that 
\begin{equation}
\label{chsmonadp2}
\sigma_{(H,\varrho)}(Ef^t(y))=\sigma_{(H,\varrho)}(Eg^t(y))\vphantom{\bigg]}
\end{equation}
for all hypergraphs $(H,\varrho)\in G_W[l]$ and $y\in E[m]$. We now consider the form this relation takes
for special choices of $(H,\varrho)$. Fix $r\in[l]$. Suppose that $H=\{X_r\}$, where $X_r=\{r\}$.
Suppose further that $w_r\in\msZ^{X_r}$ is the exponent function defined by $w_r(r)_x=1-\delta_{x,1}$ for $x\in\msR$. 
Finally, assume that $\varrho\in C(H)$ is the calibration defined by $\varrho_{X_r}(w)=\delta_{w,w_r}$ for
$w\in\msZ^{X_r}$. Then, relations  \ceqref{whgsts9} and \ceqref{chsmonadp2} together imply that 
\begin{equation}
\label{chsmonadp3}
\tr(y_{f(r)})=\tr(y_{g(r)})
\end{equation}
with $y\in E[m]$.  If by absurd $f(r)\neq g(r)$, for suitably chosen $y\in E[m]$ such that
$y_{f(r)}\neq y_{g(r)}$, we would have $\tr(y_{f(r)})\neq\tr(y_{g(r)})$, since the trace function
is linear and non identically vanishing. Thus, it necessarily holds that $f(r)=g(r)$. Given that  $r\in[l]$
is arbitrary, we have $f=g$. It follows from here that, as required, the mappings 
$f\in\Hom_\varOmega([l],[m])\mapsto\scC_Ef\in\Hom_{\ul{\rm Set}}(\scC_E[l],\scC_E[m])$, 
are all injective. 

By prop. 2.2.1 of I $\scC_E\varOmega$ is an $\varOmega$ category with $\scC_E$ as its stalk isofunctor and the layout
shown in items 1--7.

Next, we check that the definitions stated in items \cref{it:chsmonad3}, \cref{it:chsmonad4} are consistent.
For any two states $\ket{\varPhi}\in\scC_E[l]$,
$\ket{\varPsi}\in\scC_E[m]$, we have $\ket{\varPhi}\otimes\ket{\varPsi}\in\scC_E[l+m]$, by (4.2.11) of I \linebreak and
\ceqref{whgsts18}. Hence, \ceqref{cwhsmonad2} correctly defines a graded multiplicative structure.
Final\-ly, by \ceqref{whgsts19}, $\ket{0}\in\scC_E[0]$, making $\ket{0}$ an allowable graded multiplicative unit. 
We have now to show that items \cref{it:chsmonad3}, \cref{it:chsmonad4} define a graded $\varOmega$ monadic structure on 
$\scC_E\varOmega$ by verifying that conditions (2.3.1)--(2.3.3) of I are met. That 
\begin{equation}
\label{cejoint1}
\ket{\varPhi}\smallsmile(\ket{\varPsi}\smallsmile\ket{\varUpsilon})
=(\ket{\varPhi}\smallsmile\ket{\varPsi})\smallsmile\ket{\varUpsilon}
\end{equation}
for $\ket{\varPhi}\in\scC_E[l]$, $\ket{\varPsi}\in\scC_E[m]$, $\ket{\varUpsilon}\in \scC_E[n]$ and 
\begin{equation}
\label{cejoint2}
\ket{\varPhi}\smallsmile\ket{0}=\ket{0}\smallsmile\ket{\varPhi}=\ket{\varPhi}
\end{equation}
for $\ket{\varPhi}\in \scC_E[l]$ follows from \ceqref{cwhsmonad2} and the properties of vector
tensor multiplication, proving the first two properties. Using relation (4.2.6) of I, we find
{\allowdisplaybreaks
\begin{align}
\label{cejoint3}
\scC_Ef\smallsmile\scC_Eg(\ket{\varPhi}\smallsmile \ket{\varPsi})
&=\scC_E(f\smallsmile g)(\ket{\varPhi}\smallsmile\ket{\varPsi})
\\
\nonumber
&=\scH_E(f\smallsmile g)(\ket{\varPhi}\smallsmile\ket{\varPsi})
\\
\nonumber
&=\scH_Ef\smallsmile \scH_Eg(\ket{\varPhi}\smallsmile\ket{\varPsi})
\\
\nonumber
&=\scH_Ef\otimes\scH_Eg(\ket{\varPhi}\otimes\ket{\varPsi})
\\
\nonumber
&=\scH_Ef\ket{\varPhi}\otimes\scH_Eg\ket{\varPsi}=\scC_Ef(\ket{\varPhi})\smallsmile\scC_Eg(\ket{\varPsi})
\end{align}
}
\!\!for $f\in\Hom_\varOmega([l],[p])$, $g\in\Hom_\varOmega([m],[q])$, $\ket{\varPhi}\in\scC_E[l]$, $\ket{\varPsi}\in\scC_E[m]$,
showing in this way also the third property. 
\end{proof}



This is the right point to notice that weighted hypergraph states conceived in accordance to def. \cref{def:sigmaha}
enjoy basic properties analogous to properties \ceqref{whgsts16}, \ceqref{whgsts18} and \ceqref{whgsts19}
of calibrated hypergraph states whose relevance we have recalled at the beginning of this subsection.

\begin{prop} \label{prop:whsbasic}
For every $l\in\bbN$ and weighted hypergraph $(H,\alpha)\in G_W[l]$ and injective ordinal morphism $f\in\Hom_\varOmega([l],[m])$,
\begin{equation}
\label{w/whgsts16}
\scH_Ef\ket{(H,\alpha)}=\ket{G_Wf(H,\alpha)}
\end{equation}
(cf. subsects. 3.3, 4.2 of I). Furthermore, 
\begin{equation}
\label{w/whgsts18}
\ket{(H,\alpha)\smallsmile(K,\beta)}=\ket{(H,\alpha)}\smallsmile\ket{(K,\beta)}
\end{equation}
for every $l,m\in\bbN$ and weighted hypergraphs $(H,\alpha)\in G_W[l]$, $(K,\beta)\in G_W[m]$ and 
\begin{equation}
\label{w/whgsts19}
\ket{(O,\upsilon)}=\ket{0}
\end{equation}
(cf. subsect. 3.3 of I).
\end{prop}


\begin{proof}
The proof of relation \ceqref{w/whgsts16} proceeds along the same lines as that of  
relation \ceqref{whgsts16} of prop. \cref{prop:hefgcf} based on prop. \cref{prop:shefdgcfd}.
The reason why \ceqref{w/whgsts16} holds only for injective morphisms $f$ is that 
the weighted counterpart of the crucial identity \ceqref{si2si}, viz
$\sigma_{(H,\alpha)}(Ef^t(y))=\sigma_{G_Wf(H,\alpha)}(y)$,  can be shown only for such $f$. 
The proof of relations \ceqref{w/whgsts18}, \ceqref{w/whgsts19} is totally analogous to that of
the kindred relations \ceqref{whgsts18}, \ceqref{whgsts19} based on prop. \ceqref{prop:jntdhr}.
\end{proof}

\noindent
The restriction to injective ordinal morphisms is unavoidable: for a generic morphism $f$, $\scH_Ef\ket{(H,\alpha)}$
turns out to be a calibrated but generally non weighted hypergraph state. Weighted hypergraph states, therefore, are
distinguished from the calibrated ones by their reduced degree of covariance.

Relations \ceqref{w/whgsts16}--\ceqref{w/whgsts19} indicate that weighted hypergraph states are amenable to a graded 
monadic description analogous to that of their calibrated counterparts upon replacing the finite ordinal category $\varOmega$
with the wide subcategory $\varTheta$ of $\varOmega$ containing only the injective morphisms of this latter and proceeding
otherwise exactly as in the calibrated case. 


The category $\varTheta$ is a Pro category as $\varOmega$ is. 
$\varTheta$ categories, graded $\varTheta$ monads and their morphisms are defined exactly as 
$\varOmega$ categories, graded $\varOmega$ monads and their morphisms, respectively, by replacing
the ordinal category $\varOmega$ by its subcategory $\varTheta$
(cf. defs. 2.2.1, 2.3.1, 2.3.2 of I).
We observe that an $\varOmega$ category, respectively a graded $\varOmega$ monad, is automatically also a
$\varTheta$ category, respectively a graded $\varTheta$ monad, while the converse is generally false.

Prop. \cref{prop:whsbasic} leads to a graded $\varTheta$ monadic interpretation
of the weighted hypergraph state construction paralleling that of the calibrated hypergraph state construction
furnished by prop. \cref{prop:whsmapmor}.

\begin{prop} \label{prop:hsmapmor}
The specification for every $l\in\bbN$ of the function $\ket{-}:G_W[l]\rightarrow \scH_E[l]$
assigning to each weighted hypergraph $(H,\alpha)\in G_W[l]$ its corresponding
hypergraph state $\ket{(H,\alpha)}\in\scH_E[l]$
defines a distinguished morphism $\ket{-}\in\Hom_{\ul{\rm GM}_\varTheta}(G_C\varTheta,\scH_E\varTheta)$
of the objects $G_W\varTheta$, $\scH_E\varTheta\in\Obj_{\ul{\rm GM}_\varTheta}$ in $\ul{\rm GM}_\varTheta$.
\end{prop}

\begin{proof}
The proof is wholly analogous to that of prop. \cref{prop:whsmapmor}.
\end{proof}

Similarly, prop. \cref{prop:chsmonad} has an analog for weighted hypergraph states:
they organize in a graded $\varTheta$ monad. 

\begin{prop} \label{prop:whsmonad}
There exists a graded $\varTheta$ monad $\scW_E\varTheta$ defined analogously to the monad $\scC_E\varOmega$
replacing calibrated with weighted hypergraph states and the category $\varOmega$
with its subcategory $\varTheta$. 
\end{prop}

\begin{proof}
The proof proceeds along lines totally analogous to those of the proof of prop. 
\cref{prop:chsmonad} above. 
\end{proof}

\vfill\eject

\renewcommand{\sectionmark}[1]{\markright{\thesection\ ~~#1}}

\section{\textcolor{blue}{\sffamily Special results and applications}}\label{sec:techno}

In sect. \cref{sec:grstt}, we have worked out a general theory of calibrated hypergraph states
and shown that they are locally maximally entangleable stabilizer states, much as ordinary graph states. 


Leaving aside the fundamental issue of the classification of calibrated hypergraph states in
suitable entanglement classes, which we defer to future work,
in this subsection we look for selected interesting applications of the theory.
In subsect. \cref{subsec:tech}, we concentrate on the case where
the relevant Galois ring is a field and find alternative expressions of calibrated hypergraph states
exhibiting an interesting interplay with the theory of polynomials over finite fields. 
In subsect. \cref{subsec:quditcz}, focusing on  qudits of prime dimension, we present an alternative definition of hypergraph states 
built by controlled $Z$ gates, analogously to basic qubit case, and show that they are
in fact special instances of calibrated hypergraph states.



\subsection{\textcolor{blue}{\sffamily Calibrated hypergraph states of finite field qudits}}\label{subsec:tech}

In this subsection, we study the distinguished polynomial form calibrated hypergraph states take in the special but important
case where the independent states of the underlying qudits are indexed by a finite field $\msR$. 


When the Galois ring $\msR$ is a field, for any exponent $u\in\msZ$ of the associated cyclicity monoid
the power function $x\mapsto x^u$, $x\in\msR$, can be expressed as $x^u=m_u(x)$ for a 
polynomial $m_u(\sfx)\in\msR[\sfx]$ over $\msR$ depending on $u$. 
The phase function $\sigma_{(H,\varrho)}:E[l]\rightarrow\msP$ \linebreak of a calibrated hypergraph $(H,\varrho)\in G_C[l]$
given in \ceqref{whgsts9}, and so the associated hypergraph state $\ket{(H,\varrho)}\in\scH_E[l]$ by \ceqref{whgsts14}, 
are in this way expressible in terms of the configuration variable $x\in E[l]$ in a fully polynomial way. We have indeed 
\begin{equation}
\label{whgsts9/pol}
\sigma_{(H,\varrho)}(x)=\mycom{{}_\sss}{{}_{X\in H}}\mycom{{}_\sss}{{}_{w\in\msA^X}}\varrho_X(w)
\tr\left(\mycom{{}_\ppp}{{}_{r\in X}}m_{w(r)}(x_r)\right)\!.
\end{equation}
We provide next an algorithm for computing the polynomials $m_u(\sfx)$.
The results found rest on the Galois ring $\msR$ being a field and do not extend to a general Galois ring. 

Let $\msR=\GR(p,d)=\bbF_q$ be a Galois field with $q=p^d$ elements, where as usual $p,d\in\bbN_+$ with
$p\geq 2$, $p$ prime. The power matrix introduced next characterizes $\msR$. 

\begin{defi}
The power matrix of $\msR$ is the matrix $A\in\msR^{q\times q}$ given by 
\begin{equation}
\label{pwrexp}
A_{xk}=x^k
\end{equation}
with $x\in\msR$ and $k\in[q]$. 
\end{defi}

\begin{prop} \label{prop:ainv}
The power matrix $A$ is invertible in $\msR^{q\times q}$.
\end{prop}

\noindent
Although this result is well--known, we shall show it, since   
its proof proceeds by the explicit construction of the inverse matrix $A^{-1}$, 
which is of considerable usefulness.

\begin{proof}
From the theory of Galois fields, it is known that there exists a primitive element $\xi\in\msR^\times$
such that $\msR=\{0\}\cup\{\xi^j\hfpt|\hfpt j\in[q-1]\}$. Further, $\xi^{q-1}=1$ and 
\begin{equation}
\label{ainvp1}
\mycom{{}_\sss}{{}_{j\in [q-1]}}\xi^{\pm nj}=-\delta_{n,0}
\end{equation}
for $n\in[q-1]$. We recall here that $q=p^d=0$ in $\msR$ since $p$ is the characteristic of $\msR$.

As $x^k=\xi^{kj}$ for some $j\in[q-1]$ for $x\in\msR\setminus\{0\}$ and $k\in[q]$, $A$ has the block form
\begin{equation}
\label{ainvp2}
A=\bigg[
\begin{matrix}
1&0_{q-1}{}^t\\
e&S
\end{matrix}
\bigg]
\end{equation}
with respect to the natural indexing $(0,1,\xi,\ldots,\xi^{q-2})$ of $\msR$, 
where the $(q-1)\times 1$ and $(q-1)\times(q-1)$ blocks $e$ and $S$ are given by  
\begin{align}
\label{ainvp3}
&e_k=1,
\\
\label{ainvp4}
&S_{kl}=\xi^{k(l+1)}
\end{align}
with $k,l\in[q-1]$. Consider now the matrix $B\in\msR^{q\times q}$ given by 
\begin{equation}
\label{ainvp5}
B=\bigg[
\begin{matrix}1&0_{q-1}{}^t\\
f&T
\end{matrix}
\bigg],
\end{equation}
where the $(q-1)\times 1$ and $(q-1)\times(q-1)$ blocks $f$ and $T$ are given by  
\begin{align}
\label{ainvp6}
&f_k=-\delta_{k,q-2},
\\
\label{ainvp7}
&T_{kl}=-\xi^{(q-2-k)l}
\end{align}
for $k,l\in[q-1]$. Using relation \ceqref{ainvp1}, \pagebreak it is a matter of a straightforward calculation
to verify that $AB=BA=1_{q\times q}$ holds. It follows that $A$ is invertible, as claimed, with $A^{-1}=B$. 
\end{proof}

The following proposition makes precise what was anticipated at the beginning of this subsection.
The proposition furnishes an explicit expression of the polynomial $m_u(\sfx)$ and so 
is again of considerable usefulness.

\begin{prop} \label{prop:huxu}
For each exponent $u\in\msZ$, there is a unique polynomial $m_u(\sfx)\in\msR[\sfx]$
of degree at most $q-1$ with the property that 
\begin{equation}
\label{huxu1}
x^u=m_u(x)
\end{equation}
for $x\in\msR$. $m_u(\sfx)$ is given by
\begin{equation}
\label{huxu1/1}
m_u(\sfx)=\mycom{{}_\sss}{{}_{k\in [q]}}\bigg(\mycom{{}_\sss}{{}_{x\in\msR}}A^{-1}{}_{kx}x^u\bigg)\sfx^k. 
\end{equation}
\end{prop}

\begin{proof}
We first verify that the polynomial $m_u(\sfx)$ given by \ceqref{huxu1/1} has the required properties. 
First $m_u(\sfx)$ has degree at most $q-1$. Second, $m_u(\sfx)$ satisfies \ceqref{huxu1} as 
{\allowdisplaybreaks 
\begin{align}
\label{huxup1}
m_u(x)&=\mycom{{}_\sss}{{}_{k\in [q]}}\bigg(\mycom{{}_\sss}{{}_{y\in\msR}}A^{-1}{}_{ky}y^u\bigg)x^k
\\
\nonumber
&=\mycom{{}_\sss}{{}_{y\in\msR}}\bigg(\mycom{{}_\sss}{{}_{k\in [q]}}A_{xk}A^{-1}{}_{ky}\bigg)y^u
\\
\nonumber
&=\mycom{{}_\sss}{{}_{y\in\msR}}\delta_{x,y}y^u=x^u,
\end{align}
}
\!\!where \ceqref{pwrexp} was employed. Third, $m_u(\sfx)$ is the only such polynomial. 
Indeed, let $g(\sfx),h(\sfx)\in\msR[\sfx]$ be polynomials of degree at most $q-1$ such that $g(x)=h(x)=x^u$
for $x\in\msR$. 
Suppose that $g(\sfx)=\sum_{k\in [q]}a_k\sfx^k$, $h(\sfx)=\sum_{k\in [q]}b_k\sfx^k$, where $a_k,b_k\in\msR$. Then, 
{\allowdisplaybreaks 
\begin{align}
\label{huxup2}
0&=g(x)-h(x)
\\
\nonumber
&=\mycom{{}_\sss}{{}_{k\in [q]}}(a_k-b_k)x^k=\mycom{{}_\sss}{{}_{k\in [q]}}A_{xk}(a_k-b_k),
\end{align}
}
\!\!where \ceqref{pwrexp} was used again. Since the power matrix $A$ is invertible, $a_k=b_k$.
It follows that $g(\sfx)=h(\sfx)$, showing uniqueness. 
\end{proof}

We present a couple of examples which provide an explicit illustration of the techniques introduced above.

\begin{exa} \label{exa:huf3}
The Galois field $\msR=\GR(3,1)=\bbF_3$. 
{\rm For this field, the matrix $A_{}$ and its inverse $A^{-1}$ read as
\begin{equation}
\label{huf31}
A=\left[
\begin{matrix}
1&0&0\\
1&1&1\\
1&2&1
\end{matrix}
\right],
\hspace{1cm}
A^{-1}=\left[
\begin{matrix}
1&0&0\\
0&2&1\\
2&2&2
\end{matrix}
\right].
\end{equation}
Using \ceqref{huxu1/1}, we find e.g. 
\begin{align}
\label{huf32}
m_{(1,0,0)}(\sfx)&=\sfx^2,
\\
\nonumber
m_{(0,0,1)}(\sfx)&=1+\sfx+2\sfx^2,
\\
\nonumber
m_{(1,0,1)}(\sfx)&=\sfx, \qquad \text{etc.}
\end{align}
}
\end{exa}

\begin{exa} \label{exa:huf4}
The Galois field $\msR=\GR(2,2)=\bbF_4$. 
{\rm This field is described in ex. \cref{exa:gr22}. The matrix $A$ and its inverse $A^{-1}$ have the form 
\begin{equation}
\label{huf41}
A=\left[
\begin{matrix}
1&0&0&0\\
1&1&1&1\\
1&\theta&1+\theta&1\\
1&1+\theta&\theta&1
\end{matrix}
\right],
\hspace{1cm}
A^{-1}=\left[
\begin{matrix}
1&0&0&0\\
0&1&1+\theta&\theta\\  
0&1&\theta&1+\theta\\
1&1&1&1
\end{matrix}
\right].
\end{equation}
Employing \ceqref{huxu1/1}, we find e.g. 
\begin{align}
\label{huf42}
m_{(0,0,1,0)}(\sfx)&=1+\theta\sfx+\sfx^2+(1+\theta)\sfx^3,
\\
\nonumber
m_{(0,0,0,1)}(\sfx)&=1+(1+\theta)\sfx+\sfx^2+\theta\sfx^3,
\\
\nonumber
m_{(0,0,1,1)}(\sfx)&=1+\sfx+\sfx^3,\qquad\text{etc.}
\end{align}
}
\end{exa}

The universal polynomial of $\msR$ is the degree $q$ polynomial $e(\sfx)\in\msR[\sfx]$ defined by
\begin{equation}
\label{unipol1}
e(\sfx)=\mycom{{}_\ppp}{{}_{x\in\msR}}(\sfx-x).
\end{equation}
It can be shown that $e(\sfx)$ is given by 
\begin{equation}
\label{unipol2}
e(\sfx)=\sfx^q-\sfx. 
\end{equation}
Every element $x\in\sfR$ obeys the universal equation 
\begin{equation}
\label{unipol3}
e(x)=0.
\end{equation}
Since we deal throughout with polynomials $f(\sfx)\in\msR[\sfx]$ which eventually have to be evaluated
at generic elements $x\in\msR$, it is more natural to formulate the theory in terms  of the polynomial residue class ring
$\msR[\sfx]/e(\sfx)\msR[\sfx]$ rather than the ring $\msR[\sfx]$ itself.
The next considerations serve as an example of this.

For any two exponents $u,v\in\msZ$, one has $m_{u+v}(x)=m_u(x)m_v(x)$ for $x\in\msR$ by virtue of the property
\ceqref{huxu1} and the identity $x^{u+v}=x^ux^v$. But in spite of that, the mapping $u\in\msZ\mapsto m_u(\sfx)\in\msR[\sfx]$
is not a morphism of the cyclicity monoid $\msZ$ into the multiplicative monoid of the ring $\msR[\sfx]$. Indeed,
for a generic exponent pair $u,v\in\msZ$ the polynomials $m_{u+v}(\sfx)$, $m_u(\sfx)m_v(\sfx)$ differ in general,
since the former has degree at most $q-1$ while the latter has degree generally greater than $q-1$. 
The morphism property nevertheless still holds in the sense stated in the following proposition. 


\begin{prop} \label{prop:huhv}
The map $u\in\msZ\mapsto m_u(\sfx)+e(\sfx)\msR[\sfx]\in\msR[\sfx]/e(\sfx)\msR[\sfx]$ is a morphism of the cyclicity monoid
$\msZ$ into the multiplicative monoid of the ring $\msR[\sfx]/e(\sfx)\msR[\sfx]$.
\end{prop}

\begin{proof} Let $f(\sfx)\in\msR[\sfx]$ be any polynomial over $\msR$. Then, 
$f(x)=0$ for $x\in\msR$ if and only if $e(\sfx)$ divides $f(\sfx)$ in $\msR[\sfx]$,
that is if and only if $f(\sfx)+e(\sfx)\msR[\sfx]=e(\sfx)\msR[\sfx]$. This follows
readily from \ceqref{unipol1} and the polynomial division algorithm.
  
For every pair of exponents $u,v\in\msZ$, one has 
\begin{equation}
\label{huhv0}
m_{u+v}(x)=m_u(x)m_v(x)
\end{equation}
for $x\in\msR$, as already observed, as a consequence of the identity $x^{u+v}=x^ux^v$
and \ceqref{huxu1}. It follows that 
\begin{align}
\label{huhv1}
&m_{u+v}(\sfx)+e(\sfx)\msR[\sfx]=m_u(\sfx)m_v(\sfx)+e(\sfx)\msR[\sfx]
\\
\nonumber
&\hspace{5cm}=(m_u(\sfx)+e(\sfx)\msR[\sfx])(m_v(\sfx)+e(\sfx)\msR[\sfx]).
\end{align}
Furthermore, it holds that \hphantom{xxxxxxxxxx}
\begin{equation}
\label{huhv2}
m_0(x)=1
\end{equation}
for $x\in\msR$ by \ceqref{huxu1} again since $x^0=1$ for $x\in\msR$.
It follows that 
\begin{equation}
\label{huhv3}
m_0(\sfx)+e(\sfx)\msR[\sfx]=1+e(\sfx)\msR[\sfx]. 
\end{equation}
\ceqref{huhv1}, \ceqref{huhv3} imply the statement 
\end{proof}

Though the polynomials $m_{u+v}(\sfx)$, $m_u(\sfx)m_v(\sfx)$ may differ for a given exponent pair $u,v\in\msZ$, 
it is still possible to obtain the former from the latter by the following lemma.

\begin{lemma}
For every polynomial $f(\sfx)\in\msR[\sfx]$, there exists precisely one polynomial $\ol{f}(\sfx)\in\msR[\sfx]$
of degree at most $q-1$ such that $f(x)=\ol{f}(x)$ for $x\in\msR$. $\ol{f}(\sfx)$ is the remainder of the
division of $f(\sfx)$ by the universal polynomial $e(\sfx)$. 
\end{lemma}

\begin{proof}
Since the universal polynomial $e(\sfx)$ has degree $q$ by \ceqref{unipol2},
we can decompose any polynomial $f(\sfx)\in\msR[\sfx]$ as 
\begin{equation}
\label{}
f(\sfx)=\ol{f}(\sfx)+e(\sfx)q(\sfx)
\end{equation}
for certain polynomials $\ol{f}(\sfx),q(\sfx)\in\msR[\sfx]$ with $\ol{f}(\sfx)$
of degree at most $q-1$, by means of the polynomial division algorithm.
It follows that $f(x)=\ol{f}(x)$ for $x\in\msR$ by the universal equation \ceqref{unipol3}.
This shows the existence of $\ol{f}(\sfx)$.

Let $f(x)=\ol{f}_1(x)=\ol{f}_1(x)$ for $x\in\msR$ for two polynomials
$\ol{f}_1(\sfx),\ol{f}_2(\sfx)\in\msR[\sfx]$ of degree at most $q-1$. Then,
by a verification completely analogous to \ceqref{huxup2} in the proof of
prop. \cref{prop:huxu}, we find that $\ol{f}_1(\sfx)=\ol{f}_2(\sfx)$ necessarily.
\end{proof}

\noindent
Therefore, $m_{u+v}(\sfx)$ is the remainder of the division of $m_u(\sfx)m_v(\sfx)$
by the universal polynomial $e(\sfx)$. This property allows one to obtain the 
polynomials $m_u(\sfx)$ with generic $u\in\msZ$ once the polynomials $m_{u_i}(\sfx)$
of a subset $u_i\in\msZ$ generating $\msZ$ is known.

For every $x\in\msR$, $x\neq 1$, the cyclic monoid $\msZ_x$ is generated by the element
$1\in\msZ_x$ (cf. subsect. \cref{subsec:cycmongal}).
This property does not hold for cyclic monoid $\msZ_1$, which is trivial.
Therefore, by virtue of \ceqref{cycmongal6}, the cyclicity monoid $\msZ$
is generated by the special elements $s(x)=(s(x)_y)_{y\in\msR}\in\msZ$, $x\in\msR$, given by 
\begin{equation}
\label{crxymp2}
s(x)_y=\delta_{x,y}-\delta_{x,1}\delta_{1,y}.
\end{equation}  
Note indeed that for $x\neq 1$ $s(x)_y$ takes the value $1$ for $y=x$ and vanishes else, while for $x=1$
$s(x)_y$ vanishes identically. As a consequence, the polynomials $m_u(\sfx)\in\msR[\sfx]$, $u\in\msZ$, are
all expressible as products of the $m_{s(x)}(\sfx)$, $x\in\msR$, modulo $e(\sfx)\msR[\sfx]$. 


\begin{exa} \label{exa:chuf3}
The Galois field $\msR=\GR(3,1)=\bbF_3$. 
{\rm The setting considered here is the same as that of ex. \cref{exa:huf3}. The cyclicity monoid element
$(1,0,0)\in\msZ$ satisfies $(1,0,0)+(1,0,0)=(1,0,0)$. Therefore we should have
\begin{equation}
\label{}
m_{(1,0,0)}(\sfx)m_{(1,0,0)}(\sfx)=m_{(1,0,0)}(\sfx) \quad\text{mod}~\sfx^3-\sfx.
\end{equation}  
This can be indeed verified from the \ceqref{huf32}. 
}
\end{exa}

\begin{exa} \label{exa:chuf4}
The Galois field $\msR=\GR(2,2)=\bbF_4$. 
{\rm The setting is the same as that of ex. \cref{exa:huf4}. The elements $(0,0,1,0),(0,0,0,1)\in\msZ$
of the cyclicity monoid satisfy $(0,0,1,0)+(0,0,0,1)=(0,0,1,1)$. This relation implies that
\begin{equation}
\label{}
m_{(0,0,1,0)}(\sfx)m_{(0,0,0,1)}(\sfx)=m_{(0,0,1,1)}(\sfx) \quad\text{mod}~\sfx^4-\sfx,
\end{equation}  
a relation that can be checked using the \ceqref{huf42}.
}
\end{exa}

The basic power matrix introduced next is another key element of the theory. 

\begin{defi}
The basic power matrix of $\msR$ is the matrix $C\in\msR^{q\times q}$ given by
\begin{equation}
\label{crxymp1}
C_{xy}=x^{s(y)}
\end{equation}
with $x,y\in\msR$, where the cyclicity group element $s(y)\in\msZ$ is given by \ceqref{crxymp2}.
\end{defi}

\begin{prop} \label{prop:cinv}
The basic power matrix $C$ is symmetric and invertible in $\msR^{q\times q}$. 
\end{prop}

\begin{proof}
Combining \ceqref{crxymp2} and \ceqref{crxymp1}, we find
\begin{equation}
\label{cinvp1}
C_{xy}=x^{\delta_{y,x}-\delta_{y,1}\delta_{1,x}}
\end{equation}
for $x,y\in\msR$. Thus, $C_{xy}=1$ when $x\neq y$ and so
$C_{xy}=C_{yx}$, showing the symmetry of the matrix $C$. 
Next, suppose that for every $x\in\msR$
\begin{equation}
\label{crxymp3}
\mycom{{}_\sss}{{}_{y\in\msR}}C_{xy}a_y=0
\end{equation}
for certain $a_y\in\msR$. Then, owing to \ceqref{cinvp1}, 
\begin{equation}
\label{crxymp4}
xa_x-a_x+\mycom{{}_\sss}{{}_{y\in\msR}}a_y=0. 
\end{equation}
Hence, $xa_x-a_x$ is independent from $x$, implying that 
\begin{equation}
\label{crxymp5}
(x-1)a_x=(x-1)a_x|_{x=1}=0.
\end{equation}
Thus, $a_x=0$ for $x\neq 1$. Relation \ceqref{crxymp3} then reduces to
\begin{equation}
\label{crxymp6}
C_{x1}a_1=0.
\end{equation}
For all $x\in\msR$. Since $C_{11}=1$, we have $a_1=0$ too. We conclude that $a_x=0$ 
for $x\in\msR$, showing the non singularity of the matrix $C$. 
\end{proof}

The relevance of the polynomials $m_{s(y)}(\sfx)$ in our theory is made manifest also by the next result. 

\begin{prop} \label{prop:taufexp}
Suppose that $f(\sfx)\in\msR[\sfx]$ is any polynomial of degree at most $q-1$, $f(\sfx)=\sum_{k\in[q]}c_k\sfx^k$
with $c_k\in\msR$. Then, there are constants $c_y\in\msR$, $y\in\msR$, such that 
\begin{equation}
\label{taufexp1}
f(x)=\mycom{{}_\sss}{{}_{y\in\msR}}c_ym_{s(y)}(x)
\end{equation}
for $x\in\msR$. The $c_y\in\msR$ are uniquely determined and explicitly given by
\begin{equation}
\label{taufexp2}
c_y=\mycom{{}_\sss}{{}_{k\in[q]}}\mycom{{}_\sss}{{}_{z\in\msR}}C^{-1}{}_{yz}A_{zk}c_k.
\end{equation}
\end{prop}

\begin{proof}
Consider the vector space $\msV_\msR$ over $\msR$ of the polynomial functions $\tau_f(x)=f(x)$, $x\in\msR$, 
with $f(\sfx)\in\msR[\sfx]$ a polynomial over $\msR$ of degree at most $q-1$.

The polynomials $f_k(\sfx)=\sfx^k\in\msR[\sfx]$, $k\in[q]$, have degree at most $q-1$. 
The associated polynomial functions $\tau_{f_k}$ are given by $\tau_{f_k}(x)=x^k=A_{xk}$ by \ceqref{pwrexp}.
The $\tau_{f_k}$ are therefore linearly independent in $\msV_\msR$, since the power matrix $A$ is non singular
by prop. \cref{prop:ainv}, 
and moreover they clearly span $\msV_\msR$. The $\tau_{f_k}$ constitute thus a basis of $\msV_\msR$. Given that
there are $q$ of them, $\dim\msV_\msR=q$. 

The polynomials $m_{s(y)}(\sfx)\in\msR[\sfx]$, $y\in\msR$, also have degree at most $q-1$. 
The associated polynomial functions $\tau_{m_{s(y)}}$ are given by $\tau_{m_{s(y)}}(x)=m_{s(y)}(x)=x^{s(y)}=C_{xy}$
by \ceqref{huxu1} and \ceqref{pwrexp}. The $\tau_{m_{s(y)}}$ are therefore linearly independent too, by the
non singularity of the basic power matrix $C$ established in prop. \cref{prop:cinv}, 
and since there are $q$ of them they also constitute a basis of $\msV_\msR$. 

From the above discussion, it follows that the polynomial function $\tau_f\in\msV_\msR$ associated with any
polynomial $f(\sfx)\in\msR[\sfx]$ of degree at most $q-1$ can be expressed as a linear combination
of the $\tau_{m_{s(y)}}$, $y\in\msR$, 
\begin{equation}
\label{}
\tau_f=\mycom{{}_\sss}{{}_{y\in\msR}}c_y\tau_{m_{s(y)}}.
\end{equation}
Uniqueness of this expansion is obvious, because the function $\tau_{m_{s(y)}}$ form a basis of the vector space $\msV_\msR$. 
We have only to verify that \ceqref{taufexp1} holds with the constants $c_y$ given by \ceqref{taufexp2}. Let $x\in\msR$. Then,
by \ceqref{pwrexp} and \ceqref{crxymp1}, we have   
{\allowdisplaybreaks 
\begin{align}
\label{}
\mycom{{}_\sss}{{}_{y\in\msR}}c_ym_{s(y)}(x)
&=\mycom{{}_\sss}{{}_{y\in\msR}}\mycom{{}_\sss}{{}_{k\in[q]}}\mycom{{}_\sss}{{}_{z\in\msR}}C^{-1}{}_{yz}A_{zk}a_k x^{s(y)}
\\
\nonumber
&=\mycom{{}_\sss}{{}_{y\in\msR}}\mycom{{}_\sss}{{}_{k\in[q]}}\mycom{{}_\sss}{{}_{z\in\msR}}C_{xy}C^{-1}{}_{yz}a_kz^k
\\
\nonumber
&=\mycom{{}_\sss}{{}_{k\in[q]}}\mycom{{}_\sss}{{}_{z\in\msR}}\delta_{xz}a_kz^k
\\
\nonumber
&=\mycom{{}_\sss}{{}_{k\in[q]}}a_kx^k=f(x),
\end{align}
}
\!\! as required. 
\end{proof}

\noindent 
This results holds in particular when $f(\sfx)=m_u(\sfx)$ for some cyclicity monoid exponent
$u\in\msZ$ allowing one to further upgrade  expression \ceqref{whgsts9/pol} of the phase
function $\sigma_{(H,\varrho)}$ of a calibrated hypergraph $(H,\varrho)\in G_C[l]$
by expressing the polynomials $m_{w(r)}(\sfx)$ appearing in it in terms of the
$m_{s(y)}(\sfx)$ according to \ceqref{taufexp1}. 

The inverse $C^{-1}$ of the basic power matrix $C$ so appears in relevant results   
of the theory. Albeit there exists no simple formula for it, 
$C^{-1}$ can be computed straightforwardly in a few simple cases. 

\begin{exa} \label{exa:huf3ctd}
The Galois field $\msR=\GR(3,1)=\bbF_3$. 
{\rm For this field, the matrix $C$ and its inverse $C^{-1}$ read as 
\begin{equation}
\label{huf3ctd1}
C=\left[
\begin{matrix}
0&1&1\\
1&1&1\\
1&1&2
\end{matrix}
\right],
\hspace{1cm}
C^{-1}=\left[
\begin{matrix}
2&1&0\\
1&1&2\\
0&2&1
\end{matrix}
\right].
\end{equation}
}
\end{exa}

\begin{exa} \label{exa:huf4ctd}
The Galois field $\msR=\GR(2,2)=\bbF_4$. 
{\rm 
For this field (cf. ex. \cref{exa:gr22}), the matrix $C$ and its inverse $C^{-1}$ have the form 
\begin{equation}
\label{huf4ctd1}
C=\left[
\begin{matrix}
0&1&1&1\\
1&1&1&1\\
1&1&\theta&1\\
1&1&1&1+\theta
\end{matrix}
\right],
\hspace{1cm}
C^{-1}=\left[
\begin{matrix}
1&1&0&0\\
1&1&\theta&1+\theta\\
0&\theta&\theta&0\\
0&1+\theta&0&1+\theta
\end{matrix}
\right].
\end{equation}
\vspace{-4mm}}
\end{exa}


\subsection{\textcolor{blue}{\sffamily Calibrated hypergraph states assembled with controlled Z gates}}\label{subsec:quditcz}

In this subsection, applying the theory elaborated in subsect. \cref{subsec:tech} to qudits of prime dimension, 
we present an interesting alternative definition of hypergraph states as states built by controlled $Z$ gates,
analogously to basic qubit case, extending the construction of ref \ccite{Giri:2024qtt}.
We show further that these states are special instances of calibrated hypergraph states,
providing further evidence to the wide range of applicability of the methods developed in the present work.

Suppose now that $r=d=1$ so that $\msR=\GR(p,1)=\bbF_p$ is the prime field of characteristic $p$
and further $\msP=\bbF_p$ too. Fix a reference element $x^*\in\msR$, the standard choice being
$x^*=p-1\in\bbF_p$. 

Let $l\in\bbN$, $l\geq 2$, and let $X\subseteq [l]$ be a hyperedge such that $|X|\geq 2$. 
A marking of $X$ consists in a choice a distinguished vertex $r_X\in X$. The pair $(X,r_X)$
is then a marked hyperedge. We let $H_M[l]$ denote the set of all such marked hyperedges. 

\begin{defi}
The controlled phase operator $CZ_{(X,r_X)}\in\End_{\bfsfH}(\scH_E[l])$ of the  marked hyperedge
$(X,r_X)\in H_M[l]$ is defined by the expression 
\begin{equation}
\label{qcz1}
CZ_{(X,r_X)}=\mycom{{}_\sss}{{}_{x\in E[l]}}F_l{}^+\ket{x}\,\omega^{\sigma_{(X,r_X)}(x)}\bra{x}F_l,
\end{equation}
where $\omega=\exp(2\pi i/p)$ and $\sigma_{(X,r_X)}:E[l]\rightarrow \msP$ is the phase function
\begin{equation}
\label{qcz2}
\sigma_{(X,r_X)}(x)=x_{r_X}\mycom{{}_\ppp}{{}_{r\in X\setminus \{r_X\}}}\delta_{x_r,x^*}.
\end{equation}
\end{defi}

\noindent
The analogy to the standard qubit controlled phase operator should be evident. 
The vertex $r_X\in X$ is the target vertex while the remaining vertices $r\in X\,\setminus\,\{r_X\}$
are the control vertices.  $CZ_{(X,r_X)}F_l{}^+\ket{x}=F_l{}^+\ket{x}$ unless all control qudits are
in the state $\ket{x^*}\in\scH_1$, in which case we have $CZ_{(X,r_X)}F_l{}^+\ket{x}=F_l{}^+\ket{x}\,\omega^{x_{r_X}}$, 
the phase factor $\omega^{x_{r_X}}$ being so  determined by the state $\ket{x_{r_X}}\in\scH_1$ of the target
qudit  \ccite{Giri:2024qtt}. 

The following proposition establishes that the phase function $\sigma_{X,r_X}$ has a polynomial expression. 

\begin{prop} \label{prop:sxrpol}
Let $(X,r_X)\in H_M[l]$ be a marked hyperedge. Then  
\begin{equation}
\label{qcz4}
\sigma_{(X,r_X)}(x)=x_{r_X}\mycom{{}_\ppp}{{}_{r\in X\setminus \{r_X\}}}p(x_r)
\end{equation}
for $x\in E[l]$, where $p(\sfx)\in\msR[\sfx]$ is the degree at most $q-1$ polynomial 
\begin{equation}
\label{qcz5}
p(\sfx)=\mycom{{}_\sss}{{}_{k\in[q]}}A^{-1}{}_{kx^*}\sfx^k,
\end{equation}
where $A$ is the invertible power matrix displayed in eq.  \ceqref{pwrexp}.
\end{prop}

\begin{proof} From eqs. \ceqref{qcz2} and \ceqref{qcz4}, it is clear that we have only to show that 
\begin{equation}
\label{}
p(x)=\delta_{x,x^*}
\end{equation}
for $x\in\msR$. Indeed, recalling \ceqref{pwrexp}, we have 
\begin{align}
\label{}
p(x)&=\mycom{{}_\sss}{{}_{k\in[q]}}A^{-1}{}_{kx^*}x^k
\\
\nonumber
&=\mycom{{}_\sss}{{}_{k\in[q]}}A_{xk}A^{-1}{}_{kx^*}=\delta_{x,x^*},
\end{align}
as required. 
\end{proof}

A marked hypergraph $(H,r_\cdot)$ is a hypergraph
$H\in G[l]$ together with a marking $r_X\in X$ of each hyperedge $X\in H$.
We let $G_M[l]$ be the set of all such hypergraphs.

\begin{defi}
The hypergraph state associated with a marked hypergraph $(H,r_\cdot)$ is 
\begin{equation}
\label{qcz3}
\ket{(H,r_\cdot)}=\mycom{{}_\ppp}{{}_{X\in H}}CZ_{(X,r_X)}\ket{0_l}.
\end{equation}
\end{defi}

\noindent
We shall call $\ket{(H,r_\cdot)}$ a marked hypergraph state.
The definition above closely parallels the standard one of qubit graph states and extends that provided 
in \ccite{Giri:2024qtt}: $\ket{(H,r_\cdot)}$ is in fact built by the controlled
phase operators $CZ_{(X,r_X)}$. 

Combining \ceqref{qcz1} and \ceqref{qcz3}, one finds that $(H,r_\cdot)$ can be expressed as 
\begin{equation}
\label{qcz6}
\ket{(H,r_\cdot)}=D_{(H,r_\cdot)}\ket{0_l},
\end{equation}
where the operator $D_{(H,r_\cdot)}\in\End(\scH_E[l])$ is given by 
\begin{equation}
\label{qcz7}
D_{(H,r_\cdot)}=\mycom{{}_\sss}{{}_{x\in E[l]}}F_l{}^+\ket{x}\,\omega^{\sigma_{(H,r_\cdot)}(x)}\bra{x}F_l,
\end{equation}
the phase function $\sigma_{(H,r_\cdot)}:E[l]\rightarrow \msP$ being given by 
\begin{equation}
\label{qcz8}
\sigma_{(H,r_\cdot)}=\mycom{{}_\sss}{{}_{X\in H}}\sigma_{(X,r_X)}.
\end{equation}
Upon comparing relations \ceqref{qcz2},
\ceqref{qcz6}--\ceqref{qcz8} and \ceqref{whgsts10}, \ceqref{whgsts9}, \ceqref{whgsts14}, it is apparent
that the marked hypergraph state $(H,r_\cdot)$ has an overall formal structure completely analogous to that
of a calibrated hypergraph state. 
Indeed, every marked hypergraph state is an instance of calibrated hypergraph state, as the next proposition
establishes.

\begin{prop} \label{prop:hreqhrho}
Let $(H,r_\cdot)\in G_M[l]$ a marked hypergraph. Then, there exists a calibration
$\varrho\in C(H)$ such that $\ket{(H,r_\cdot)}=\ket{(H,\varrho)}$. 
\end{prop}

\begin{proof}
The proof relies on the following lemma. Consider again the polynomial $p(\sfx)$ displayed in \ceqref{qcz5}. 

\begin{lemma} The identity 
\begin{equation}
\label{pxexp1}
p(x)=\mycom{{}_\sss}{{}_{y\in\msR}}C^{-1}{}_{yx^*}m_{s(y)}(x) 
\end{equation}
holds for $x\in\msR$, 
where $C$ is the invertible basic power matrix shown in  eq. \ceqref{crxymp1}.
\end{lemma}

\begin{proof} Prop. \cref{prop:taufexp} guarantees that
\begin{equation}
\label{pxexp1/1}
p(x)=\mycom{{}_\sss}{{}_{y\in\msR}}c_ym_{s(y)}(x) 
\end{equation}
for unique constants $c_y\in\msR$. We have only to compute these latter. 
From combining \ceqref{taufexp2} and \ceqref{qcz5}, we obtain 
{\allowdisplaybreaks
\begin{align}
\label{}
c_y&=\mycom{{}_\sss}{{}_{k\in[q]}}\mycom{{}_\sss}{{}_{z\in\msR}}C^{-1}{}_{yz}A_{zk}A^{-1}{}_{kx^*}
\\
\nonumber
&=\mycom{{}_\sss}{{}_{z\in\msR}}C^{-1}{}_{yz}\delta_{z,x^*}=C^{-1}{}_{yx^*},
\end{align}
}
\!\!leading to the stated result. 
\end{proof}

\noindent
We can now complete the proof of the proposition. Combining \ceqref{qcz4}, \ceqref{qcz8} and \ceqref{pxexp1},
we obtain 
{\allowdisplaybreaks
\begin{align}
\label{}
\sigma_{(H,r_\cdot)}(x)&=\mycom{{}_\sss}{{}_{X\in H}}
x_{r_X}\mycom{{}_\ppp}{{}_{r\in X\setminus \{r_X\}}}\Big(\mycom{{}_\sss}{{}_{y\in\msR}}C^{-1}{}_{yx^*}m_{s(y)}(x_r)\Big)
\\
\nonumber
&=\mycom{{}_\sss}{{}_{X\in H}}\mycom{{}_\sss}{{}_{y_r\in\msR,\,r\in X\setminus\{r_X\}}}x_{r_X}
\mycom{{}_\ppp}{{}_{r\in X\setminus\{r_X\}}}C^{-1}{}_{y_rx^*}x_r{}^{s(y_r)}
\end{align}
}
\!\!for $x\in E[l]$. Let $\varrho\in C(H)$ be the calibration defined as follows:
\begin{equation}
\label{}
\varrho_X(w)=\left\{
\begin{array}{ll}
\prod_{r\in X\setminus\{r_X\}}C^{-1}{}_{y_rx^*}&\text{if $w(r)=s(y_r)$ for some $y_r\in\msR$ with}\\
\hphantom{\prod_{r\in X\setminus\{r_X\}}C^{-1}{}_{y_rx^*}}&\text{\hspace{1.8cm} $r\in X\setminus\{r_X\}$ and $w(r_X)=s^*$},\\
0&\text{else} 
\end{array}\right.
\end{equation}
for $X\in H$ and $w\in\msA^X$, where $s^*\in\msZ$ is the special element
\begin{equation}
\label{}
s^*{}=\mycom{{}_\sss}{{}_{y\in\msR}}s(y)
\end{equation}
of the cyclicity monoid with  the property that
$x^{s^*}=x$ for all $x\in\msR$. We can express $\sigma_{H,r_\cdot}(x)$ through $\varrho$ finding 
\begin{equation}
\label{}
\sigma_{(H,r_\cdot)}(x)
=\mycom{{}_\sss}{{}_{X\in H}}\mycom{{}_\sss}{{}_{w\in\msA^X}}\varrho_X(w)\mycom{{}_\ppp}{{}_{r\in X}}x_r{}^{w(r)}
=\sigma_{(H,\varrho)}
\end{equation}
by virtue of \ceqref{whgsts9}. Inserting this relation in \ceqref{qcz7}, we find that 
\begin{equation}
\label{}
D_{(H,r_\cdot)}=D_{(H,\varrho)},
\end{equation}
owing to \ceqref{whgsts10}. The claimed identity 
now follows readily from \ceqref{whgsts14} \ceqref{qcz6}
\end{proof}

\begin{exa} Three qutrit marked hypergraph states.
{\rm We consider again the three qutrit setting studied in ex. \cref{exa:3qtchg}.
Consider once more the hyperedges  
$X^0=\{0,1\}$, $X^1=\{1,2\}$, $X^2=\{0,2\}$ and $X^3=\{0,1,2\}$ of $[3]$
and endow them with the markings $r^0=1$, $r^1=2$, $r^2=0$, $r^3=2$, respectively. 
Build then the marked hypergraphs $(H^i,r^i{}_\cdot)\in G_M[3]$, $i=\sfa,\sfb,\sfc,\sfd,\sfe$
specified by 
{\allowdisplaybreaks
\begin{align}
\label{3qtmhgp3}
&H^\sfa=\{X^3\},&&r^\sfa{}_\cdot=\{r^3\},
\\
\nonumber
&H^\sfb=\{X^0,X^1\},&&r^\sfb{}_\cdot=\{r^0,r^1\},
\\
\nonumber 
&H^\sfc=\{X^0,X^1,X^2\},&&r^\sfc{}_\cdot=\{r^0,r^1,r^2\},
\\
\nonumber
&H^\sfd=\{X^0,X^1,X^3\},&&r^\sfd{}_\cdot=\{r^0,r^1,r^3\},
\\
\nonumber
&H^\sfe=\{X^0,X^1,X^2,X^3\},&&r^\sfe{}_\cdot=\{r^0,r^1,r^2,r^3\}. 
\end{align}
}
\!\!With the $(H^i,r^i{}_\cdot)$ there are associated marked hypergraph states $\ket{(H^i,r^i{}_\cdot)}$
by virtue of the prescription \ceqref{qcz3}. We notice now that
the marked hypergraphs $(H^i,r^i{}_\cdot)$ have the same underlying hypergraphs $H^i$ as
the calibrated hypergraphs $(H^i,\varrho^i)\in G_C[3]$ displayed in \ceqref{3qtchgp3}.
Indeed, the $\ket{(H^i,\varrho^i)}$ are precisely calibrated hypergraph states 
corresponding to the marked hypergraph states $\ket{(H^i,r^i{}_\cdot)}$
according to prop. \cref{prop:hreqhrho}
\begin{equation}
\label{}
\ket{(H^i,r^i{}_\cdot)}=\ket{(H^i,\varrho^i)}. 
\end{equation}
Underlying this identity of these hypergraph states is the identity
$\sigma_{(H^i,r^i{}_\cdot)}=\sigma_{(H^i,\varrho^i)}$ of the corresponding phase functions.
The polynomial expression of the $\sigma_{(H^i,r^i{}_\cdot)}$, which can be easily obtained from
combining \ceqref{qcz4}, \ceqref{qcz5} and \ceqref{qcz8},
coincides in this way to that of the $\sigma_{(H^i,\varrho^i)}$ based on eq. \ceqref{whgsts9/pol}. By
systematically employing prop. \cref{prop:sxrpol}, we find the following expressions
{\allowdisplaybreaks
\begin{align}
\label{}
\sigma_{(H^\sfa,r^\sfa{}_\cdot)}(x_0,x_1,x_2)&=\sigma_{(H^\sfa,\varrho^\sfa)}(x_0,x_1,x_2)=(x_0+2x_0{}^2)(x_1+2x_1{}^2)x_2,
\\
\nonumber
\sigma_{(H^\sfb,r^\sfb{}_\cdot)}(x_0,x_1,x_2)&=\sigma_{(H^\sfb,\varrho^\sfb)}(x_0,x_1,x_2)=(x_0+2x_0{}^2)x_1+(x_1+2x_1{}^2)x_2,
\\
\nonumber
\sigma_{(H^\sfc,r^\sfc{}_\cdot)}(x_0,x_1,x_2)&=\sigma_{(H^\sfc,\varrho^\sfc)}(x_0,x_1,x_2)
\\
\nonumber
&\hspace{2.6cm}=(x_0+2x_0{}^2)x_1+(x_1+2x_1{}^2)x_2+(x_2+2x_2{}^2)x_0,
\\
\nonumber
\sigma_{(H^\sfd,r^\sfd{}_\cdot)}(x_0,x_1,x_2)&=\sigma_{(H^\sfd,\varrho^\sfd)}(x_0,x_1,x_2)
\\
\nonumber
&\hspace{.5cm}=(x_0+2x_0{}^2)x_1+(x_1+2x_1{}^2)x_2+(x_0+2x_0{}^2)(x_1+2x_1{}^2)x_2,
\\
\nonumber
\sigma_{(H^\sfe,r^\sfe{}_\cdot)}(x_0,x_1,x_2)&=\sigma_{(H^\sfe,\varrho^\sfe)}(x_0,x_1,x_2)
\\
\nonumber
&\hspace{1cm}=(x_0+2x_0{}^2)x_1+(x_1+2x_1{}^2)x_2+(x_2+2x_2{}^2)x_0
\\
\nonumber
&\hspace{6.3cm}+(x_0+2x_0{}^2)(x_1+2x_1{}^2)x_2. 
\end{align}
}
\!\!This calculation shows also that the hypergraph states $\ket{(H^i,r^i{}_\cdot)}$, $|(H^i,\varrho^i)\rangle$
are not weighted, since their phase functions  $\sigma_{(H^i,r^i{}_\cdot)}$, $\sigma_{(H^i,\varrho^i)}$ are given  by
polynomials of the qutrit variables $x_0,x_1,x_2$ containing second powers of these latter which cannot appear
in any weighted hypergraph state. 
}
\end{exa}


\vfill\eject

\vfill\eject

\appendix

\renewcommand{\sectionmark}[1]{\markright{\thesection\ ~~#1}}

\section{\textcolor{blue}{\sffamily Basic Galois ring theory}}\label{app:galois}

In this appendix section, we review basic notions of theory of Galois rings used in the main text.
Standard references for these topics are \ccite{Wan:2011gfr,Bini:2002gfr}.


\subsection{\textcolor{blue}{\sffamily Galois fields and rings}}\label{subsec:galois}

In this appendix, we review the main definitions and results of the theory of Galois fields and rings.
We assume that the reader has some familiarity with the basic notions of field and ring theory. 

Galois field and ring theory employs extensively polynomial rings and quotients of polynomial rings by polynomial 
ideals. We provide some background on this topic first because of its relevance. 
Let $\msK$ be a commutative ring with unity. 
A polynomial $h(\sfx)$ over $\msK$ is an expression of the form \hphantom{xxxxxxxx}
\begin{equation}
\label{poly}
h(\sfx)=\mycom{{}_\sss}{{}_{0\leq i\leq d}}a_i\sfx^i
\end{equation}
with $a_i\in\msK$ for $0\leq i\leq d$ and $a_d\neq0$ when $d>0$, where $\sfx$ is a formal indeterminate
and the $\sfx^i$ are formal powers of $\sfx$. The $a_i$ and $d$ are called respectively the coefficients and the degree
of $h(\sfx)$.  $h(\sfx)$ is said to be monic if $a_d=1$. 

The set of all polynomials over $\msK$ with the obvious operations of addition and multiplication and
additive and multiplicative unities constitutes a ring, called the polynomial ring $\msK[\sfx]$ over $\msK$.
$\msK$ can be identified with the subring of $\msK[\sfx]$ of degree $0$ polynomial. 

If $h(\sfx)\in\msK[\sfx]$ is a polynomial over $\msK$, $h(\sfx)\msK[\sfx]\subset\msK[\sfx]$ is an ideal of $\msK[\sfx]$ and the
quotient ring $\msK[\sfx]/h(\sfx)\msK[\sfx]$ is then defined. $h(\sfx)\msK[\sfx]$ consists of the polynomials over $\msK$
divisible by $h(\sfx)$. When $h(\sfx)$ is monic of degree $d>0$, the polynomial division algorithm can be used for the division
by $h(\sfx)$ and a description of $\msK[\sfx]/h(\sfx)\msK[\sfx]$ is possible: every element
$a\in\msK[\sfx]/h(\sfx)\msK[\sfx]$ has a unique expression of the form 
\begin{equation}
\label{alp0}
a=\mycom{{}_\sss}{{}_{0\leq i\leq d-1}}a_i\sfx^i+h(\sfx)\msK[\sfx]
\end{equation}
with $a_i\in\msK$ for $0\leq i\leq d-1$. \ceqref{alp0} has an interesting interpretation. 
As $\msK$ can be identified with its image in $\msK[\sfx]/h(\sfx)\msK[\sfx]$ under the quotient map,
every polynomial over $\msK$ can be regarded as one over $\msK[\sfx]/h(\sfx)\msK[\sfx]$ with the same coefficients.
Let $\xi=\sfx+h(\sfx)\msK[\sfx]\in\msK[\sfx]/h(\sfx)\msK[\sfx]$. Then, $h(\xi)=0$, that is $\xi$ is a root
of $h(\sfx)$ in $\msK[\sfx]/h(\sfx)\msK[\sfx]$. Further, \ceqref{alp0} can be written uniquely as
\begin{equation}
\label{alp}
a=\mycom{{}_\sss}{{}_{0\leq i\leq d-1}}a_i\xi^i.
\end{equation}
$\msK[\sfx]/h(\sfx)\msK[\sfx]$ is called the residue class ring of $\msK[\sfx]$ modulo $h(\sfx)$. 
By what said above, $\msK[\sfx]/h(\sfx)\msK[\sfx]$ can be regarded as the extension ring
$\msK[\xi]$ of $\msK$ by the root $\xi$.



We are now ready to survey the theory of Galois fields and rings. We consider first Galois fields and then shift
to Galois rings. At this point, it is appropriate to recall that any finite ring $\msK$ is characterized by
two integer parameters: its characteristic, the smallest positive integer $n$ such that $n1=0$,
and its cardinality, the number of its elements.


\begin{defi} \label{def:gr1}
A Galois field is a finite field $\msF$. 
\end{defi}

\noindent
The following theorem provides a classification of Galois fields up to isomorphism.

\begin{theor} \label{theo:gr1}
Let $\msF$ be a Galois field. Then, there exist positive integers $p,d$ with $p\geq 2$ prime
such that the characteristic and cardinality of $\msF$ equal respectively $p$ and $p^d$.
Conversely, for any two positive integers $p,d$ with $p\geq 2$ prime there exists a Galois field
of characteristic $p$ and cardinality $p^d$.
Two Galois fields with the same characteristic and cardinality are isomorphic. 
\end{theor}

\noindent
One usually writes $\msF=\GF(p^d)$ or $\msF=\bbF_{p^d}$ to emphasize the isomorphy class which $\msF$ belongs to.
$d$ is called the degree of $\msF$.

\begin{theor} \label{theo:gr2}
A Galois field $\msF=\bbF_{p^d}$ contains precisely one Galois field $\msF'=\bbF_{p^{d'}}$ as a subfield
for each degree $d'$ dividing $d$. 
\end{theor}

\noindent
In particular, $\msF$ always contains a minimal subfield, the prime subfield
$\msP_\msF=\bbF_p$. $\msP_\msF$ is the subfield generated by the multiplicative unity $1$, so that 
$\msP_\msF\simeq\bbZ_p$. The field $\msF$ itself is an extension field of $\msP_\msF$.


\begin{theor} \label{theo:gr3}
Let $\msF=\bbF_{p^d}$ be a Galois field and $\msP_\msF=\bbF_p$ be its prime field. Then,
$\msF\simeq\msP_\msF[\sfx]/h(\sfx)\msP_\msF[\sfx]$ for any irreducible monic 
polynomial $h(\sfx)\in\msP_\msF[\sfx]$ of degree $d$. 
\end{theor}

\noindent
As discussed above, every element $a$ of $\msF$ can then be expressed uniquely either as the residue class
of a polynomial over $\msP_\msF$ of degree at most $d-1$ in the quotient ring $\msP_\msF[\sfx]/h(\sfx)\msP_\msF[\sfx]$,
as in \ceqref{alp0}, or as the evaluation at the root $\xi=\sfx+h(\sfx)\msP_\msF[\sfx]$ of $h(\sfx)$ of a polynomial over
$\msP_\msF$ of degree at most $d-1$, as in \ceqref{alp}. $\msF$, so, is an extension field $\msP_\msF[\xi]$ of
$\msP_\msF$ by $\xi$. When $d=1$, we recover the field $\msP_\msF$.

A Galois field $\msF=\bbF_{p^d}$ is a vector space over its prime field $\msP_\msF=\bbF_p$ of dimension $d$. When $\msF$
is regarded as the extension field $\msP_\msF[\xi]$ of $\msP_\msF$ by the root $\xi$ of an irreducible monic polynomial
$h(\sfx)\in\msP_\msF[\sfx]$ of degree $d$ as above, the powers $\xi^i$, $0\leq i\leq d-1$,
constitute a basis of $\msF$ over $\msP_\msF$, as shown by \ceqref{alp}.

An element $\theta\in\msF$ of a Galois field $\msF=\bbF_{p^d}$ is said to be primitive if it holds that 
$\msF=\{0\}\cup\{\theta^i\hfpt|\hfpt 0\leq i\leq p^d-2\}$. An irreducible polynomial $h(\sfx)\in\msP_\msF[\sfx]$ of
degree $d$ is called primitive if the root $\xi=\sfx+h(\sfx)\msP_\msF[\sfx]\in\msF$ of $h(\sfx)$ is primitive.

The group of units $\msF^\times$ of a Galois field $\msF=\bbF_{p^d}$ is generated by any primitive element
$\theta\in\msF$. Therefore, $\msF^\times=\msC_{p^d-1}$, the cyclic group of order $p^d-1$. 


The theory of Galois rings is a natural generalization of that of Galois fields.

\begin{defi} \label{def:gr2}
A Galois ring is a finite commutative ring $\msR$ with unity whose zero divisors
constitute a principal ideal of the form $p\msR$ for some prime number $p\geq 2$. 
\end{defi}

\noindent
The following theorem classifies Galois rings up to isomorphism.

\begin{theor} \label{theo:gr4}
Let $\msR$ be a Galois ring. Then, there are positive integers $p,r,d$ with $p\geq2$ prime
such that the characteristic and cardinality of $\msR$ equal respectively $p^r$ and $p^{rd}$.
Conversely, for any three positive integers $p,r,d$ with $p\geq 2$ prime there exists a Galois ring
of characteristic $p^r$ and cardinality $p^{rd}$.
Two Galois rings with the same characteristic and cardinality are isomorphic. 
\end{theor} 

\noindent
One customarily writes $\msR=\GR(p^r,d)$ to emphasize the isomorphy class which $\msR$ belongs to.
$d$ is called the degree of $\msR$. 

\begin{theor} \label{theo:gr5}
Let $\msR=\GR(p^r,d)$ be a Galois ring. Then, the principal ideals $p^i\msR$ with $0\leq i\leq r$ are linearly ordered by
inclusion as 
\begin{equation}
\label{chainring}
\{0\}=p^r\msR\subset p^{r-1}\msR\subset\ldots\subset p\msR\subset p^0\msR=\msR.
\end{equation}
For each $i$, $|p^i\msR|=p^{(r-i)d}$. Further, $p\msR$ is a maximal ideal, is the only maximal ideal of $\msR$
and consists precisely of the zero divisors of $\msR$. 
\end{theor} 

\noindent
The above theorem has several implications. First, by \ceqref{chainring}, 
$\msR$ is chain ring of length $r$. Second, $\msR$ is local, since $p\msR$ is its only maximal ideal.
As a consequence, the only idempotents of $\msR$ are $0$, $1$. Third, the zero divisors of $\msR$ are
all nilpotent since they form the ideal $p\msR$ and $(p\msR)^r\subset p^r\msR=\{0\}$.
Fourth, since the ideal $p\msR$ is maximal, the quotient ring $\msF_\msR=\msR/p\msR$ is a field
with $|\msF_\msR|=p^d$ so that $\msF_\msR=\bbF_{p^d}$. $\msF_\msR$ is called the residue field of $\msR$.
When $r=1$, $\msR=\msF_\msR$ and $\msR$ reduces to a Galois field.

\begin{theor} \label{theo:gr6}
A Galois ring $\msR=\GR(p^r,d)$ contains as a subring precisely one Galois ring $\msR'=\GR(p^r,d')$ 
for each degree $d'$ dividing $d$. 
\end{theor}

\noindent
In particular, $\msR$ always contains a minimal subring, the prime subring
$\msP_\msR=\GR(p^r,1)$. $\msP_\msR$ is the subring generated by the multiplicative unity $1$,
so that $\msP_\msR\simeq\bbZ_{p^r}$. The ring $\msR$ itself is an extension ring of $\msP_\msR$.

Let $\msR=\GR(p^r,d)$ be a Galois ring. Then, modulo $p$ reduction induces a canonical projection ring morphism
$\varpi:\msP_\msR\rightarrow\msP_{\msF_\msR}$. $\varpi$ extends naturally to a polynomial projection
ring morphism $\varpi:\msP_\msR[\sfx]\rightarrow\msP_{\msF_\msR}[\sfx]$. 

Let $g(\sfx)\in\msP_\msR[\sfx]$ be a monic polynomial of degree $d\geq 1$. Then, $\varpi(g(\sfx))\in\msP_{\msF_\msR}[\sfx]$. 
If $\varpi(g(\sfx))$ is irreducible, then $g(\sfx)$ is said to be basic irreducible monic. If $\varpi(g(\sfx))$ is primitive,
then $g(\sfx)$ is said to be basic primitive monic. Basic primitive monic polynomials 
are also basic irreducible monic.

\begin{theor} \label{theo:gr}
Let $\msR=\GR(p^r,d)$ be a Galois ring and $\msP_\msR=\GR(p^r,1)$ be its prime ring.
Then, one has $\msR\simeq\msP_\msR[\sfx]/h(\sfx)\msP_\msR[\sfx]$ for every basic irreducible monic 
polynomial $h(\sfx)\in\msP_\msR[\sfx]$ of degree $d$. 
\end{theor}

\noindent
As for a Galois field, each element $a$ of $\msR$ can then be expressed uniquely either as the residue class
of a polynomial over $\msP_\msR$ of degree at most $d-1$ in the quotient ring $\msP_\msR[\sfx]/h(\sfx)\msP_\msR[\sfx]$,
as in \ceqref{alp0}, or as the evaluation at the root $\xi=\sfx+h(\sfx)\msP_\msR[\sfx]$ of $h(\sfx)$ of a polynomial over
$\msP_\msR$ of degree at most $d-1$, as in \ceqref{alp}. $\msR$, so, is an extension ring $\msP_\msR[\xi]$ of
$\msP_\msR$ by $\xi$. When $d=1$, we recover the ring $\msP_\msR$.

A Galois ring $\msR=\GR(p^r,d)$ is a free module over its prime ring $\msP_\msR=\GR(p^r,1)$.
The ring $\msP_\msR$ enjoys the so called invariant dimension property. Consequently, any free basis of $\msR$
over $\msP_\msR$ will consist of the same number elements, which is $d$ in the present case. When $\msR$
is viewed as the extension ring $\msP_\msR[\xi]$ of $\msP_\msR$ by the root $\xi$ of an irreducible monic
polynomial $h(\sfx)\in\msP_\msR[\sfx]$ of degree $d$ as in the previous paragraph, 
the powers $\xi^i$, $0\leq i\leq d-1$, constitute one such basis,
as shown by \ceqref{alp}. 

If $\msR=\GR(p^r,d)$ is Galois ring with prime ring $\msP_\msR=\GR(p^r,1)$, there exists a nonzero
element $\theta\in\msR$ of order $p^d-1$, which is a root of a unique  basic primitive monic
polynomial $h(\sfx)\in\msP_\msR[\sfx]$ of degree $d$
and dividing $\sfx^{p^d-1}-1$ in $\msP_\msR[\sfx]$ with the following property. Every element $a\in\msR$ has a unique
$p$-adic representation %
\begin{equation} 
\label{pad}
a=\mycom{{}_\sss}{{}_{0\leq \alpha\leq r-1}}a_\alpha p^\alpha
\end{equation}
\noindent with $a_\alpha\in\msT_\msR$ for $0\leq \alpha \leq r-1$, where $\msT_\msR$ is the Teichm\"uller representative set
\begin{equation}
\label{tei}
\msT_\msR=\{0\}\cup\{\theta^i\hfpt|\hfpt 0\leq i\leq p^d-2\}.
\end{equation}
$a$ is a unit (respectively a zero divisor) of $\msR$ exactly when $a_0\neq0$ (respectively $a_0=0$). 
Further, $\varpi(\theta)$ is a primitive element of $\msF_\msR$ and thus $\varpi(\msT_\msR)=\msF_R$.

The groups of units $\msR^\times$ of a Galois ring $\msR$ can be factorized in the direct product $\msC_\msR\times\msD_\msR$,
where $\msC_\msR=\msC_{p^d-1}$ is the cyclic group of order $p^d-1$ and $\msD_\msR=\{1+u\hfpt|\hfpt u\in p\msR\}$ is the
residue class of $1$ in the residue field $\msF_\msR$ of $\msR$. Thus, $\msR^\times$ has cardinality $(p^d-1)p^{(r-1)d}$.

\subsection{\textcolor{blue}{\sffamily Frobenius automorphism and trace map}}\label{subsec:galoistrace}

In this appendix, we precisely define and study the main properties of the ring trace map
of a Galois ring. 

Let $\msR=\GR(p^r,d)$ be a Galois ring and $\msP_\msR=\GR(p^r,1)$ be the prime subring of $\msR$.
Let $\Aut(\msR/\msP_\msR)$ be the group of automorphisms of the ring $\msR$ which fix $\msP_\msR$.
It turns out that the subring of $\msR$ fixed by all elements of $\Aut(\msR/\msP_\msR)$ is precisely $\msP_\msR$
and not a proper extension of it contained in $\msR$. Then, $\msR$ is said to be a Galois extension of $\msP_\msR$,
and $\Aut(\msR/\msP_\msR)$ is called the Galois group of the extension and is commonly
denoted as $\Gal(\msR/\msP_\msR)$.

\begin{defi} \label{def:rt1}
The ring Frobenius map $\phi_\msR:\msR\rightarrow\msR$ of $\msR$ is defined by  
\begin{equation}
\label{fro}
\phi_\msR(a)=\mycom{{}_\sss}{{}_{0\leq \alpha \leq r-1}}a_\alpha{}^pp^\alpha
\end{equation}
for $a\in\msR$, where $a$ is expressed through the $p$-adic representation \ceqref{pad}.
\end{defi}

\noindent
When $r=1$ and $\msR$ reduces to its residue field $\msF_\msR=\bbF_{p^d}$, $\phi_\msR$ reproduces the usual
field Frobenius map of $\msF_\msR$ given by $\phi_{\msF_\msR}(a)=a^p$ for $a\in\msF_\msR$.  
The following theorem holds. 

\begin{theor} \label{prop:rt1}
$\Gal(\msR/\msP_\msR)\simeq\msC_d$ is a cyclic group of order $d$. Further, 
$\phi_\msR\in\Gal(\msR/\msP_\msR)$, has order $d$ and generates $\Gal(\msR/\msP_\msR)$. 
\end{theor}

The ring trace map of $\msR$ is defined in terms of $\phi_\msR$. 

\begin{defi}  \label{def:rt2}
The ring trace map $\tr_\msR:\msR\rightarrow\msR$ is given by 
\begin{equation}
\label{gentra}
\tr_\msR(a)=\mycom{{}_\sss}{{}_{0\leq i\leq d-1}}\phi_\msR{}^i(a)
\end{equation}
for $a\in\msR$
\end{defi}

\noindent
Again, when $r=1$ and $\msR$ reduces to its residue field $\msF_\msR$,
the ring trace map $\tr_\msR$ reduces to the field trace map $\tr_{\msF_\msR}:\msF_\msR\rightarrow\msP_{\msF_\msR}$
given by $\tr_\msR(a)=\sum_{0\leq i\leq d-1}a^{p^i}$. 
The main properties of the ring trace map $\tr_\msR$ are summarized by the following theorem. 

\begin{theor} \label{prop:rt2}
The ring trace map $tr_\msR$ maps $\msR$ onto $\msP_\msR\subset\msR$. Regarding $\msR$ and $\msP_\msR$ as $\msP_\msR$ modules,
$\tr_\msR$ is a module morphism. Further, $\tr_\msR(\phi_\msR(a))=\phi_\msR(\tr_\msR(a))=\tr_\msR(a)$ for all $a\in\msR$.
\end{theor}

The following technical proposition describes in a rather detailed manner the values the ring trace map can take
\ccite{Horadam:2005xqf,Horadam:2006xqf}.

\begin{prop} \label{prop:rt3}
For $0\leq j\leq r$, let $\msS_\msR{}^{(j)}=p^j\msR\setminus p^{j+1}\msR\subset\msR$ where 
$p^{r+1}\msR=\emptyset$ by convention and $\msB_\msR^{(j)}=\{p^jn1\hfpt|\hfpt 0\leq n\leq p^{r-j}-1\}\subset\msP_\msR$.
Then, for  every $0\leq j\leq r$ and $a\in\msS_\msR{}^{(j)}$ the mapping $x\in\msR\mapsto\tr_\msR(ax)\in\msP_\msR$
is onto $\msB_\msR^{(j)}$ and each value in $\msB_\msR^{(j)}$ is taken exactly $p^{r(d-1)+j}$ times.
\end{prop}

\noindent In particular, if $\tr_\msR(ax)=0$ for all $x\in\msR$, then $a=0$ since the only set $\msB_\msR^{(j)}$
containing only $0$ is $\msB_\msR^{(r)}$ and so $a\in\msS_\msR{}^{(r)}=\{0\}$.

\vfil\eject

\vspace{.5cm}

\noindent
\markright{\textcolor{blue}{\sffamily Acknowledgements}}

\noindent
\textcolor{blue}{\sffamily Acknowledgements.} 
The author acknowledges financial support from INFN Research Agency
under the provisions of the agreement between Alma Mater Studiorum University of Bologna and INFN. 
He is grateful to the organizers of the Conference PAFT24 - Quantum Gravity and Information, where part
of this work was done, for hospitality and support. Most of the algebraic calculations presented in this paper
have been carried out employing the WolframAlpha computational platform.

\vspace{.5cm}

\noindent
\textcolor{blue}{\sffamily Conflict of interest statement.}
The author declares that the results of the current study do not involve any conflict of interest.

\vspace{.5cm}

\noindent
\textcolor{blue}{\sffamily Data availability statement.}
The data that support the findings of this study are openly available on scientific journals or in the arxiv repository.

\vfill\eject

\noindent
\textcolor{blue}{\bf\sffamily References}

\end{document}